\documentclass[11pt]{article}

\usepackage{arxiv}

\usepackage{graphicx} 
\usepackage{parskip} 

\usepackage[font=small]{caption}

\usepackage{algorithm}
\usepackage{algorithmic}
\usepackage{amssymb} 
\usepackage{amsmath}
\usepackage{booktabs}
\usepackage[pdfencoding=auto, psdextra]{hyperref}
\usepackage{amsthm}
\usepackage{thm-restate}
\newtheorem{theorem}{Theorem}

\newtheorem{definition}{Definition}
\newtheorem{Lemma}{Lemma}

\usepackage{tikz}
\usetikzlibrary{decorations.pathmorphing,patterns}
\usetikzlibrary{shadows}
\usetikzlibrary{shapes}
\usetikzlibrary{decorations}
\usetikzlibrary{arrows,decorations.markings} 
\usepackage{csvsimple}
\usepackage{color, colortbl}
\usepackage{array}

\usepackage{verbatim}
\usepackage{parskip}

\usepackage{float} 
\usepackage{multirow}
\usepackage{caption}
\usepackage{subcaption}
\usepackage{graphicx}

\usepackage{mathtools}
\usepackage{tikz}
\usetikzlibrary{shapes.geometric}
\usepackage{tikz-3dplot}
\usetikzlibrary{arrows,calc,backgrounds}
\usetikzlibrary{decorations.markings}
\usetikzlibrary{backgrounds}
\usepackage{tikz-3dplot}
\usepackage{bbm}
\usetikzlibrary{decorations.pathmorphing,patterns}
\usetikzlibrary{shadows}
\usetikzlibrary{shapes}
\usetikzlibrary{decorations}
\usetikzlibrary{arrows,decorations.markings} 
\usepackage{csvsimple}

\newcount\Comments  
\newcommand{\kibitz}[2]{\ifnum\Comments=1{\color{#1}{#2}}\fi}

\title{Characterizing Nash Equilibria in Zero-Sum Games: A Physics-Inspired, Parallelizable Approach with a Linear Number of Gradient Queries}

\author{
 Taemin Kim \\
  Industrial and Systems Engineering\\
  Rensselaer Polytechnic Institute\\
  \texttt{kimt12@rpi.edu}\AND James P. Bailey \\
  Industrial and Systems Engineering\\
  Rensselaer Polytechnic Institute\\
  \texttt{bailej6@rpi.edu} 
  }

\date{}
\begin{document}

\maketitle

\begin{abstract}
We study online optimization methods for zero-sum games, a fundamental problem in adversarial learning in machine learning, economics, and many other domains.
Traditional methods approximate Nash equilibria (NE) using either regret-based methods (time-average convergence) or contraction-map-based methods (last-iterate convergence).
We propose a new method based on Hamiltonian dynamics in physics and prove that it can characterize the set of NE in a finite (linear) number of iterations of alternating gradient descent in the unbounded setting, modulo degeneracy, a first in online optimization.
Unlike standard methods for computing NE, our proposed approach can be parallelized and works with arbitrary learning rates, both firsts in algorithmic game theory. 
Experimentally, we support our results by showing our approach drastically outperforms standard methods.  
\end{abstract}



\section{Introduction}\label{sec:intro}
Zero-sum games are a fundamental concept in game theory, where the gain of one agent corresponds to the loss of another, resulting in a total net payoff of zero. 
They are typically modeled using matrices that represent agents' payoffs for different strategies. 
Zero-sum games are fundamentally based on the minimax theorem, where agent 1 aims to maximize its gain, while agent 2 tries to minimize agent 1's gain.
\begin{equation}
        \max_{x \in \mathbb{R}^{k}}\min_{y \in \mathbb{R}^{k}} \left\langle x,Ay\right\rangle 
    \tag{Minimax Optimization}\label{Minimax}
\end{equation}
Von Neumann's proof of the \ref{MinimaxThm} became the foundation of game theory \cite{vonNeumann28}. 
It shows that in a two-person zero-sum game, the optimal utility of \ref{Minimax} is the same for both agents, given that they both play optimally \cite{vonNeumann28, roughgarden2016twenty}\footnote{Von Neumann's proof was for compact strategy spaces, but the results naturally extend to the unbounded setting if $\max$ and $\min$ are well-defined, i.e., if an optimal solution exists.}.
\begin{equation}
        \max_{x \in \mathbb{R}^{k}}\min_{y \in \mathbb{R}^{k}} \left\langle x,Ay\right\rangle = \min_{y \in \mathbb{R}^{k}}\max_{x \in \mathbb{R}^{k}} \left\langle x,Ay\right\rangle 
    \tag{Minimax Theorem}\label{MinimaxThm}
\end{equation}
Zero-sum games have provided insight across various fields, including economics, computer science, machine learning, and finance.
In economics, they are used in bargaining and auction theory to analyze the dynamics of entities competing for limited resources in competitive markets \cite{pu2020online}. 
Zero-sum games also give a foundational framework for modeling artificial intelligence systems that adopt adversarial learning, such as generative adversarial networks (GANs) \cite{goodfellow2014generative,daskalakis2017training}. 
They have been applied in various fields such as image processing and computer vision to improve resolution or generate images that look similar to those of the real world \cite{li2016precomputed, zhou2020image_resolution, brock2018large}. 
In healthcare, they are used to generate images that support disease diagnosis, for example, by segmenting brain tumor images \cite{cirillo2020brain_tumor}, improving the accuracy of synthesized MRI images \cite{zhao2020bayesian}, and aligning different 3D images into the same coordinate frame \cite{zhang2020deform}.

Although it was proved in \cite{adler2013equivalence} that a \ref{Minimax} problem can be solved via linear programming, iterative first-order methods are often preferred. 
For example, \cite{nesterov2005excessive} used first-order methods for solving non-smooth convex problems,
\cite{nemirovski2004prox} implemented prox-type methods, and, more recently, \cite{mokhtari2020convergence} applied optimistic gradient descent (OGDA) and extragradient (EG) methods for saddle point problems, achieving the convergence rate $O(1/T)$. 
The algorithmic ideas of nonlinear programming used in \cite{nemirovski2004prox} have been reinterpreted in online optimization frameworks, particularly no-regret learning \cite{daskalakis2011near}, which are commonly used due to their simplicity of application and robustness of the solution \cite{daskalakis2019last}.
In these approaches, the goal is to find Nash equilibria (NE) by playing games repeatedly while iteratively updating strategies based on past decisions and the opponent's choices. 
Although most of the works assume fixed payoff matrices, the analysis has been extended to time-varying payoff matrices as well \cite{duvocelle2023multiagent}.
Two common methods to evaluate convergence to NE are time-average convergence and last-iterate convergence.
These methods assess how strategies converge to NE over time in specific algorithms.

\textbf{Time-average convergence} refers to the process of showing that the time-average of the strategies, $(x^t,y^t)_{t=1}^\infty$, converges to the set of NE. 
Informally, time-average convergence implies $\lim_{T\to \infty} \frac{1}{T} \sum_{t=0}^{T} (x^t, y^t) \rightarrow (x^*, y^*)$ where $(x^*, y^*)$ is a NE. 
The standard method of proving time-average convergence requires showing the \textit{regret} of the strategies grows sublinearly (see, e.g. \cite{cesa2006prediction,roughgarden2016twenty}). 
Regret measures the difference between the cumulative loss of the chosen actions and the cumulative loss of the best action.
The two most well-known algorithms that achieve time-average convergence are the multiplicative weights update algorithm (MWU) and the (simultaneous) gradient descent update algorithm (GD). 
Both of these algorithms achieve $O(1/\sqrt{T})$ time-average convergence when the strategy space is compact and learning rates decay at rate $O(1/\sqrt{T})$ \cite{cesa2006prediction}.
More recently, \cite{Bailey20Regret} shows that GD in 2-agent, 2-strategy games achieves $O(1/\sqrt{T})$ time-average convergence to the set of NE with arbitrary learning rates. 

Perhaps surprisingly, it has recently been shown that MWU, GD, and the more general class of follow-the-regularized-leader (FTRL) algorithms result in strategies that diverge from the set of NE \cite{Bailey18Divergence}. 
In contrast, methods like optimistic and alternating gradient descent either converge to or cycle around the set of Nash equilibria, respectively. 
Both of these algorithms have been shown to have smoothness properties that result in faster ($O(1/T)$) time-average convergence to the set of NE using sufficiently small fixed learning rates \cite{bailey2021left,mokhtari2020convergence}.

\textbf{Last-iterate convergence} means the most recent strategies approach NE, informally implying $\lim_{T\to \infty} (x^t, y^t) \rightarrow (x^*, y^*)$ where $(x^*, y^*)$ is a NE. 
Standard proof techniques rely on the establishment of a contraction map.
Contraction maps rely on showing that the update rules contract to a single attractor -- specifically, the NE. 
The two most well-known algorithms that guarantee last-iterate convergence are optimistic gradient descent (OGD) and extragradient descent (EG). 
It is shown in \cite{golowich2020iterate} that last-iterate strategies of EG converge to NE at a rate $O(1/\sqrt{T})$, which is slower than the time-average convergence rate of $O(1/T)$.
This gap in convergence rate drives the need for new algorithms with faster last-iterate convergence.
This is achieved in \cite{yoon2021accelerated} by introducing accelerated EG in unconstrained settings. 
Further improvement is made in \cite{pmlr-v235-cai24f}, which extends accelerated EG to constrained spaces. 
Another algorithm is Hamiltonian GD in \cite{abernethy2021last}, which achieves linear last-iterate convergence while traditional GD does not guarantee time-average convergence. 
In general, last-iterate algorithms are more suitable in real-time online learning \cite{daskalakis2019last}. 
For example, \cite{daskalakis2018limit} showed that OGD has a larger set of stable limit points than GD, while avoiding unstable limit points and ensuring a more consistent convergence to NE.

\subsection{Our Contributions}
We propose a third method for approximating NE in zero-sum and coordination games. 
Our method is based on a recent connection between Hamiltonian dynamics in physics and online optimization in games \cite{Bailey19Hamiltonian}.
Specifically, we leverage the physics of online optimization to show that each iteration of alternating gradient descent can be used to yield a linear equation describing the set of NE.
This means that, modulo degeneracy, we can precisely \underline{characterize} the set of NE in a finite number of iterations of alternating gradient descent, corresponding to a first in online optimization in games.
We present three different models that rely on different types of information from each agent to extend the possible applications of our approach. 
Each model yields a linear system in which the NE is an unknown variable. 
Unlike existing methods that require alternating gradient updates until time-average or last-iterate strategies converge to NE, our approach requires a finite number of updates to generate the required number of linear equations.
We remark that we do not claim to solve for the NE in a finite number of iterations --- rather, we generate a system of equations $A'(x^*,y^*)=b'$ that can be used to solve for the NE using feedback from alternating gradient descent. 

We also introduce a parallelization technique that allows the use of arbitrary learning rates, challenging the previously held belief that small learning rates are necessary for learning NE. 
Experimentally, we show each method drastically outperforms existing methods in three categories: number of training points needed, the empirical run time (including the additional time to solve $A'(x^*,y^*)=b'$), and accuracy.
We remark that our results extend to polymatrix zero-sum games through standard transformations, e.g., \cite{cai2016zero,bailey2021left}.

\subsection{Outline of Paper}
In Section \ref{sec:preliminaries}, we introduce the basics of online optimization in games, including how learning dynamics in zero-sum games are related to Hamiltonian systems in physics. 
In Section \ref{sec:InvariantFunctions}, we show that the time-invariant function can be derived when strategies are updated using discrete-time alternating gradient descent in both zero-sum and coordination games.
In Section \ref{sec:ValidEquations}, we show that each iteration of alternating gradient descent can be used to generate a linear equation characterizing the set of NE --- this section includes three different models that rely on different types of information from each agent. 
In Section \ref{sec:Base_Experiments}, we construct linear systems from the derived linear equations to solve directly for the set of NE. 
We evaluate the performance of the models through experiments, measuring relative errors to assess how accurately the set of NE solutions can be obtained. 
In Section \ref{sec:Improved_Experiments_Parallelization}, we propose a parallelization method that addresses the numerical instability issues observed in Section \ref{sec:Base_Experiments}.

\section{Preliminaries} \label{sec:preliminaries}
In this paper, two-agent games are studied in which agent 1 and agent 2 try to maximize their utilities while interacting with each other. 
Both agents choose their strategies $x$ and $y$ from a strategy space $\mathbb{R}^{k}$ receiving utilities of $\left\langle x, Ay - b_1 \right\rangle$ and $\left\langle y, Bx - b_2 \right\rangle$, respectively. 
The terms $b_1$ and $b_2$ are the linear costs associated with the playing strategies $x$ and $y$, respectively. 
This results in the following game:
\begin{equation}
    \begin{split}
        \max_{x\in \mathbb{R}^{k}} \left\langle x, Ay - b_1 \right\rangle& \\
        \max_{y\in \mathbb{R}^{k}} \left\langle y, Bx - b_2 \right\rangle&.\hspace{.6in}
    \end{split}
    \tag{Two-Agent Games with Linear Costs}\label{eqn:Game}
\end{equation}

A point $(x^*, y^*)$ is a Nash equilibrium (NE) of a \ref{eqn:Game} where each agent cannot unilaterally improve its utility by deviating from it. 
NE satisfies the following conditions.
\begin{equation}
    \begin{split}
        \left\langle x^*, Ay^* - b_1 \right\rangle \geq \left\langle x, Ay^* - b_1 \right\rangle \quad \forall x \in \mathbb{R}^{k} \\
        \left\langle y^*, Bx^* - b_2 \right\rangle \geq \left\langle y, Bx^* - b_2 \right\rangle \quad \forall y \in \mathbb{R}^{k}
    \end{split}
    \tag{NE }\label{eqn:NE}
\end{equation}

The gradient of each agent's utility must be $0$ for $x^*$ and $y^*$ to be NE since the strategy space is unbounded, i.e., $Ay^* - b_1= 0$,  $Bx^* - b_2=0$. 
Otherwise, if these gradients are not zero, each agent can move their strategies in a beneficial direction to increase their utility, given that the other agent's strategies remain fixed.
Thus, an equivalent definition of NE is 
\begin{equation}
    \begin{split}
        Ay^* - b_1 = 0\phantom{.} \\
        Bx^* - b_2 = 0.
    \end{split}
    \tag{NE Conditions in Unbounded Setting}\label{eqn:NEconditions}
\end{equation}

Payoffs due to agents' interactions are captured by the payoff matrices $A$ and $B$. 
We say that a game is \textbf{zero-sum} if the net sum of payoffs due to interaction is zero, i.e., $\langle x, Ay\rangle + \langle y, Bx\rangle=0$ for all $x$ and $y$, i.e., if $B=-A^\intercal$. 
A game is a \textbf{coordination} game if $B=A^\intercal$, i.e., if both agents receive the same payoff due to interaction. 
We remark that learning algorithms for \textbf{general-sum} games (no restrictions on $A$,$B$), e.g., \cite{daskalakis2021nearoptimal}, are not studied in this paper.

\subsection{Unbounded Strategy Space vs. Probability Vectors}
\label{sec:UnboundedMotivation}

We remark that while many applications of games restrict agent strategies to probability vectors, i.e., $x\in \Delta^k:=\left\{x\in \mathbb{R}^k_{\geq 0}: \sum_{i=1}^k x_i=1\right\}$, the unbounded setting with linear costs is a powerful tool for motivating learning algorithms. 
In the setting with probability vectors, games are commonly studied without linear costs, i.e., 
\begin{equation}
    \begin{split}
        \max_{x\in \Delta^k} \left\langle x, Ay \right\rangle& \\
        \max_{y\in \Delta^k} \left\langle y, Bx  \right\rangle&.
    \end{split}
    \tag{Games with Probability Vectors}\label{eqn:ProbGame}
\end{equation}

In this section, we briefly discuss how understanding learning in \ref{eqn:Game} provides insight into \ref{eqn:ProbGame}. 

First, observe that the linear costs are recovered from \ref{eqn:ProbGame} after performing a variable substitution to remove the constraint $\sum_{i=1}^k x_i = 1$. 
Let $\tilde{x}^\intercal = (x_1, x_2, \dots, x_{k-1})^\intercal \in \mathbb{R}^{k-1}$. 
We can then express $x\in \Delta^k$ as $x^\intercal=(\tilde x, x_k = 1 - \left\langle \mathbbm{1}, \tilde{x}\right\rangle)^\intercal$ where $\mathbbm{1}$ is a $(k-1)$-dimensional column vector of 1's. 
Performing this variable substitution results in the new strategy space $\tilde \Delta^{k-1} =\{\tilde x\in \mathbb{R}^{k-1}_{\geq 0}: \left\langle \mathbbm{1}, \tilde{x}\right\rangle \leq 1\}$. 
Similarly, for agent 2, $y\in \Delta^k$ is expressed as $y^\intercal=(\tilde y, 1-\left\langle \mathbbm{1}, \tilde{y}\right\rangle)^\intercal$.
Applying both substitutions to agent 1's objective function in \ref{eqn:ProbGame} yields the equivalent objective
\begin{align*}
    \max_{x\in \Delta^k} \left\langle x, Ay \right\rangle& = \max_{\tilde x \in \tilde \Delta^{k-1}} \left\langle \begin{bmatrix}
    \tilde{x} \\
    1 - \mathbbm{1}^{\intercal} \tilde{x}
    \end{bmatrix},  A \begin{bmatrix}
    \tilde{y} \\
    1 - \mathbbm{1}^{\intercal} \tilde{y}
    \end{bmatrix}\right\rangle\\
     &=\max_{\tilde x\in \tilde \Delta^{k-1}} \left\langle\tilde{x}, (A_{[k-1],[k-1] }  - A_{[k-1],k}\mathbbm{1}^\intercal - \mathbbm{1} A_{k,[k-1]} + \mathbbm{1} A_{kk} \mathbbm{1}^\intercal)\tilde y\right\rangle  \\
     &\phantom{= \max_{\tilde x\in \tilde \Delta^{k-1}}} - \left\langle\tilde x, \mathbbm{1} A_{kk} - A_{[k-1],k}\right\rangle\\
     &\phantom{= \max_{\tilde x\in \tilde \Delta^{k-1}}} + \left\langle A_{k,[k-1]}^\intercal - \mathbbm{1}A_{kk} , \tilde y\right\rangle + A_{kk}
\end{align*}
where $A_{[k-1],[k-1]}$ is the first $k-1$ rows and columns of $A$, 
where $A_{[k-1],k}$ is the first $k-1$ rows and last column of $A$,  $A_{k, [k-1]}$ is the last  row and first $k-1$  columns of $A$, and $A_{kk}$ is the last row and last column of $A$. 
Observe that the last two terms (the last line) are independent of agent $1$'s decision, and therefore it suffices for agent $1$ to solve
\begin{align*}
    \max_{\tilde x \in \tilde \Delta^{k-1}} \langle \tilde x,  \tilde A\tilde y - \tilde b_1\rangle
\end{align*}
where $\tilde A= A_{[k-1],[k-1] }  - A_{[k-1],k}\mathbbm{1}^\intercal - \mathbbm{1} A_{k,[k-1]} + \mathbbm{1} A_{kk} \mathbbm{1}^\intercal$ and $\tilde b_1 =\mathbbm{1} A_{kk} - A_{[k-1],k}$. 
Performing the same process on agent 2's objective function yields 
\begin{align*}
    \max_{\tilde y \in \tilde \Delta^{k-1}} \langle \tilde y,  \tilde B\tilde x - \tilde b_2\rangle
\end{align*}
where $\tilde B$ and $\tilde b_2$ are defined analogously. 
We remark that it is straightforward to show that if $B=(\pm) A^\intercal $ then $\tilde B = (\pm)  \tilde A^\intercal $ (zero-sum and coordination games are preserved through this substitution). 
This recovers the linear costs given in \ref{eqn:Game}.

Finally, observe that the affine-hull of $\tilde \Delta^{k-1}$ is $\mathbb{R}^{k-1}$, i.e.,
\ref{eqn:Game} is simply \ref{eqn:ProbGame} where the domains of the strategy spaces are extended to their affine hulls.
In the event the optimizer of \ref{eqn:ProbGame} is in the relative interior, i.e., none of the constraints are binding, then it suffices to solve \ref{eqn:ProbGame} over the affine-hull, i.e., $(x^*,y^*)\in rel.int(\Delta^{k})\times rel.int(\Delta^{k})$ is an optimizer to \ref{eqn:ProbGame} if and only if it is also an optimizer to \ref{eqn:Game}. 
Thus, if there is an interior solution (also known as a fully-mixed NE) to \ref{eqn:ProbGame}, then it suffices to solve \ref{eqn:Game}. 
As a result, learning dynamics in the unconstrained space are a powerful tool for solving games with probability vectors, and we focus our analysis on unconstrained strategy spaces.

\subsection{Learning in Games}\label{Hamiltonian_and_zerosum}
In zero-sum games, when $A, B=-A^\intercal, b_1$, and $b_2$ are known, the set of NE can be computed by linear programming.
However, in most applications, we do not have full information \cite{cesa2006prediction}. 
Instead, in most settings, agents learn NE by playing a repeated game where they only observe their respective payoffs and corresponding gradients,  i.e., agent 1 observes $\langle x^t, Ay^t- b_1\rangle$ and $Ay^t- b_1$ where $x^t$ and $y^t$ correspond to agent 1's strategy and agent 2's strategy respectively at iteration $t$.
The agents then iteratively update their strategies by moving in a beneficial direction. 
As discussed in Section \ref{sec:intro}, if agents' update rules satisfy certain properties, agents' strategies may converge to the set of NE (last-iterate convergence), or the time-average of the strategies may converge to the set of NE (time-average convergence).

In this paper, we study the alternating gradient descent update rule, which is given by: 
\begin{equation}
    \begin{split}
        x^{t+1} &= x^{t} + \eta_1 (Ay^{t}-b_1) \\
        y^{t+1} &= y^{t} + \eta_2 (Bx^{t+1}-b_2)
    \end{split}
    \tag{Alternating GD} \label{eqn:AltGD}
\end{equation}
where $\eta_1$ and $\eta_2$ are the learning rates for each agent, respectively. 
\cite{bailey2021left} shows that \ref{eqn:AltGD} achieves $O(1/T)$ time-average convergence to the set of NE in zero-sum games when $\sqrt {\eta_1\cdot \eta_2} \leq 2/ ||A||$. 

In this paper, we propose a new method for computing NE in zero-sum games. 
Rather than relying on time-average or last-iterate convergence, we will utilize knowledge of the update rule to directly characterize the set of NE after a finite set of iterations. 
Our approach is motivated by a recently developed connection between \ref{eqn:AltGD} in zero-sum games and Hamiltonian dynamics in physics \cite{Bailey19Hamiltonian}.

\subsection{Motivation via Hamiltonian Dynamics}\label{sec:Hamiltonian}

To understand the connection between a discrete algorithm and continuous-time dynamics from physics, we first introduce a continuous version of \ref{eqn:AltGD} known as \ref{eqn:CGD}. 
\begin{equation}
        \begin{aligned}
           x(T) &= x(0) + \eta_1 \int_0^T (Ay(t)-b_1) dt \\
            y(T) &= y(0) + \eta_2 \int_0^T (-A^\intercal x(t)-b_2) dt \\
        \end{aligned}\tag{Continuous-time GD}\label{eqn:CGD}
\end{equation}

As depicted in Figure \ref{fig:spring}, \cite{Bailey19Hamiltonian} shows that \ref{eqn:CGD} is a special case of Hamiltonian dynamics. 

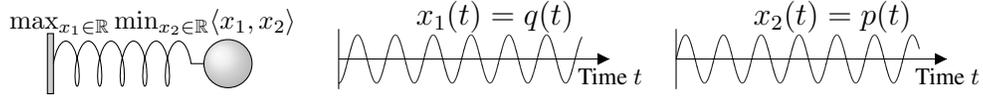
\begin{figure}[!ht]
	\centering
    \def\scaler{.8}
	\begin{tabular}{m{1.6in} m{1.6in} m{1.6in}}
		\begin{tikzpicture}[scale=\scaler]
		\node at (1.65,.7) {$\max_{x_1\in \mathbb{R}} \min_{x_2\in \mathbb{R}} \langle x_1, x_2\rangle$};		
		\draw[decoration={aspect=0.3, segment length=3mm, amplitude=3mm,coil},decorate] (0,0) -- (2.5,0); 
		\fill[white] (2.9,0) circle(.4);
		\shade[ball color = gray!40, opacity = 0.4] (2.9,0) circle (.4);
		\draw (2.9,0) circle (.4);
		\fill[gray, opacity=.4] (0,-.5)--(0,.5)--(-.1,.5)--(-.1,-.5)--cycle;
		\draw (0,-.5)--(0,.5)--(-.1,.5)--(-.1,-.5)--cycle;
		\node at (0,1.05) {};
		\end{tikzpicture}	
		&
		\begin{tikzpicture}[scale=\scaler]	
		\node at  (2.6,1.2) {\large $x_1(t)=q(t)$};
		\draw (0,1)--(0,0);
		\draw (0,.5)--(4.5,.5);
		\fill (4.5,.5)--(4.3,.4)--(4.3,.6);
		\node[below left] at (5.2,.5) {\small Time $t$};
		\csvreader[ head to column names,%
		late after head=\xdef\aold{\a}\xdef\bold{\b},%
		after line=\xdef\aold{\a}\xdef\bold{\b}]%
		{Images/Position.dat}{}{%
			\draw (\aold, \bold) -- (\a,\b);
		}	
		\end{tikzpicture}
		&
		\begin{tikzpicture}[scale=\scaler]	
		\node at  (2.6,1.2) {\large $x_2(t)=p(t)$};
		\draw (0,1)--(0,0);
		\draw (0,.5)--(4.5,.5);
		\fill (4.5,.5)--(4.3,.4)--(4.3,.6);
		\node[below left] at (5.2,.5) {\small Time $t$};
		\csvreader[ head to column names,%
		late after head=\xdef\aold{\a}\xdef\bold{\b},%
		after line=\xdef\aold{\a}\xdef\bold{\b}]%
		{Images/velocity.dat}{}{%
			\draw (\aold, \bold) -- (\a,\b);
		}
		\end{tikzpicture}
	\end{tabular}
	\caption{Let $(q(t),p(t))$ represent the position and momentum of a unit mass on a frictionless spring with spring constant $k=1$, and let $(x_1(t),x_2(t))$ represent agent strategies in the zero-sum game $\max_{x_1\in \mathbb{R}} \min_{x_2\in \mathbb{R}} \langle x_1, x_2 \rangle$ where agents update strategies with \ref{eqn:CGD} with learning rate $\eta=1$. \cite{Bailey19Hamiltonian} provides the framework to show $(q(t),p(t))=(x_1(t),x_2(t)) \ \forall t$ if $(q(0),p(0))=(x_1(0),x_2(0))$, i.e., \ref{eqn:CGD} is equivalent to a mass on a frictionless spring. } \label{fig:spring}
\end{figure}

A \textbf{Hamiltonian system} is defined by a set of coordinates $(p,q)$ and the Hamiltonian function $H(p,t)=H(p(t),q(t))$ where $p(t)$ is momentum, $q(t)$ is position at time $t$. 
$H(p,q)$ represents the total energy of the system that satisfies $\frac{dq}{\partial t} = \frac{\partial H}{\partial p}$ and $\frac{dp}{\partial t} = -\frac{\partial H}{\partial q}$ which are called Hamiltonian equations. 
Since $\frac{dH}{dt} = \frac{\partial H}{\partial p}\frac{dp}{dt} + \frac{\partial H}{dq} \frac{dq}{dt} = -\frac{\partial H}{\partial p}\frac{\partial H}{\partial q} + \frac{\partial H}{dq} \frac{\partial H}{\partial p} = 0$, the energy of the system is time-invariant. 
In the transformation given in \cite{Bailey19Hamiltonian}, the energy generated by \ref{eqn:CGD} corresponds to the distance to the set of NE \cite{Mertikopoulos2018CyclesAdverserial,Bailey19Hamiltonian}, i.e., the dynamics of \ref{eqn:CGD} cycle around the set of NE as depicted in Figure \ref{fig:HS_discretization}!

\begin{figure}[!ht]
    \centering
    \includegraphics[width=0.5\linewidth]{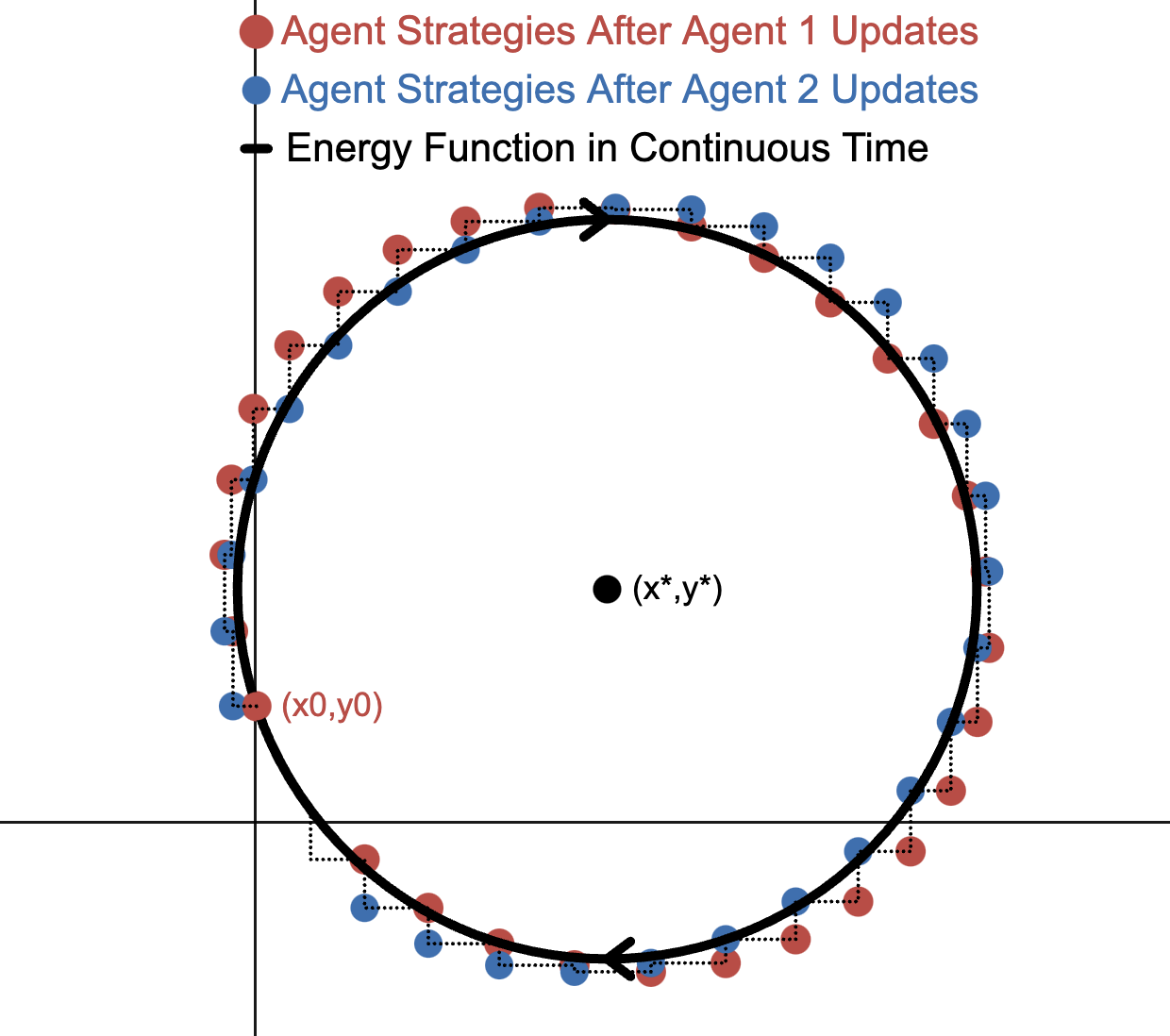}
    \caption{
    The \ref{eqn:CGD} dynamics cycle around the set of NE on a closed orbit. 
    We show that \ref{eqn:AltGD} closely approximates these dynamics as depicted by the discrete updates given in red and blue. We then show that the center of the trajectory can still be computed with the perturbed discrete updates.
    } 
    \label{fig:HS_discretization} 
\end{figure}

This observation directly motivates our analysis. 
While it is true that the average position of \ref{eqn:CGD} will converge to the set of NE, we can compute the set of NE much more simply by observing that the dynamics lie on a hypersphere centered at a NE. 
Specifically, in 2 dimensions, we only require 3 points from the circle to directly compute its center, i.e., we only need to observe \ref{eqn:CGD} at three time-steps to compute the set of NE. 

While we cannot directly implement \ref{eqn:CGD}, we can closely approximate the system using a discrete-time algorithm via integration approaches. 
In physics and mathematical dynamics, there exists a class of integrators known as symplectic integrators that closely approximate Hamiltonian dynamics. 
\cite{Bailey19Hamiltonian,Bailey20Regret} show that the symplectic integrator, St\"{o}rmer-Verlet integration \cite{Hairer2006EnergyConserve}, yields \ref{eqn:AltGD} when applied to \ref{eqn:CGD}. 
This discretization is depicted in Figure \ref{fig:HS_discretization}. 

By \cite{Hairer2006EnergyConserve}, this discretization will approximately preserve the invariant energy function for \ref{eqn:CGD}. 
We provide an exact invariant energy function for \ref{eqn:AltGD} that is a small perturbation of the distance to the set of NE and prove that the invariant function can be used to directly characterize the set of NE after a finite number of iterations.

\section{Invariant Functions}\label{sec:InvariantFunctions}

In this section, we show that the approximated energy function corresponds to an invariant, perturbed energy function in discrete time, both in zero-sum and coordination games. 
We remark that the result for zero-sum games was first shown in \cite{Bailey20Regret}, and the result for coordination games was first shown in \cite{bailey2021left}. 
Both results were shown without linear cost vectors $b_1,b_2$, but the results from \cite{Bailey20Regret,bailey2021left} extend to this setting via a variable substitution.  
To keep the paper self-contained, we include the proofs here. 
Moreover, unlike previous results in the literature, we will use this invariant energy to quickly characterize the set of Nash equilibria (NE) in Section \ref{sec:ValidEquations}. 

\subsection{Zero-sum Games}

We first examine \ref{eqn:AltGD} in zero-sum games and show that the following perturbed energy is invariant:

\begin{definition}\label{def:Zero_Energy}
    Given that each agent uses \ref{eqn:AltGD} to obtain strategies $\{x^t,y^t\}_{t=0}^\infty$ in a zero-sum game, we denote the perturbed energy as
        \begin{equation*}
            h_-^{t}= \frac{\| x^{t} - x^{*} \|^{2}}{{\eta_1}} + \frac{\| y^{t} - y^{*} \|^{2}}{{\eta_2}} + \left\langle x^{t}, \ Ay^{t}-b_1 \right\rangle + \left\langle y^{t}, \ b_2 \right\rangle\label{eqn:ZeroEnergy}.
        \end{equation*}
\end{definition}

\begin{restatable}{theorem}{ZeroInvariant} \label{thm:zero_invariant}
    The energy function $h_-^{t}$ is time-invariant when agents update their strategies using \ref{eqn:AltGD} in a zero-sum game, i.e., $h_-^{t}=h_-^{0}, \forall t\in \mathbb{Z}_{\geq 0}$.
\end{restatable}
In Section \ref{sec:Hamiltonian}, we observed that the energy generated by \ref{eqn:CGD} corresponds to the distance to NE and the dynamics form a circle cycling around the set of NE.
Since \ref{eqn:AltGD} approximates this energy, there is inevitably a small perturbation in the distance to the set of NE.
Therefore, we first establish the following lemma, which analyzes the change in distance to the set of NE and provides a key property for Theorem \ref{thm:zero_invariant}. 
Then, in Section \ref{sec:ValidEquations}, we use this invariant function to describe valid linear equations to capture the set of NE followed by Sections \ref{sec:Base_Experiments} and \ref{sec:Improved_Experiments_Parallelization}, where we demonstrate that these equations can be used to solve for the set of NE. 

\begin{Lemma}\label{lem:zero_partial} With \ref{eqn:AltGD} updates in a zero-sum game, the change in the distance to NE for each agent is
        \begin{align*}
            \frac{\| x^{t+1} - x^{*} \|^{2} - \| x^{t} - x^{*} \|^{2}}{{\eta_1}}&= \left\langle x^{t+1} + x^{t} - 2x^{*}, \ A(y^{t}-y^{*}) \right\rangle \\
            \frac{\| y^{t+1} - y^{*} \|^{2} - \| y^{t} - y^{*} \|^{2}}{{\eta_2}}&= \left\langle y^{t+1} + y^{t} - 2y^{*}, \ -A^{\intercal}(x^{t+1}-x^{*}) \right\rangle.
        \end{align*}
\end{Lemma}

\begin{proof}
Expanding the right-hand side of the lemma yields
\begin{align*}
    \left\langle x^{t+1}+x^{t}-2x^{*}, \ A(y^{t}-y^{*}) \right\rangle
    =
    &\left\langle x^{t+1}+x^{t}-2x^{*}, \ Ay^{t}-Ay^{*} \right\rangle \\
    =
    &\left\langle x^{t+1}+x^{t}-2x^{*}, \ Ay^{t}-b_1 \right\rangle
\end{align*}
since $Ay^{*} = b_1$ by \ref{eqn:NEconditions}.
Then, by \ref{eqn:AltGD} and using the fact that the payoff matrices satisfy $B=-A^{\intercal}$ in a zero-sum game, we obtain
\begin{align*}
    \left\langle x^{t+1}+x^{t}-2x^{*}, \ A(y^{t}-y^{*}) \right\rangle
    =
    &\left\langle x^{t+1}+x^{t}-2x^{*}, \ Ay^{t}-b_1 \right\rangle\\
    =
    &\left\langle x^{t+1}+x^{t}-2x^{*}, \ \frac{x^{t+1}-x^{t}}{{\eta_1}} \right\rangle \\
    =
    &\frac{\left\langle x^{t+1}+x^{t}-2x^{*}, \ x^{t+1}-x^{t} \right\rangle}{{\eta_1}} \\
    = 
    &\frac{\| x^{t+1} \|^{2} -  \| x^{t} \|^{2} -2 \left\langle x^{*}, \ x^{t+1}-x^{t} \right\rangle}{{\eta_1}} \\
    = 
    &\frac{\| x^{t+1} \|^{2} - 2 \left\langle x^{*}, \ x^{t+1} \right\rangle + \| x^{*} \|^{2} - (\| x^{t} \|^{2} - 2 \left\langle x^{*}, \ x^{t} \right\rangle + \| x^{*} \|^{2})}{{\eta_1}} \\
    = 
    &\frac{\| x^{t+1} - x^{*}\|^{2} - \| x^{t} - x^{*} \|^{2}}{{\eta_1}}
\end{align*}
By following the same steps symmetrically for agent 2, we obtain 
\begin{align*}
\left\langle y^{t+1}+y^{t}-2y^{*}, \ -A^{\intercal}(x^{t+1}-x^{*}) \right\rangle
= 
\frac{\| y^{t+1} - y^{*}\|^{2} - \| y^{t} - y^{*} \|^{2}}{{\eta_2}}
\end{align*}
since $A^{\intercal}x^{*} = -b_2$ and $-A^{\intercal}x^{t+1}-b_2 = \frac{y^{t+1} - y^{t}}{{\eta_2}}$ in a zero-sum game, thereby yielding the statement of the lemma.
\end{proof}

Using Lemma \ref{lem:zero_partial}, we now proceed to prove the existence of a time-invariant function $h_-^{t}$ in discrete time, as stated in Theorem \ref{thm:zero_invariant}.

\begin{proof}[Proof of Theorem \ref{thm:zero_invariant}.]
Adding both equations in Lemma \ref{lem:zero_partial} and organizing the $t^{th}$ terms on the left and the $(t+1)^{th}$ terms on the right yields\\
    \begin{align*}
         h_-^t&=\frac{\| x^{t}-x^* \|^{2}}{\eta_1} + \frac{\| y^{t}-y^* \|^{2}}{\eta_2} + \left\langle x^{t}, \ Ay^{t}-b_1 \right\rangle + \left\langle y^{t}, \ b_2 \right\rangle \\ 
         &= \frac{\| x^{t+1}-x^* \|^{2}}{\eta_1} + \frac{\| y^{t+1}-y^* \|^{2}}{\eta_2} + \left\langle x^{t+1}, \ Ay^{t+1}-b_1 \right\rangle + \left\langle y^{t+1}, \ b_2 \right\rangle=h_{-}^{t+1}
    \end{align*}
Therefore, $h_-^{t} = h_-^{0}$, $\forall t\in \mathbb{Z}_{\geq 0}$.
\end{proof}

A more detailed proof of Theorem \ref{thm:zero_invariant} is given in Appendix \ref{sec:zero_invariant}. 

\subsection{Coordination Games}
We also establish a perturbed energy function for a coordination game.
To characterize this concept, we introduce the following definition of the discrete-time perturbed energy function.

\begin{definition}\label{def:Coord_Energy}
    Given that each agent uses \ref{eqn:AltGD} to obtain strategies $\{x^t,y^t\}_{t=0}^\infty$ in a coordination game, we denote the perturbed energy as
    \begin{align*}
        h_+^{t}= \frac{\| x^{t} - x^{*} \|^{2}}{{\eta_1}} - \frac{\| y^{t} - y^{*} \|^{2}}{{\eta_2}} + \left\langle x^{t}, \ Ay^{t}-b_1 \right\rangle - \left\langle y^{t}, \ b_2 \right\rangle.\label{eqn:Coord_Energy}
    \end{align*}
\end{definition}

\begin{restatable}{theorem}{CoordInvariant}\label{thm:coord_invariant}
    The energy function $h_+^t$ is time-invariant when agents update their strategies using \ref{eqn:AltGD} in a coordination game, i.e., $h_+^{t}=h_+^{0}, \forall t\in \mathbb{Z}_{\geq 0}$. 
\end{restatable}
The proof of Theorem \ref{thm:coord_invariant} follows identically to Theorem \ref{thm:zero_invariant}. 
The full proof can be found in Appendix \ref{sec:coord_invariant}.
\begin{restatable}{Lemma}{CoordPartial}\label{lem:coord_partial}
    With \ref{eqn:AltGD} in a coordination game, the change in the distance to NE for each agent is
    \begin{equation}
        \begin{aligned}
            \frac{\| x^{t+1} - x^{*} \|^{2} - \| x^{t} - x^{*} \|^{2}}{{\eta_1}} &= \left\langle x^{t+1} + x^{t} - 2x^{*}, \ A(y^{t}-y^{*}) \right\rangle \\
            \frac{\| y^{t+1} - y^{*} \|^{2} - \| y^{t} - x^{*} \|^{2}}{{\eta_2}} &= \left\langle y^{t+1} + y^{t} - 2y^{*}, \ A^{\intercal}(x^{t+1}-x^{*}) \right\rangle.
        \end{aligned}
    \end{equation}
\end{restatable}
The proof of Lemma \ref{lem:coord_partial} follows identically to Lemma \ref{lem:zero_partial}. 
The full proof can be found in Appendix \ref{sec:coord_partial}.

Having established the invariant energy functions $h_{-}^{t}$ and $h_{+}^{t}$ for zero-sum and coordination games, we now highlight the role of time-invariant functions in formulating valid linear equations that characterize the set of NE: $x^*,y^*$.
In Section \ref{sec:ValidEquations}, we use these invariants to capture the set of NE and demonstrate the computational performance of our approach in Sections \ref{sec:Base_Experiments} and \ref{sec:Improved_Experiments_Parallelization}.

\section{Characterizing the Set of Nash Equilibria}\label{sec:ValidEquations}
In this and subsequent sections, we use boldface script to denote unknown quantities, for example, $\boldsymbol{x^{*}}$ and $\boldsymbol{y^*}$ denote that $x^*$ and $y^*$ are unknown. 
All terms that are not bold-faced are observable quantities obtained via \ref{eqn:AltGD}, e.g., agent 1's strategy $x^t$ and agent 1's gradient $Ax^t-b_1$.
We remark that we only assume that we have access to agents' strategies and gradients, and not the payoff matrix $A$ or linear cost $b_i$, which is consistent with standard applications of online optimization in games. 
Our work then focuses on characterizing $\boldsymbol{x^{*}}$ and $\boldsymbol{y^*}$ using observable quantities.

We characterize the set of Nash equilibria (NE) using the time-invariant energy functions from Theorems \ref{thm:zero_invariant} and \ref{thm:coord_invariant}. 
By taking the difference of the energy functions with respect to consecutive iterations, we show that the resulting expression is linear in $\boldsymbol{x^{*}}, \boldsymbol{y^*},$ and $\boldsymbol{b_2}$, and depends on observable quantities such as agent strategies $x^t, y^t$ and agent 1's gradient $Ay^t-b_1$. 
Moreover, as long as $3k$ generated equations are linearly independent, we can solve for $\boldsymbol{x^{*}}, \boldsymbol{y^*},$ and $\boldsymbol{b_2}$, i.e., $3k$ iterations of \ref{eqn:AltGD} are sufficient to characterize the set of NE $(\boldsymbol{x^{*}},\boldsymbol{y^*})$ and agent 2's cost $\boldsymbol{b_2}$, modulo degeneracy.

However, our initial model suffers from degeneracy since $B\boldsymbol{x^{*}}=\boldsymbol{b_2}$. 
In Sections \ref{sec:ModelForSimulation}--\ref{sec:FlModel}, we resolve this issue with three modified models to characterize the set of NE. 
In particular, each model requires only $2k$, $2k+1$, and $k$ iterations of alternating gradient descent, respectively, an improvement over our initial model. 
Each model relies on different types of information from agents as given in Table \ref{tab:Models}.
Figure \ref{fig:valid_linear_equations} illustrates, through the model in Section \ref{sec:ModelForSimulation}, how we obtain the NE from valid linear equations, where $x^*, y^* \in \mathbb{R}^1$, requiring two independent linear equations derived via two iterations of \ref{eqn:AltGD}.

\begin{table}[!ht]
\centering
\renewcommand{\arraystretch}{1.6} 
\setlength{\tabcolsep}{2.6pt} 
    \begin{tabular}{|c|c |c|} 
    \hline
    Model & Solution Characterized & {Required Observations} \\ 
    \hline
    Section \ref{sec:ModelForSimulation}& $(\boldsymbol{x^{*}},\boldsymbol{y^*})$ & \multicolumn{1}{l|}{$\{x^t\}_{t=0}^{2k}, \{Ay^t-b_1\}_{t=0}^{2k}$,  $\{y^t\}_{t=0}^{2k}$} \\ 
    \hline
   Section \ref{sec:EconomicModel}& $(\boldsymbol{x^{*}},\boldsymbol{b_1})$ &  \multicolumn{1}{l|}{$\{x^t\}_{t=0}^{2k+1}, \{Ay^t-b_1\}_{t=0}^{2k+1}$, $\{||y^t||\}_{t=0}^{2k+1}$} \\ 
    \hline
    Section \ref{sec:FlModel} & $\boldsymbol{x^{*}}$ &  \multicolumn{1}{l|}{$\{x^t\}_{t=0}^{k}, \{Ay^t-b_1\}_{t=0}^{k}$, $\{||B x^t-b_2||\}_{t=0}^{k}$} \\ 
    \hline
    \end{tabular}
\caption{Three models to characterize the NE in zero-sum ($B=-A^\intercal$) and coordination  ($B=A^\intercal$) games. For the models in Sections \ref{sec:EconomicModel} and \ref{sec:FlModel}, a symmetric model can be used to characterize $(\boldsymbol{y^*},\boldsymbol{b_2})$ and $\boldsymbol{y^*}$, respectively. }\label{tab:Models}
\end{table}

\def\SIZER{.8}
\begin{figure}[H]
    \centering
    \begin{subfigure}{0.45\linewidth}
        \includegraphics[width=\SIZER\linewidth]{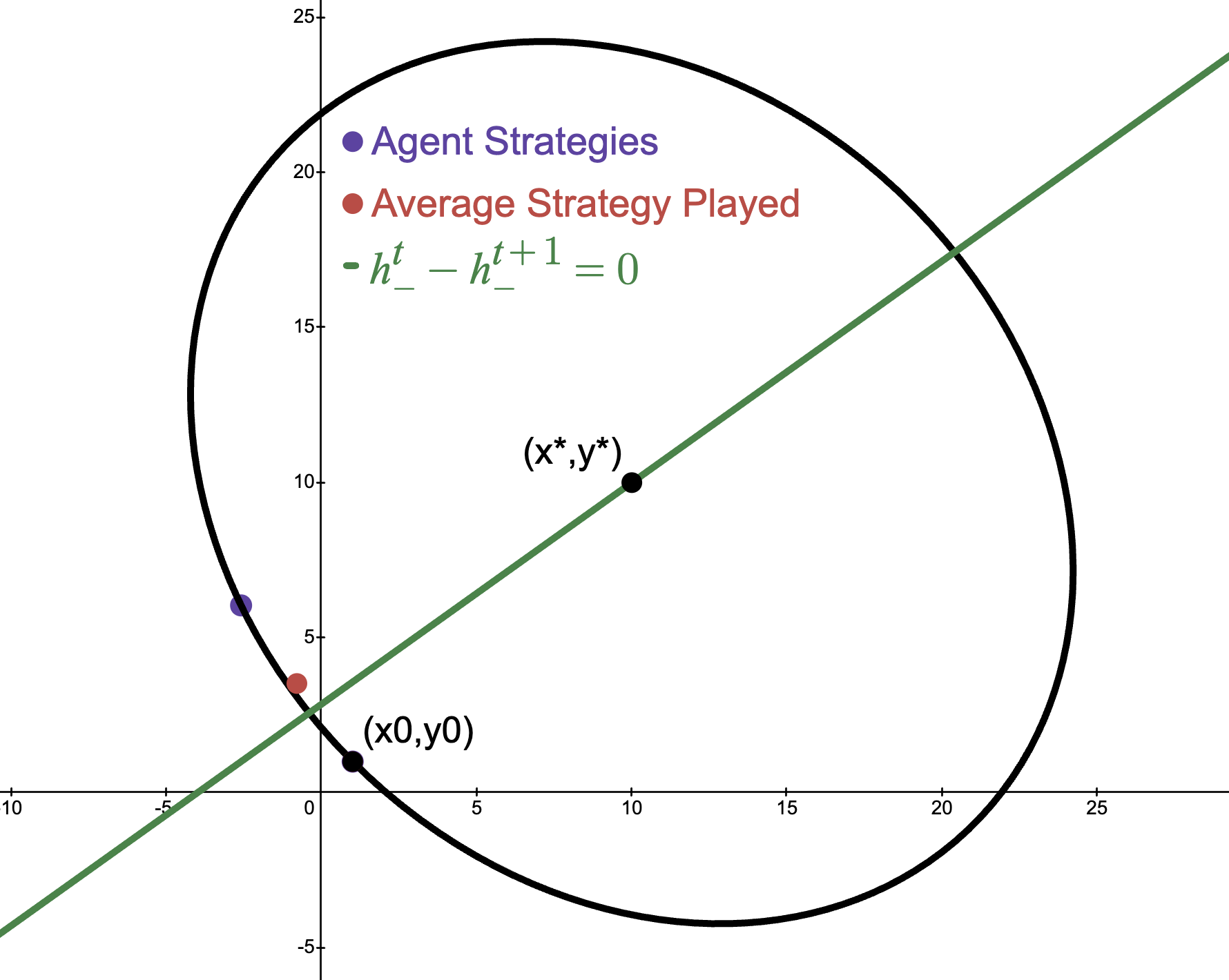}
        \caption{A linear equation after one update.}
        \label{fig:subfigA}
    \end{subfigure}
    \begin{subfigure}{0.45\linewidth}
            \includegraphics[width=\SIZER\linewidth]{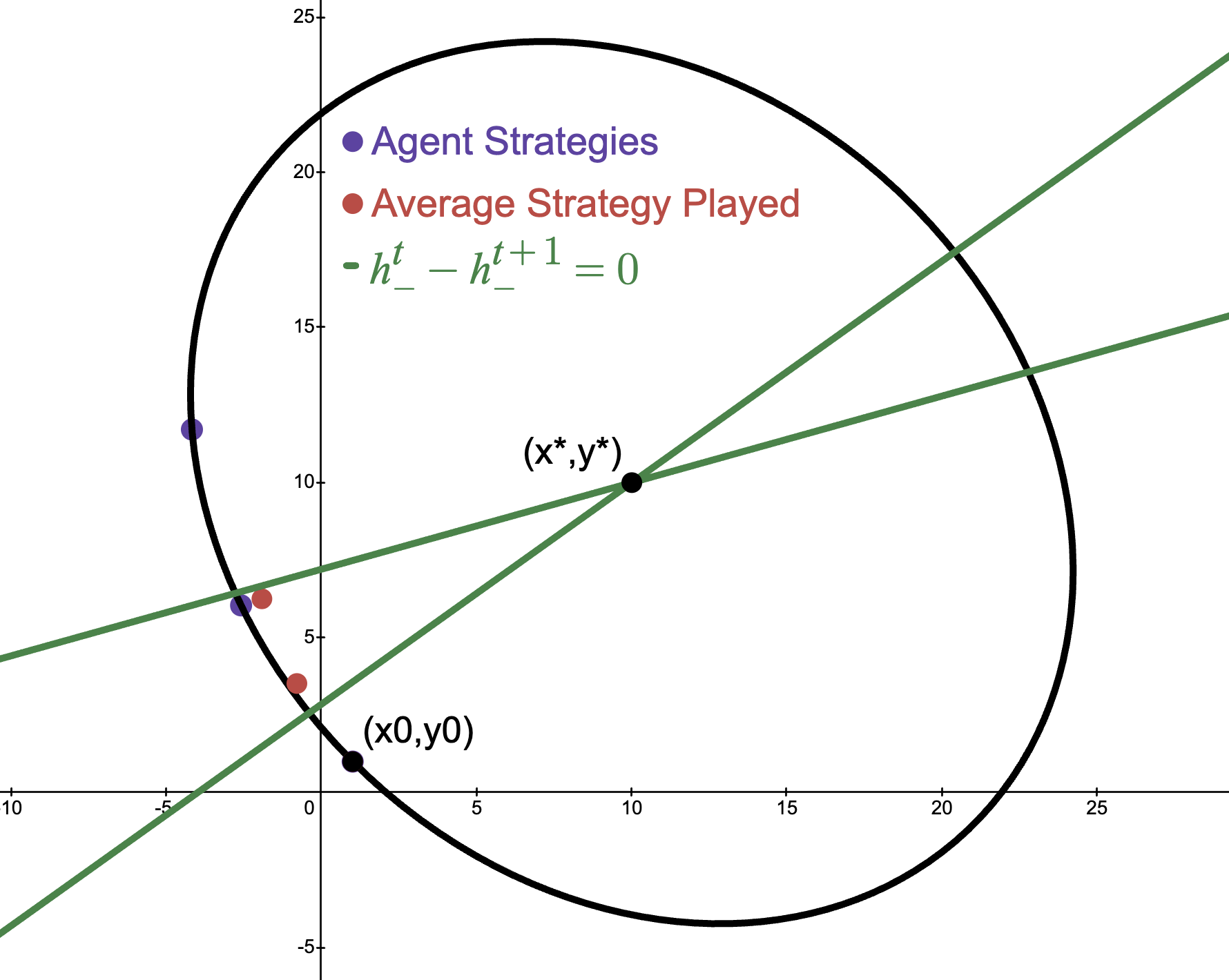}
            \caption{Two linear equations after two updates.}
            \label{fig:subfigC}
         \end{subfigure}
    \caption{Characterizing the set of NE via \ref{eqn:AltGD} in dimension 1. In contrast, the time-average of the strategies are still far from the set of NE. }
    \label{fig:valid_linear_equations}
\end{figure}

All three models are framed from the perspective of agent 1 and require knowledge of agent 1's gradient $(Ay^t-b_1)$ and agent 1's strategy $x^t$. 
For the model in Section \ref{sec:ModelForSimulation}, we additionally require knowledge of agent 2's strategy $y^t$; this is quite natural in settings such as Generative Adversarial Neural Networks (GANs) \cite{goodfellow2014generative} where a third party can observe the state of all agents in every iteration. 
In this model, we directly characterize $(\boldsymbol{x^{*}},\boldsymbol{y^*})$. 

In true settings of competition, e.g., economic models, or in settings where privacy is important, e.g., federated learning, directly accessing both agents' strategies is less realistic.
As a result, in the second two models (Sections \ref{sec:EconomicModel} and \ref{sec:FlModel}), we focus on decreasing the amount of information required from agent 2; instead of requiring agent 2's exact strategy, $y^t$, the proposed models instead only require the size of agent 2's strategy, $||y^t||$, or the size of agent 2's gradient, $||Bx^t-b_2||$, respectively. 
The increased level of privacy comes at a cost; when decreasing the amount of information from agent 2, we only characterize agent 1's NE $\boldsymbol{x^{*}}$. 
We remark that a symmetric model defined from the perspective of agent 2 can simultaneously learn agent 2's NE, $\boldsymbol{y^*}$.

\subsection{An Initial Model for Characterizing NE} \label{sec:DegenerateModel}

Since $h_-^t$ and $h_+^t$ are time-invariant by Theorem \ref{thm:zero_invariant} and \ref{thm:coord_invariant}, we can use consecutive iterations to characterize the set of NE. 
Specifically, we have $h_{-}^{t} - h_{-}^{t+1}=0$ in zero-sum games and $h_{+}^{t} - h_{+}^{t+1} = 0$ in coordination games. 
The unknown variables form a 3$k$-dimensional vector $(\boldsymbol{x^{*}}, \boldsymbol{y^*}, \boldsymbol{b_2})$, where $\boldsymbol{x^{*}}, \boldsymbol{y^*}, \boldsymbol{b_2} \in \mathbb{R}^{k}$. 
Theorems \ref{thm:zero_linear} and \ref{thm:coord_linear} ensure that by expanding $h_{-}^{t} - h_{-}^{t+1}=0$ and $h_{+}^{t} - h_{+}^{t+1} = 0$ in terms of $\boldsymbol{x^{*}}$, $\boldsymbol{y^*}$, and $\boldsymbol{b_2}$, we can construct linear systems describing these unknown quantities. 
Thus, it suffices to generate $3k$ equations using $3k+1$ pairs of ($x^{t}$, $y^{t}$) to characterize the set of NE.
We assume that at each iteration, agent 1 knows its strategy $x^t$, its gradient $Ay^t-b_1$, and has access to agent 2's strategy $y^t$.

\begin{theorem}\label{thm:zero_linear} 
Given two consecutive iterations of \ref{eqn:AltGD} in a zero-sum game, unknown variables $\boldsymbol{x^{*}}$, $\boldsymbol{y^*}$, and $\boldsymbol{b_2}$ satisfy a linear equation with observable quantities $\{x^t\}_{t=0}^{3k}$, $\{y^t\}_{t=0}^{3k}$ and $\{Ay^t-b_1\}_{t=0}^{3k}$.
Specifically, 
    \begin{align*}
        0= h_{-}^{t} - h_{-}^{t+1} = &\frac{\| x^{t} \|^{2} - \| x^{t+1} \|^{2} -2\left\langle x^{t}-x^{t+1}, \boldsymbol{x^{*}} \right\rangle}{\eta_1} + \frac{\| y^{t} \|^{2} - \|y^{t+1}\|^{2} -2\left\langle y^{t}-y^{t+1}, \boldsymbol{y^{*}} \right\rangle}{\eta_2}\\ &+ \left\langle x^{t}, \ Ay^{t}-b_1 \right\rangle + \left\langle y^{t}, \boldsymbol{\ b_2} \right\rangle - \left\langle x^{t+1}, \ Ay^{t+1}-b_1 \right\rangle - \left\langle y^{t+1}, \ \boldsymbol{b_2} \right\rangle .
    \end{align*}
\end{theorem}

\begin{proof}
By Theorem \ref{thm:zero_invariant}, the change in energy is expressed as
    \begin{align*}
        0= h_{-}^{t} - h_{-}^{t+1} 
        = 
        &\frac{\| x^{t} - \boldsymbol{x^{*}} \|^{2}}{{\eta_1}} + \frac{\| y^{t} - \boldsymbol{y^{*}} \|^{2}}{{\eta_2}} + \left\langle x^{t}, \ Ay^{t}-b_1 \right\rangle + \left\langle y^{t}, \ \boldsymbol{b_2} \right\rangle \\
        &- 
        \frac{\| x^{t+1} - \boldsymbol{x^{*}} \|^{2}}{{\eta_1}} - \frac{\| y^{t+1} - \boldsymbol{y^{*}} \|^{2}}{{\eta_2}} - \left\langle x^{t+1}, \ Ay^{t+1}-b_1 \right\rangle - \left\langle y^{t}, \ \boldsymbol{b_2} \right\rangle \\
        = 
        &\frac{\lVert x^{t}-\boldsymbol{x^*} \rVert^{2}-\lVert x^{t+1}-\boldsymbol{x^*} \rVert^{2}}{\eta_1}+ \frac{\lVert y^{t}-\boldsymbol{y^*} \rVert^{2}-\lVert y^{t+1}-\boldsymbol{y^*}\rVert^{2}}{\eta_2} \\
        &+ 
        \left\langle x^{t}, \ Ay^{t}-b_1 \right\rangle + \left\langle y^{t}, \ \boldsymbol{b_2} \right\rangle - \left\langle x^{t+1}, \ Ay^{t+1}-b_1 \right\rangle - \left\langle y^{t+1}, \ \boldsymbol{b_2} \right\rangle
    \end{align*}

Note that the first two terms in the above equations are the change in the distance to the set of NE of Lemma \ref{lem:zero_partial}.
By expanding the terms containing $\boldsymbol{x^{*}}$ and $\boldsymbol{y^*}$ as unknowns, we get
    \begin{align*}
        &\frac{\lVert x^{t}-\boldsymbol{x^*} \rVert^{2}-\lVert x^{t+1}-\boldsymbol{x^*} \rVert^{2}}{\eta_1}+ \frac{\lVert y^{t}-\boldsymbol{y^*} \rVert^{2}-\lVert y^{t+1}-\boldsymbol{y^*}\rVert^{2}}{\eta_2}\\ 
        =
        & \frac{\| x^{t} \|^{2} - \| x^{t+1} \|^{2} -2\left\langle x^{t}-x^{t+1}, \boldsymbol{x^*} \right\rangle}{\eta_1} + \frac{\| y^{t} \|^{2} - \|y^{t+1}\|^{2} -2\left\langle y^{t}-y^{t+1}, \boldsymbol{y^*} \right\rangle}{\eta_2}
    \end{align*}
    Then, substituting the result into $0= h_{-}^{t} - h_{-}^{t+1}$ yields 
    \begin{align*}
        0= h_{-}^{t} - h_{-}^{t+1}
        = & 
        \frac{\| x^{t} \|^{2} - \| x^{t+1} \|^{2} -2\left\langle x^{t}-x^{t+1}, \boldsymbol{x^*} \right\rangle}{\eta_1} + \frac{\| y^{t} \|^{2} - \|y^{t+1}\|^{2} -2\left\langle y^{t}-y^{t+1}, \boldsymbol{y^*} \right\rangle}{\eta_2} \\
        & +
        \left\langle x^{t}, \ Ay^{t}-b_1 \right\rangle + \left\langle y^{t}, \ \boldsymbol{b_2} \right\rangle - \left\langle x^{t+1}, \ Ay^{t+1}-b_1 \right\rangle - \left\langle y^{t+1}, \ \boldsymbol{b_2} \right\rangle
    \end{align*}
    as desired. 
\end{proof}

The established equation $h_{-}^{t} - h_{-}^{t+1}=0$ is linear in the variables $\boldsymbol{x^{*}}$, $\boldsymbol{y^*}$, and $\boldsymbol{b_2}$
allowing us to solve for these variables, provided that we have $3k$ linearly independent equations.
Therefore, modulo degeneracy, we can characterize the set of NE with linear equations by observing $3k$ consecutive iterations of \ref{eqn:AltGD}. 
A similar result holds for coordination games.

\begin{restatable}{theorem}
{CoordLinear}\label{thm:coord_linear}
    Given two consecutive iterations of \ref{eqn:AltGD} in a coordination game, unknown variables $\boldsymbol{x^{*}}$, $\boldsymbol{y^*}$, and $\boldsymbol{b_2}$ satisfy a linear equation with observable quantities $\{x^t\}_{t=0}^{3k}$, $\{y^t\}_{t=0}^{3k}$ and $\{Ay^t-b_1\}_{t=0}^{3k}$.
    Specifically, 
    \begin{align*}
        0= h_{+}^{t} - h_{+}^{t+1} = &\frac{\| x^{t} \|^{2} - \| x^{t+1} \|^{2} -2\left\langle x^{t}-x^{t+1}, \boldsymbol{x^{*}} \right\rangle}{\eta_1} - \frac{\| y^{t} \|^{2} - \|y^{t+1}\|^{2} -2\left\langle y^{t}-y^{t+1}, \boldsymbol{y^{*}} \right\rangle}{\eta_2}\\ &+ \left\langle x^{t}, \ Ay^{t}-b_1 \right\rangle - \left\langle y^{t}, \ \boldsymbol{b_2} \right\rangle - \left\langle x^{t+1}, \ Ay^{t+1}-b_1 \right\rangle + \left\langle y^{t+1}, \ \boldsymbol{b_2} \right\rangle.
    \end{align*}
\end{restatable}

The proof of Theorem \ref{thm:coord_linear} follows identically to Theorem \ref{thm:zero_linear} and appears in Appendix \ref{sec:coord_linear}. 
Thus, for both zero-sum and coordination games, our results suggest that we can characterize the NE after $3k$ iterations of \ref{eqn:AltGD}. 
However, the linear systems in Theorems \ref{thm:zero_linear} and \ref{thm:coord_linear} present a challenge;
since $B\boldsymbol{x^*}=\boldsymbol{b_2}$ by \ref{eqn:NEconditions}, the unknown variables $\boldsymbol{x^{*}}$ and $\boldsymbol{b_2}$ are linearly dependent. 
To address this issue, we propose three new models in the following subsections.

\subsection{Characterizing NE after 2k Iterations of Alternating GD}\label{sec:ModelForSimulation} 
We first show that the dependency on $\boldsymbol{b_2}$ can be removed. 
As a result, modulo degeneracy, we characterize $(\boldsymbol{x^{*}},\boldsymbol{y^*})$ after $2k$ iterations of \ref{eqn:AltGD} in both zero-sum and coordination games. 

\begin{restatable}{theorem}{ZeroW/B2}\label{thm:zero_w/_b2} Given two consecutive iterations of \ref{eqn:AltGD} in a zero-sum game, unknown variables $\boldsymbol{x^{*}}$ and $\boldsymbol{y^*}$ satisfy a linear equation with observable quantities $\{x^t\}_{t=0}^{2k}, \{y^t\}_{t=0}^{2k}$, $\{Ay^t-b_1\}_{t=0}^{2k}$. Specifically, 
\begin{align*}
    0 = &h_{-}^{t} - h_{-}^{t+1} \\
    = &\frac{\| x^{t} \|^{2} - \| x^{t+1} \|^{2} -2\left\langle x^{t}-x^{t+1}, \boldsymbol{x^{*}} \right\rangle}{\eta_1} + \frac{\| y^{t} \|^{2} - \|y^{t+1}\|^{2} -2\left\langle y^{t}-y^{t+1}, \boldsymbol{y^{*}} \right\rangle}{\eta_2} \\
    & + \left\langle x^{t}-\boldsymbol{x^{*}}, \ Ay^{t}-b_1 \right\rangle - \left\langle x^{t+1}-\boldsymbol{x^{*}}, \ Ay^{t+1}-b_1 \right\rangle.
\end{align*} 
\end{restatable}

\begin{proof}
    Recall agent 1's NE is such that $-A^\intercal \boldsymbol{x^{*}}= \boldsymbol{b_2}$, and therefore
    \begin{align*}
        \left\langle y^{t}, \ \boldsymbol{b_2} \right\rangle &= \left\langle y^{t}, \ -A^{\intercal}\boldsymbol{x^{*}} \right\rangle \\
        &= -\left\langle \boldsymbol{x^{*}}, \ Ay^{t} \right\rangle \\
        &= -\left\langle \boldsymbol{x^{*}}, \ Ay^{t}-b_1 \right\rangle-\left\langle \boldsymbol{x^{*}}, \ \boldsymbol{b_1} \right\rangle
    \end{align*}

    Following the above steps identically,
    \begin{align*}
        \left\langle y^{t+1}, \ \boldsymbol{b_2} \right\rangle = -\left\langle \boldsymbol{x^{*}}, \ Ay^{t+1}-b_1 \right\rangle-\left\langle \boldsymbol{x^{*}}, \ \boldsymbol{b_1} \right\rangle
    \end{align*}

    Theorem \ref{thm:zero_w/_b2} then follows by substituting these expressions for $\left\langle y^{t}, \ \boldsymbol{b_2} \right\rangle$ and $\left\langle y^{t+1}, \ \boldsymbol{b_2} \right\rangle$ into Theorem \ref{thm:zero_linear}.
    Notably, in this substitution, the nonlinear terms $\langle \boldsymbol{x^{*}}, \boldsymbol{b_1}\rangle$ cancel out.
\end{proof}

\begin{restatable}{theorem}{CoordWoB}\label{thm:coord_w/_b2} Given two consecutive iterations of \ref{eqn:AltGD} in a coordination game, unknown variables $\boldsymbol{x^{*}}$ and $\boldsymbol{y^*}$ satisfy linear equations with observable quantities $\{x^t\}_{t=0}^{2k}, \{y^t\}_{t=0}^{2k}$, $\{Ay^t-b_1\}_{t=0}^{2k}$. Specifically, 
    \begin{align*}
        0 = &h_{+}^{t} - h_{+}^{t+1} \\
        = &\frac{\| x^{t} \|^{2} - \| x^{t+1} \|^{2} -2\left\langle x^{t}-x^{t+1}, \boldsymbol{x^{*}} \right\rangle}{\eta_1} - \frac{\| y^{t} \|^{2} - \|y^{t+1}\|^{2} -2\left\langle y^{t}-y^{t+1}, \boldsymbol{y^{*}} \right\rangle}{\eta_2} \\
        & + \left\langle x^{t} - \boldsymbol{x^{*}}, \ Ay^{t}-b_1 \right\rangle - \left\langle x^{t+1} - \boldsymbol{x^{*}}, \ Ay^{t+1}-b_1 \right\rangle.
    \end{align*}
\end{restatable}

The proof of Theorem \ref{thm:coord_w/_b2} follows identically to Theorem \ref{thm:zero_w/_b2} and is provided in Appendix \ref{sec:coord_w/_b2}. 
The new linear equations in Theorem \ref{thm:coord_w/_b2} contain no terms involving $\boldsymbol{b_2}$, allowing us to solve the linear system with only $2k$ equations.

\subsection{Characterizing Agent 1's NE Using the Size of Agent 2's Strategy}\label{sec:EconomicModel}

In economic settings, it may be unnatural to assume that agent 1 has access to agent 2's strategies when computing their NE. 
In this section, we improve on the previous model by decreasing its dependency on agent 2's strategy. 
We remark that this adjustment initially introduces a nonlinear term $\langle \boldsymbol{x^*}, \boldsymbol{b_1}\rangle$ that is easily resolved in Theorem \ref{thm:zero_less_yt_linear}.

\begin{restatable}{theorem}{ZeroLessBnonlinear}\label{thm:zero_less_yt_nonlinear} Given two consecutive iterations of \ref{eqn:AltGD} in a zero-sum game, unknown variables $\boldsymbol{x^*}$ and $\boldsymbol{b_1}$ are characterized as \underline{nonlinear} equations with observable quantities $\{x^t\}_{t=0}^{2k}, \{||y^t||\}_{t=0}^{2k}, \{Ay^t-b_1\}_{t=0}^{2k}$. Specifically,
    \begin{align*}
        0 = &h_{-}^{t} - h_{-}^{t+1} \\ 
        = &\frac{\| x^{t} \|^{2} - \| x^{t+1} \|^{2} -2\left\langle x^{t}-x^{t+1}, \boldsymbol{x^{*}} \right\rangle}{\eta_1} 
        + 
        \frac{\| y^{t} \|^{2} - \| y^{t+1} \|^{2}}{\eta_2} -2\left\langle x^{t+1}, \boldsymbol{b_1} \right\rangle 
        +
        2\left\langle \boldsymbol{x^{*}}, \boldsymbol{b_1} \right\rangle \\
        &+ 
        \left\langle x^{t}-\boldsymbol{x^{*}}, \ Ay^{t}-b_1 \right\rangle 
        - 
        \left\langle x^{t+1}-\boldsymbol{x^{*}}, \ Ay^{t+1}-b_1 \right\rangle
    \end{align*}

    \begin{proof}
        Expanding the terms in $h_{-}^{t}-h_{-}^{t+1}=0$, which involve $y^{t}$ and $y^{t+1}$ in Theorem \ref{thm:zero_invariant}, results in the following.
        \begin{align*}
            \frac{\| y^{t}-\boldsymbol{y^*} \|^{2}}{\eta_2} - \frac{\| y^{t+1}-\boldsymbol{y^*} \|^{2}}{\eta_2}=&\frac{\| y^{t} \|^{2} - 2\left\langle y^{t}, \boldsymbol{y^*} \right\rangle + \|\boldsymbol{y^*} \|^{2}}{\eta_2} - \frac{\| y^{t+1} \|^{2} - 2\left\langle y^{t+1}, \boldsymbol{y^*} \right\rangle + \|\boldsymbol{y^*} \|^{2}}{\eta_2} \\=&\frac{\| y^{t} \|^{2} - \| y^{t+1} \|^{2} - 2\left\langle y^{t}, \boldsymbol{y^*} \right\rangle +2\left\langle y^{t+1}, \boldsymbol{y^*} \right\rangle}{\eta_2}
        \end{align*}

        Applying \ref{eqn:AltGD}, \ref{eqn:NEconditions}, and $B=-A^{\intercal}$ yields
        \begin{align*}
            - \frac{2\left\langle y^{t}, \boldsymbol{y^*} \right\rangle}{\eta_2} + \frac{2\left\langle y^{t+1}, \boldsymbol{y^*} \right\rangle}{\eta_2} &= - \frac{2\left\langle y^{t}, \boldsymbol{y^*} \right\rangle}{\eta_2} + \frac{\left\langle 2y^{t} + 2\eta_2(-A^{\intercal}x^{t+1}-b_2), \boldsymbol{y^*} \right\rangle}{\eta_2} \\
            &= - \frac{2\left\langle y^{t}, \boldsymbol{y^*} \right\rangle}{\eta_2} + \frac{2\left\langle y^{t}, \boldsymbol{y^*} \right\rangle}{\eta_2} + \frac{2\left\langle \eta_2(-A^{\intercal}x^{t+1}-b_2), \boldsymbol{y^*} \right\rangle}{\eta_2} \\
            &= -2\left\langle A^{\intercal}x^{t+1}, \boldsymbol{y^*} \right\rangle -2\left\langle \boldsymbol{b_2}, \boldsymbol{y^*} \right\rangle \\
            &= -2\left\langle x^{t+1}, A\boldsymbol{y^*} \right\rangle -2\left\langle -A^{\intercal}\boldsymbol{x^*}, \boldsymbol{y^*} \right\rangle \\
            &= -2\left\langle x^{t+1}, \boldsymbol{b_1} \right\rangle +2\left\langle \boldsymbol{x^*}, A\boldsymbol{y^*} \right\rangle \\
            &= -2\left\langle x^{t+1}, \boldsymbol{b_1} \right\rangle +2\left\langle \boldsymbol{x^*}, \boldsymbol{b_1} \right\rangle
        \end{align*}
        implying
        \begin{align*}
            \frac{\| y^{t}-y^* \|^{2}}{\eta_2} - \frac{\| y^{t+1}-y^* \|^{2}}{\eta_2} &=\frac{\| y^{t} \|^{2} - \| y^{t+1} \|^{2}}{\eta_2} -2\left\langle x^{t+1}, \boldsymbol{b_1} \right\rangle +2\left\langle \boldsymbol{x^*}, \boldsymbol{b_1} \right\rangle
        \end{align*}
        Therefore, 
        \begin{align*}
            0 = &h_{-}^{t} - h_{-}^{t+1} \\= &\frac{\| x^{t} \|^{2} - \| x^{t+1} \|^{2} -2\left\langle x^{t}-x^{t+1}, \boldsymbol{x^*} \right\rangle}{\eta_1} + \frac{\| y^{t} \|^{2} - \| y^{t+1} \|^{2}}{\eta_2} -2\left\langle x^{t+1}, \boldsymbol{b_1} \right\rangle +2\left\langle \boldsymbol{x^{*}}, \boldsymbol{b_1} \right\rangle \\
            &+ \left\langle x^{t}-\boldsymbol{x^*}, \ Ay^{t}-b_1 \right\rangle - \left\langle x^{t+1}-\boldsymbol{x^*}, \ Ay^{t+1}-b_1 \right\rangle.
        \end{align*}
        as desired.
    \end{proof}
\end{restatable}

Since the nonlinear term $\langle \boldsymbol{x^*}, \boldsymbol{b_1}\rangle$ is independent of time, it is straightforward to eliminate the nonlinearities by combining consecutive equations.

\begin{restatable}{theorem}{ZeroLessBlinear}\label{thm:zero_less_yt_linear} Given three consecutive iterations of \ref{eqn:AltGD} in a zero-sum game, unknown variables $\boldsymbol{x^*}$ and $\boldsymbol{b_1}$ are characterized as linear equations with observable quantities $\{x^t\}_{t=0}^{2k+1}, \{||y^t||\}_{t=0}^{2k+1}, \{Ay^t-b_1\}_{t=0}^{2k+1}$. Specifically,
    \begin{align*}
        0 &= (h_{-}^{t} - h_{-}^{t+1}) - (h_{-}^{t+1} - h_{-}^{t+2}) \\[5pt]
        &= 
        h_{-}^{t} - 2 h_{-}^{t+1} + h_{-}^{t+2} \\[5pt]
        &= 
        \frac{\|x^t\|^2 - 2\|x^{t+1}\|^2 + \|x^{t+2}\|^2}{\eta_1} 
        +
        \frac{\|y^t\|^2 - 2\|y^{t+1}\|^2 + \|y^{t+2}\|^2}{\eta_2} \\[5pt]
        & \quad + 
        \left\langle x^t - \boldsymbol{x^*},\ Ay^t - b_1 \right\rangle 
        - 
        2\left\langle x^{t+1} - \boldsymbol{x^*},\ Ay^{t+1} - b_1 \right\rangle 
        + \left\langle x^{t+2} - \boldsymbol{x^*},\ Ay^{t+2} - b_1 \right\rangle \\[5pt]
        & \quad - 
        \frac{2\left\langle x^t - 2x^{t+1} + x^{t+2},\ \boldsymbol{x^*} \right\rangle}{\eta_1}
        -
        2\left\langle x^{t+1} - x^{t+2},\ \boldsymbol{b_1} \right\rangle.
    \end{align*}
\end{restatable}

A proof follows directly by subtracting consecutive equations from Theorem \ref{thm:zero_less_yt_nonlinear}.
Doing so requires an additional constraint, which requires $2k+1$ iterations of \ref{eqn:AltGD} to characterize the set of NE.
The same result also holds for coordination games. 

\begin{restatable}{theorem}{CoordLessBnonlinear}\label{thm:coord_less_yt_nonlinear} Given two consecutive iterations of \ref{eqn:AltGD} in a coordination game, unknown variables $\boldsymbol{x^*}$ and $\boldsymbol{b_1}$ are characterized as \underline{nonlinear} equations with observable quantities $\{x^t\}_{t=0}^{2k}, \{||y^t||\}_{t=0}^{2k}, \{Ay^t-b_1\}_{t=0}^{2k}$. Specifically,
    \begin{align*}
        0 = &h_{+}^{t} - h_{+}^{t+1} \\
        = &\frac{\| x^{t} \|^{2} - \| x^{t+1} \|^{2} -2\left\langle x^{t}-x^{t+1}, \boldsymbol{x^{*}} \right\rangle}{\eta_1} - \frac{\| y^{t} \|^{2} - \| y^{t+1} \|^{2}}{\eta_2}-2\left\langle x^{t+1}, \boldsymbol{b_1} \right\rangle +2\left\langle \boldsymbol{x^{*}}, \boldsymbol{b_1} \right\rangle \\
        & + \left\langle x^{t}-\boldsymbol{x^{*}}, \ Ay^{t}-b_1 \right\rangle - \left\langle x^{t+1}-\boldsymbol{x^{*}}, \ Ay^{t+1}-b_1 \right\rangle.
    \end{align*}
\end{restatable}
A proof of Theorem \ref{thm:coord_less_yt_nonlinear} follows identically to Theorem \ref{thm:zero_less_yt_nonlinear}. The full details are provided in Appendix \ref{sec:coord_less_yt}.

\begin{restatable}{theorem}{CoordLessBlinear}\label{thm:coord_less_yt_linear} Given three consecutive iterations of \ref{eqn:AltGD} in a coordination game, unknown variables $\boldsymbol{x^*}$ and $\boldsymbol{b_1}$ are characterized as linear equations with observable quantities $\{x^t\}_{t=0}^{2k+1}, \{||y^t||\}_{t=0}^{2k+1}, \{Ay^t-b_1\}_{t=0}^{2k+1}$. Specifically,
    \begin{align*}
        0 
        &= 
        (h_{+}^{t} - h_{+}^{t+1}) - (h_{+}^{t+1} - h_{+}^{t+2}) \\[5pt]    
        &= 
        h_{+}^{t} - 2 h_{+}^{t+1} + h_{+}^{t+2} \\[5pt]
        &= 
        \frac{\|x^t\|^2 - 2\|x^{t+1}\|^2 + \|x^{t+2}\|^2}{\eta_1} 
        -
        \frac{\|y^t\|^2 - 2\|y^{t+1}\|^2 + \|y^{t+2}\|^2}{\eta_2} \\[5pt]
        &\quad + 
        \left\langle x^t - \boldsymbol{x^*},\ Ay^t - b_1 \right\rangle 
        - 
        2\left\langle x^{t+1} - \boldsymbol{x^*},\ Ay^{t+1} - b_1 \right\rangle 
        + 
        \left\langle x^{t+2} - \boldsymbol{x^*},\ Ay^{t+2} - b_1 \right\rangle \\[5pt]
        &\quad -
        \frac{2\left\langle x^t - 2x^{t+1} + x^{t+2},\ \boldsymbol{x^*} \right\rangle}{\eta_1}
        - 
        2\left\langle x^{t+1} - x^{t+2},\ \boldsymbol{b_1} \right\rangle
    \end{align*}
\end{restatable}
A proof follows directly by subtracting consecutive equations from Theorem \ref{thm:coord_less_yt_nonlinear}.
Doing so requires an additional constraint, which requires $2k+1$ iterations of \ref{eqn:AltGD} to characterize the set of NE.

\subsection{Characterizing Agent 1's NE Using the Size of Agent 2's Gradient}\label{sec:FlModel}

In this section, we provide an alternative model to the model in Section \ref{sec:EconomicModel}. 
Like the model in the previous section, we do not directly rely on observing agent 2's strategy when characterizing agent 1's NE.  
Instead of requiring the size of the opponent's strategy, we require the size of the opponent's gradient.

\begin{restatable}{theorem}{ZeroNoY}\label{thm:zero_no_y} Given two consecutive iterations of \ref{eqn:AltGD} in a zero-sum game, unknown variable $\boldsymbol{x^*}$ is characterized as linear equations with observable quantities $\{x^t\}_{t=0}^{k}, \{||-A^\intercal x^t-b_2||^2\}_{t=0}^{k}, \{Ay^t-b_1\}_{t=0}^{k}$. Specifically,
    \begin{align*}
        {0} = &h_{-}^{t} - h_{-}^{t+1} \\
        = &\frac{\| x^{t} \|^{2} - \| x^{t+1} \|^{2}}{\eta_1} + 
        2\left\langle x^{t+1}, \frac{x^{t+1}-x^{t}}{\eta_1} \right\rangle - \eta_2\|-A^{\intercal}x^{t+1}-b_2\|^{2}\\
        & + \left\langle x^{t}-\boldsymbol{x^{*}}, \ Ay^{t}-b_1 \right\rangle - \left\langle x^{t+1}-\boldsymbol{x^{*}}, \ Ay^{t+1}-b_1 \right\rangle. \\
    \end{align*}
\end{restatable}
    \begin{proof}
        Appplying \ref{eqn:AltGD}, \ref{eqn:NEconditions}, and $B=-A^{\intercal}$ yields
        \begin{align*}
            \frac{\| y^{t} \|^{2} - \| y^{t+1} \|^{2}}{\eta_2} =&\frac{\left\langle y^{t}-y^{t+1}, \ y^{t}+y^{t+1} \right\rangle}{\eta_2} \\
            = &\frac{\left\langle \eta_2(A^{\intercal}x^{t+1}+b_2), \ 2y^{t}+\eta_2(-A^{\intercal}x^{t+1}-b_2) \right\rangle}{\eta_2}\\
            = &\left\langle A^{\intercal}x^{t+1}+b_2, \ 2y^{t}\right\rangle - \eta_2\left\langle A^{\intercal}x^{t+1}+b_2, A^{\intercal}x^{t+1}+b_2\right\rangle \\
            = &\left\langle A^{\intercal}x^{t+1}+b_2, \ 2y^{t}\right\rangle - \eta_2\|A^{\intercal}x^{t+1}-A^{\intercal}\boldsymbol{x^*}\|^{2} \\
            = &2\left\langle A^{\intercal}x^{t+1}, \ y^{t}\right\rangle + 2\left\langle \boldsymbol{b_2}, \ y^{t}\right\rangle - \eta_2\|A^{\intercal}(x^{t+1}-\boldsymbol{x^*})\|^{2} \\
            = &2\left\langle x^{t+1}, \ Ay^{t}\right\rangle + 2\left\langle -A^{\intercal}\boldsymbol{x^*}, \ y^{t}\right\rangle - \eta_2\|A^{\intercal}(x^{t+1}-\boldsymbol{x^*})\|^{2} \\
            = &2\left\langle x^{t+1}, \ Ay^{t}\right\rangle - 2\left\langle \boldsymbol{x^*}, \ Ay^{t}\right\rangle - \eta_2\|A^{\intercal}(x^{t+1}-\boldsymbol{x^*})\|^{2} \\
            = &2\left\langle x^{t+1}-\boldsymbol{x^*}, \ Ay^{t}\right\rangle - \eta_2\|A^{\intercal}(x^{t+1}-\boldsymbol{x^*})\|^{2} \\
            = &2\left\langle x^{t+1}-\boldsymbol{x^*}, \frac{x^{t+1}-x^{t}}{\eta_1}+\boldsymbol{b_1}\right\rangle - \eta_2\|-A^{\intercal}x^{t+1}-b_2\|^{2} \\
            = &2\left\langle x^{t+1}, \frac{x^{t+1}-x^{t}}{\eta_1}+\boldsymbol{b_1}\right\rangle - 2\left\langle \boldsymbol{x^{*}}, \frac{x^{t+1}-x^{t}}{\eta_1}+\boldsymbol{b_1}\right\rangle - \eta_2\|-A^{\intercal}x^{t+1}-b_2\|^{2} \\
            = &2\left\langle x^{t+1}, \frac{x^{t+1}-x^{t}}{\eta_1}\right\rangle + 2\left\langle x^{t+1}, \boldsymbol{b_1} \right\rangle - 2\left\langle \boldsymbol{x^{*}}, \frac{x^{t+1}-x^{t}}{\eta_1} \right\rangle - 2\left\langle \boldsymbol{x^{*}}, \boldsymbol{b_1} \right\rangle  - \eta_2\|-A^{\intercal}x^{t+1} - b_2\|^{2} \\
        \end{align*}
    By plugging in the above result into the equation in Theorem \ref{thm:zero_less_yt_nonlinear}, we obtain 
        \begin{align*}
            0 = &h_{-}^{t} - h_{-}^{t+1} \\
            = &\frac{\| x^{t} \|^{2} - \| x^{t+1} \|^{2}}{\eta_1} + 
            2\left\langle x^{t+1}, \frac{x^{t+1}-x^{t}}{\eta_1} \right\rangle - \eta_2\|-A^{\intercal}x^{t+1}-b_2\|^{2}\\
            & + \left\langle x^{t}-\boldsymbol{x^{*}}, \ Ay^{t}-b_1 \right\rangle - \left\langle x^{t+1}-\boldsymbol{x^{*}}, \ Ay^{t+1}-b_1 \right\rangle. 
        \end{align*}
    which is linear in $x^{*}$, as desired.
\end{proof}

\begin{restatable}{theorem}{CoordNoY}\label{thm:coord_no_y} Given two consecutive iterations of \ref{eqn:AltGD} in a coordination game, unknown variable $\boldsymbol{x^*}$ is characterized as linear equations with observable quantities $\{x^t\}_{t=0}^{k}, \{||A^\intercal x^t-b_2||^2\}_{t=0}^{k}, \{Ay^t-b_1\}_{t=0}^{k}$. Specifically,
    \begin{align*}
        0 = &h_{+}^{t} - h_{+}^{t+1} \\
        = &\frac{\| x^{t} \|^{2} - \| x^{t+1} \|^{2}}{\eta_1} 
        + 
        2\left\langle x^{t+1}, \frac{x^{t+1}-x^{t}}{\eta_1} \right\rangle 
        + 
        \eta_2\|A^{\intercal}x^{t+1}-b_2\|^{2} \\
        &+ \left\langle x^{t}-\boldsymbol{x^{*}}, \ Ay^{t}-b_1 \right\rangle - \left\langle x^{t+1}-\boldsymbol{x^{*}}, \ Ay^{t+1}-b_1 \right\rangle.
    \end{align*}
\end{restatable}

A proof of Theorem \ref{thm:coord_no_y} follows identically to Theorem \ref{thm:zero_no_y}. Full details are provided in Appendix \ref{sec:coord_no_y}.

\section{Base Experiments}\label{sec:Base_Experiments}

In this section, we investigate the numerical stability and solution quality for the set of Nash equilibria (NE) in two-player zero-sum and coordination games using different characterizations discussed in Section \ref{sec:ValidEquations}.
While our results provide theoretical support for our proposed methods, constructing such linear systems comes with numerical challenges.
To accurately capture the set of NE, the system matrices must be non-singular or have condition numbers that are not too large.
In practice, the matrices built from \ref{eqn:AltGD} updates often become nearly singular or ill-conditioned due to the correlation of strategies, especially in higher-dimensional games. 
It enlarges the numerical errors that deteriorate the quality of the approximated set of NE.
To address the issue, we explore alternatives such as the least squares method and the Tikhonov regularization (ridge regression) method, which can stabilize solutions to some extent. 
Least squares methods can mitigate the effects of small perturbations on sensitive changes in the approximated set of NE, while Tikhonov regularization introduces a controlled bias to improve the condition of the constructed linear system.
Our experiments systematically examine how the dimension of the game affects the stability of the standard solution, least squares, and Tikhonov regularization methods, motivating the need for more robust techniques.
Although these two regularization methods partially resolve stability issues, we also introduce a refined method using parallelization in Section \ref{sec:Improved_Experiments_Parallelization}.

\subsection{Experiment Setup}
We assess the performance of our models by how well they approximate the set of NE in zero-sum games with small and large learning rates, as well as coordination games across different dimensions.
The results are analyzed in terms of relative errors between the true set of NE, which is obtained by solving \ref{eqn:NEconditions}, and the set of NE solutions obtained by solving the constructed linear system.
We consider dimensions 1 through 20, generating 30 random instances per dimension. 
Each game is defined by a payoff matrix $A$ and cost vectors $b_1$ and $b_2$ of two agents.
For zero-sum games with small learning rates, the payoff matrix $A$ and the cost vectors $b_1$ and $b_2$ are randomly generated following the uniform distribution, specifically $A\sim [U(-1,1)]^{k\times k} \ \text{and} \  b_1, b_2 \in [U(-1,1)]^k$.
The learning rates $\eta_1$ and $\eta_2$ are randomly generated following a uniform distribution in $\left[-\frac{1}{k}, \frac{1}{k}\right]$, which satisfies the condition $\sqrt{\eta_1 \cdot \eta_2} < \frac{2}{\|A\|}$ ensuring bounded trajectories and time-average convergence of strategies via \ref{eqn:AltGD} \cite{bailey2021left}.
Although our theoretical results do not explicitly require these properties, bounded orbits naturally improve numerical stability. 
We also consider larger learning rates that violate these conditions. 
Specifically, we use the same game instances but with learning rates multiplied by 10 to make the time-average strategies diverge from the true set of NE.
Finally, for coordination games, the data is randomly generated identically to zero-sum games with small learning rates, but the payoff matrix $A$ is set to represent the structure of the coordination game instead of the zero-sum.
We remark that it is well known that strategies diverge in this setting, as in the zero-sum case when learning rates are large. 
Experiments are conducted on a MacBook Pro (Mac15,3) with an 8-core Apple M3 Chip and 16GB RAM using Python 3.12.7.
The Python code for all experiments is available at \href{https://github.com/walnutpie777/RPI}{https://github.com/walnutpie777/RPI}.

\textbf{A Note on Small vs. Large Learning Rates:}
Our analysis thus far has been independent of learning rates. 
We consider this a strength of our approach as our proposed method, at least theoretically, is the first method that uses learning algorithms with arbitrary learning rates to find the set of NE in a dimension larger than two. 
All existing methods require learning rates to be sufficiently small; for instance, \ref{eqn:AltGD} and optimistic gradient descent achieve time-average convergence if and only if the learning rates satisfy $\sqrt{\eta_1\cdot\eta_2}< 2/||A||$ and $\sqrt{\eta_1\cdot\eta_2}\leq 1/||A||$ respectively, otherwise the learning dynamics and the time-average of the dynamics, rapidly diverge from the set of NE.
As a result, when verifying the numerical stability of our approach, we consider two sets of experiments with both small and large learning rates. 
Moreover, with a parallelization refinement, which is discussed in Section \ref{sec:Improved_Experiments_Parallelization}, we decouple convergence behavior from the need to adjust learning rates. 
This provides an exception to the widely held belief that small learning rates are essential to finding NE.

\subsection{Standard Solution Method}

We first evaluate our proposed approaches from Sections \ref{sec:ModelForSimulation}--\ref{sec:FlModel} to directly solve the constructed linear systems without considering the numerical instability.
Given a payoff matrix $A$ and cost vectors $b_1, b_2$ that define a game environment, we simulate a trajectory of $2k, 2k+1$, or $k$ strategies depending on the model by applying \ref{eqn:AltGD} updates after initializing with a random starting strategy. 
Using the difference between two (or three) consecutive iterations of the energy function, we formulate linear equations whose unknowns are the set of NE, specifically $\boldsymbol{x^*}, \boldsymbol{y^*} \in \mathbb{R}^{k}$ for the model in Section \ref{sec:ModelForSimulation}, $\boldsymbol{x^*}, \boldsymbol{b_1} \in \mathbb{R}^{k}$ for the model in Section \ref{sec:EconomicModel}, and $\boldsymbol{x^*} \in \mathbb{R}^{k}$ for the model in Section \ref{sec:FlModel}.
The resulting linear systems are then solved directly, and the relative errors are measured with respect to the true set of NE.

\subsubsection{Zero-sum Games with Small Learning Rates} \label{subsubsec:standard_zero_non_violated}

As observed in Table \ref{tab:zero_avg_rel_errors_rt}, the determinant of the system matrices rapidly drops to zero as the dimension increases.
For instance, even in dimension 3, the determinant is in order of $10^{-7}$.
As a result, relative errors in solutions $x^*, y^*$ grow fast as the dimension increases. 
This suggests that directly solving the linear systems generally is not reliable without additional stabilization techniques. 
The complete results of Table \ref{tab:zero_avg_rel_errors_rt} across dimensions 1 to 20 are provided in Appendix \ref{app:tables}.

\begin{table}[!ht]
    \centering
    \resizebox{1\textwidth}{!}{%
    \begin{tabular}{c|cc|c|c|c}
        \toprule
        \textbf{Model}: & \multicolumn{2}{c|}{\textbf{Section \ref{sec:ModelForSimulation}}} & \textbf{Section} \ref{sec:EconomicModel} & \textbf{Section} \ref{sec:FlModel} & \\
        \hline 
        \textbf{Dimension} & \textbf{Avg Rel Err of $\boldsymbol{x^{*}}$} & \textbf{Avg Rel Err in $\boldsymbol{y^*}$} & \textbf{Avg Rel Err in $\boldsymbol{x^{*}}$} & \textbf{Avg Rel Err in $\boldsymbol{x^{*}}$} & $\boldsymbol{\textbf{det} (A_{system_1}})$ \\
        \hline
        \midrule
        5  & 1.43E+01  & 2.03E+00  & 2.69E+01  & 7.19E-05  & 3.61E-31\\
        10 & 4.40E+01  & 2.88E+01  & 2.15E+02  & 5.23E+01  & 2.20E-14\\
        15 & 1.07E+02  & 2.53E+02  & 4.65E+02  & 5.68E+02  & 3.40E-26\\
        20 & 1.43E+02  & 1.20E+02  & 3.05E+02  & 6.02E+02  & $\approx$ 0.0\\
        \bottomrule
    \end{tabular}
    }
    \caption{Zero-sum game with small learning rates: average relative error of NE solutions and determinant of linear system matrices in dimensions 5, 10, 15, 20 derived from solving linear systems.}
    \label{tab:zero_avg_rel_errors_rt}
\end{table}

\subsubsection{Zero-sum Games with Large Learning Rates}\label{subsubsec:standard_zero_violated}

As seen in Table \ref{tab:zero_avg_rel_errors_rt_violation}, the numerical stability of the linear systems used to solve for the set of NE deteriorates dramatically as the dimension of the games increases.
Interestingly, however, the error is not due to a near-singular matrix. 
Rather, the system of equations' condition numbers rapidly increase, indicating that the system is extremely sensitive to numerical perturbation.
Unfortunately, when learning rates are too large, strategies rapidly expand from the set of Nash equilibria, resulting in significant rounding errors. 
When combined with the large condition number, we should expect the poor results shown in Table \ref{tab:zero_avg_rel_errors_rt_violation}. 
The complete results of Table \ref{tab:zero_avg_rel_errors_rt_violation} across dimensions 1 to 20 are provided in Appendix \ref{app:tables}.

\begin{table}[!ht]
    \centering
    \resizebox{1\textwidth}{!}{%
    \begin{tabular}{c|cc|c|c|c|c}
        \toprule
        \textbf{Model}: & \multicolumn{2}{c|}{\textbf{Section \ref{sec:ModelForSimulation}}} & \textbf{Section} \ref{sec:EconomicModel} & \textbf{Section} \ref{sec:FlModel} & \\
        \hline 
        \textbf{Dimension} & \textbf{Avg Rel Err of $\boldsymbol{x^{*}}$} & \textbf{Avg Rel Err in $\boldsymbol{y^*}$} & \textbf{Avg Rel Err in $\boldsymbol{x^{*}}$} & \textbf{Avg Rel Err in $\boldsymbol{x^{*}}$} & $\boldsymbol{\textbf{det} (A_{system_1}})$ & $\boldsymbol{\textbf{cond} (A_{system_1}})$ \\
        \hline
        \midrule
        5  & 2.69E+26    & 2.00E+27    & 4.46E+28     & 8.76E+00     & 1.11E+78    & 5.21E+27 \\
        10 & 5.43E+54    & 8.29E+55    & 8.46E+55     & 2.42E+19     & $\infty$    & 5.05E+44 \\
        15 & 2.24E+94    & 5.45E+95    & 2.00E+96     & 2.36E+43     & $\infty$    & 2.93E+58 \\
        20 & 1.62E+128   & 3.01E+129   & 4.63E+132    & 1.29E+62     & $\infty$    & 1.68E+67 \\
        \bottomrule
    \end{tabular}
    }
    \caption{Zero-sum game with large learning rates: average relative error of NE solutions with determinant and condition number of linear system matrices in dimensions 5, 10, 15, 20 derived from solving linear systems.}
    \label{tab:zero_avg_rel_errors_rt_violation}
\end{table}

\subsubsection{Coordination Games}\label{subsubsec:standard_coordination}

Like zero-sum games with small learning rates, the system of equations describing the set of Nash equilibria has an approximately singular matrix. 
As a result, without regularization, our proposed approach performs poorly. 
The experiments revealed no new observations that could be useful for improving the implementation of our model, and we defer the table of results to Table~\ref{tab:coord_avg_rel_errors_rt} in Appendix \ref{app:tables}.

\subsection{Improving Stability with the Least Squares Method}\label{subsec:least_squares_method}

To address the numerical instability caused by near-singular systems in higher dimensions, we explore an alternative approach by minimizing the total squared violation. 
Instead of solving $Ax = b$ exactly, we formulate and solve a least squares problem that seeks a solution minimizing the aggregate deviation from the fixed-point condition $h^t - h^{t+1} = 0$. 
This relaxation improves numerical stability and can yield better NE approximations, especially in higher-dimensional settings where the linear system becomes ill-conditioned.

\subsubsection{Zero-sum Games with Small Learning Rates}
\label{subsubsec:least_squares_zero_non_violated}

Applying the least squares method stabilizes the system. 
In Tables~\ref{tab:zero_avg_rel_errors_ls_non_violation} and \ref{tab:zero_nonviolation_least_squares_rel_errors_rt_high_dim}, relative errors for the model from Section \ref{sec:ModelForSimulation} remain small even in high dimensions.
Although relative errors for the models from Sections \ref{sec:EconomicModel} and \ref{sec:FlModel} show higher variance, the error size decreases dramatically compared to the standard solution method in Section \ref{subsubsec:standard_zero_non_violated}.
The complete results of Table \ref{tab:zero_avg_rel_errors_ls_non_violation} across dimensions 1 to 20 are provided in Appendix \ref{app:tables}.

\begin{table}[!ht]
    \centering
    \begin{tabular}{c|cc|c|c}
        \toprule
        \textbf{Model:} & \multicolumn{2}{c|}{\textbf{Section \ref{sec:ModelForSimulation}}} & \textbf{Section} \ref{sec:EconomicModel} & \textbf{Section} \ref{sec:FlModel} \\
        \hline
        \textbf{Dimension} & \textbf{Avg Rel Err of $\boldsymbol{x^{*}}$} & \textbf{Avg Rel Err of $\boldsymbol{y^*}$} & \textbf{Avg Rel Err of $\boldsymbol{x^{*}}$} & \textbf{Avg Rel Err of $\boldsymbol{x^{*}}$} \\
        \hline
        \midrule
        5  & 3.06E-01 & 2.59E-01 & 2.79E+00 & 7.18E-05 \\
        10 & 8.56E-01 & 8.69E-01 & 1.87E+00 & 8.83E-01 \\
        15 & 8.96E-01 & 9.05E-01 & 3.25E+00 & 4.30E+00 \\
        20 & 9.20E-01 & 9.09E-01 & 2.62E+00 & 2.79E+00 \\
        \bottomrule
    \end{tabular}
    \caption{Zero-sum game with small learning rates: average relative error of NE solutions via the least squares method in dimensions 5, 10, 15, 20.}
    \label{tab:zero_avg_rel_errors_ls_non_violation}
\end{table}

\begin{table}[!ht]
    \centering
    \begin{tabular}{c|cc|c|c}
        \toprule
        \textbf{Model:} & \multicolumn{2}{c}{\textbf{Section \ref{sec:ModelForSimulation}}} & \textbf{Section \ref{sec:EconomicModel}} & \textbf{Section \ref{sec:FlModel}} \\
        \hline
        \textbf{Dimension} & \textbf{Avg Rel Err of $\boldsymbol{x^{*}}$} & \textbf{Avg Rel Err in $\boldsymbol{y^*}$} & \textbf{Avg Rel Err in $\boldsymbol{x^{*}}$} & \textbf{Avg Rel Err in $\boldsymbol{x^{*}}$} \\
        \hline
        \midrule    
        100 & 1.00E+00 & 9.51E-01 & 9.97E-01 & 9.98E-01 \\
        200 & 9.99E-01 & 1.00E+00 & 3.21E+01 & 1.28E+00 \\
        300 & 9.96E-01 & 9.93E-01 & 1.00E+00 & 1.01E+00 \\
        400 & 9.98E-01 & 9.77E-01 & 9.97E-01 & 9.98E-01 \\
        500 & 1.00E+00 & 1.00E+00 & 3.16E+03 & 9.96E+01 \\
        \bottomrule
    \end{tabular}
    \caption{Zero-sum games with small learning rates: relative error of the set of NE solutions and determinant of linear system matrices for one instance each in dimensions 100, 200, 300, 400, and 500 via the least squares method.}
    \label{tab:zero_nonviolation_least_squares_rel_errors_rt_high_dim}
\end{table}

\subsubsection{Zero-sum Games with Large Learning Rates}\label{subsubsec:least_squares_zero_violated}

While the least squares method avoids the issue of computing the inverse, the source of numerical instability from Section \ref{subsubsec:standard_zero_violated} remains when working with large learning rates. 
For a full table of results, see Table \ref{tab:zero_large_avg_rel_errors_ls} in Appendix \ref{app:tables}.

\subsubsection{Coordination Games}
\label{subsubsec:least_squares_coordination}

The results are consistent with the results for zero-sum games with small learning rates; 
the least squares method significantly improves numerical stability in coordination games compared to the standard solution method in Section \ref{subsubsec:standard_coordination}.
For full results, see Table~\ref{tab:coord_avg_rel_errors_ls} in Appendix \ref{app:tables}.


\subsection{Tikhonov Regularizer (Ridge Regression) Method}\label{subsec:tikhonov_reg_method}

Another common method to address ill-conditioned matrices in higher dimensions is to use the Tikhonov regularization \cite{calvetti2000tikhonov}.
By introducing Tikhonov regularizer $\lambda >0$, we solve new linear systems, specifically, $\boldsymbol{(A^{\intercal}A + \lambda I) x = A^{\intercal}b}$. 
New system matrices $A^{\intercal}A + \lambda I$ are now better conditioned, which yields a more stabilized and better approximate set of NE solutions.

In all instances, Tikhonov regularization performs similarly to the least squares methods --- our approach performs well in both zero-sum games with small learning rates and coordination games but still has instability issues for zero-sum games with large learning rates. 
See Appendix \ref{app:tables} for a full list of results. 


\subsection{Discussion}

The results of the base experiments uncover critical issues. 
While the standard method provides exact solutions theoretically, it is highly sensitive to dimension, quickly deteriorating in capturing the set of NE. 
When learning rates are sufficiently small, the least squares and Tikhonov regularizer methods mitigate the issue, but there still exists a small approximation error. 
We suspect this is caused by the similarity of consecutively generated equations -- since learning rates must be small in higher dimensions, the strategies don't change much from iteration to iteration, resulting in similar equations. 
In the next section, we introduce parallelization with multiple initial conditions to avoid redundant equations. 

When learning rates are large, the condition matrix grows rapidly and errors propagate due to rounding large, divergent strategies thereby deteriorating our approximation. 
We also introduce parallelization to keep the strategies bounded in a relatively small region.

\section{Improved Estimates via Parallelization}\label{sec:Improved_Experiments_Parallelization}

To address the aforementioned issues, we introduce a parallelizable approach that removes the correlation among equations. 
We begin by observing that the equations generated by our theoretical results only require observations from two consecutive iterations of \ref{eqn:AltGD}. 
To generate $2k$ valid equations, we could start with multiple initial conditions and run \ref{eqn:AltGD} in parallel to decrease the overall number of updates to a single initial condition, e.g., we could generate $1k$ valid equations by applying \ref{eqn:AltGD} to the initial strategies $(x,y)$ and the second set of $1k$ valid equations from a different set of initial strategies $(x',y')$. 
This means that the Nash equilibrium (NE) can be characterized in parallel, a first in algorithmic game theory. 

More importantly, our approach addresses the two primary obstacles in the previous section. 
First, by collecting only a single equation from an initial strategy, we avoid correlation between consecutive equations caused by consecutive strategies being close together. 
Second, we also resolve the issue with large learning rates. 
By updating a strategy only once, strategies remain relatively small despite their tendency to diverge from the set of NE in the long run when learning rates are large. 

The updates from random initializations can be computed independently and simultaneously, enabling scalable parallelized computation. 
By eliminating equation correlation, the resulting system matrices become well-conditioned across dimensions. 

\subsection{Experiment Setup}

Our settings for the experiments remain unchanged from those in the previous section. 
However, instead of starting with a single initial strategy and updating it $2k, 2k+1$, or $k$ times depending on the solution method, we instead start with $2k$ or $k$ initial strategies and update each of them to generate a single equation. 
Each initial strategy generates a single equation that describes the set of NE.

Our new set of experiments confirms that we can quickly obtain extremely accurate estimates of the set of NE in a finite number of iterations. 
We also confirm this by repeating our experiments in much larger games (up to 500 strategies).
We also compare our results to standard methods.
We remark that we do not compare our results with time-average guarantees of alternating gradient descent and not optimistic gradient descent --- \cite{bailey2021left} experimentally shows that alternating gradient descent returns approximations that are $\approx 2.35$ times closer to the set of NE than optimistic gradient descent in unconstrained zero-sum games.

\subsection{Zero-sum Games with Small Learning Rates}

Using new initial conditions for each generated equation resolves the issue of the vanishing determinant, and we can accurately compute the Nash equilibrium. 
Table \ref{tab:zero_par_nonviolation_avg_rel_errors_rt} shows that the relative errors remain near $10^{-11}$ up to dimension 20. 
The results are consistent across all instances for each dimension.
Moreover, we maintain accurate predictions even in higher dimensions; 
Table \ref{tab:zero_nonviolation_parallel_rel_errors_rt_high_dim} shows that the relative errors remain extremely low when the strategy space is 500-dimensional.
\begin{table}[!ht]
    \centering
    \begin{tabular}{c|cc|c|c|c}
        \toprule
        \textbf{Model:} & \multicolumn{2}{c}{\textbf{Section \ref{sec:ModelForSimulation}}} & \textbf{Section} \ref{sec:EconomicModel} & \textbf{Section} \ref{sec:FlModel} & \\
        \hline
        \textbf{Dimension} & \textbf{Avg Rel Err of $\boldsymbol{x^{*}}$} & \textbf{Avg Rel Err in $\boldsymbol{y^*}$} & \textbf{Avg Rel Err in $\boldsymbol{x^{*}}$} & \textbf{Avg Rel Err in $\boldsymbol{x^{*}}$} & $\boldsymbol{\textbf{det} (A_{par}})$ \\
        \hline
        \midrule
        5  & 3.56E-13 & 3.60E-13 & 3.56E-13 & 3.56E-13 & 4.71E+01 \\
        10 & 1.68E-12 & 1.78E-12 & 1.68E-12 & 1.68E-12 & 3.33E+06 \\
        15 & 7.44E-11 & 1.36E-10 & 7.44E-11 & 7.44E-11 & 7.62E+14 \\
        20 & 3.02E-11 & 2.58E-11 & 3.02E-11 & 3.02E-11 & 7.02E+23 \\
        \bottomrule
    \end{tabular}
    \caption{Zero-sum games with small learning rates: relative errors of NE solutions and average determinants of linear system matrices in dimensions 5, 10, 15, 20, captured by solving $A_{system}x = b$ in zero-sum games via a parallelization method.}
    \label{tab:zero_par_nonviolation_avg_rel_errors_rt}
\end{table}

\begin{table}[!ht]
    \centering
    \resizebox{1\textwidth}{!}{%
    \begin{tabular}{c|cc|c|c|c}
        \toprule
        \textbf{Model:} & \multicolumn{2}{c}{\textbf{Section \ref{sec:ModelForSimulation}}} & \textbf{Section \ref{sec:EconomicModel}} & \textbf{Section \ref{sec:FlModel}} & \\
        \hline
        \textbf{Dimension} & \textbf{Avg Rel Err of $\boldsymbol{x^{*}}$} & \textbf{Avg Rel Err in $\boldsymbol{y^*}$} & \textbf{Avg Rel Err in $\boldsymbol{x^{*}}$} & \textbf{Avg Rel Err in $\boldsymbol{x^{*}}$} & $\boldsymbol{\textbf{det} (A_{system_1}})$ \\
        \hline
        \midrule    
        100 & 8.11E-10 & 1.41E-09 & 8.11E-10 & 1.87E-05 & 1.38E+249 \\
        200 & 3.62E-10 & 3.34E-10 & 3.62E-10 & 2.28E-04 & $\infty$ \\
        300 & 4.54E-08 & 4.96E-08 & 4.54E-08 & 9.42E-05 & $\infty$ \\
        400 & 1.40E-08 & 2.54E-08 & 1.40E-08 & 6.96E-05 & $\infty$ \\
        500 & 2.19E-07 & 2.24E-07 & 2.19E-07 & 3.18E-02 & $\infty$ \\
        \bottomrule
    \end{tabular}}
    \caption{Zero-sum games with small learning rates: relative error of NE solutions and determinant of linear system matrices for one instance each in dimensions 100, 200, 300, 400, and 500 via parallelization method.}
    \label{tab:zero_nonviolation_parallel_rel_errors_rt_high_dim}
\end{table}

Finally, we compare our proposed method with standard techniques for approximating NE via time-average convergence. 
Even in dimension 7, as depicted in Table \ref{tab:zero_nonviolation_num_iters_updated_comparison_time_avg_few_instances}, 1000s of iterations are required to get a decent approximation (within $.1\%=10^{-3}$) of the Nash equilibrium using standard methods. 
Our proposed model obtains a much higher level of precision ($10^{-13}$) while completing only 14 iterations of \ref{eqn:AltGD}.  
On the other hand, the time-average convergence method via \ref{eqn:AltGD} requires more computation time, achieving worse relative errors than our parallelization method.

\begin{table}[!ht]
\centering
\resizebox{1\textwidth}{!}{%
\begin{tabular}{c|*{5}{>{\raggedleft\arraybackslash}p{2.5cm}}|c}
    \toprule
    & \multicolumn{5}{c|}{\textbf{Iterations Required for Time-Average Convergence to First Reach Given Relative Error}} & \\ 
    \textbf{Instance} & \textbf{10\%} & \textbf{5\%} & \textbf{1\%} & \textbf{0.5\%} & \textbf{0.1\%} & \textbf{Parallelization} \\
    \hline
    \midrule
    1  & 3073 & 3241 & 6842 & 10208 & 34148 & 14 \\
    2  & 454  & 497  & 1786 & 3067  & 10535 & 14 \\
    3  & 1633 & 1678 & 6967 & 8713  & 24466 & 14 \\
    4  & 425  & 569  & 2365 & 3683  & 10042 & 14 \\
    5  & 133  & 548  & 1390 & 2233  & 5164  & 14 \\
    $\vdots$  & $\vdots$  & $\vdots$  & $\vdots$ & $\vdots$  & $\vdots$ & $\vdots$ \\
    30 & 509  & 532  & 802  & 2418  & 8055  & 14 \\
    \hline
    \midrule
    Average & 1856 & 2055 & 3225 & 6161 & 19812 & 14 \\
    Max     & 11056 & 11145 & 11254 & 56580 & 243743 & 14 \\
    Min     & 133 & 163 & 802 & 828 & 2524 & 14 \\
    \bottomrule
\end{tabular}}
\caption{Time-average convergence vs. model in Section \ref{sec:ModelForSimulation} via parallelization method: the number of iterations needed to approximate NE in dimension 7.}
\label{tab:zero_nonviolation_num_iters_updated_comparison_time_avg_few_instances}
\end{table}

We also compared the performance when \ref{eqn:AltGD} runs for a fixed time.
This allows a more accurate comparison, as our proposed method requires additional time to approximate $A'x^*=b'$ using the information collected from gradients. 
When our parallelization method is applied to solve the model in Section \ref{sec:ModelForSimulation} --- including solving $A'x^*=b'$ --- the empirical runtime is at most 0.000446 seconds, while achieving a higher precision of $10^{-10}$, as shown in Table \ref{tab:zero_par_nonviolation_avg_rel_errors_rt}.
In contrast, as shown in Table \ref{tab:time_avg_convergence_dim_7}, standard time-average convergence results only reach a precision level of approximately $10^{-6}$ after running for 10 minutes in dimension 7. 
As shown in Table \ref{tab:zero_avg_rel_errors_time_average_non_violation_time_minutes} continues to drastically outperform standard convergence guarantees even in higher dimensions. 
We note that in Tables \ref{tab:time_avg_convergence_dim_7} and \ref{tab:zero_avg_rel_errors_time_average_non_violation_time_minutes}, there are cases when (average) relative error slightly increases after longer \ref{eqn:AltGD} updates. 
It is fairly natural in \ref{eqn:AltGD} since the time-average trajectory oscillates around the NE (see Figure \ref{fig:spiral_trajectory_of_time_average}). 
The complete results of Tables \ref{tab:zero_par_nonviolation_avg_rel_errors_rt} and \ref{tab:time_avg_convergence_dim_7} in dimensions 1 through 20 and \ref{tab:zero_nonviolation_num_iters_updated_comparison_time_avg_few_instances} across 30 instances are provided in Appendix \ref{app:tables2}.

\begin{table}[!ht]
\centering
\begin{tabular}{c|*{4}{>{\raggedleft\arraybackslash}p{2cm}}}
    \toprule
    & \multicolumn{4}{c}{\textbf{Rel Err of $\boldsymbol{x^{*}}$ via Time-Average Convergence After}} \\ 
    \textbf{Instance} & \textbf{1 minute} & \textbf{3 minutes} & \textbf{5 minutes} & \textbf{10 minutes}\\
    \hline
    \midrule
    1  & 2.32E-04 & 8.30E-05 & 2.58E-06 & 3.88E-06 \\
    2  & 1.38E-06 & 4.83E-07 & 3.28E-07 & 1.10E-07 \\
    3  & 1.81E-05 &  1.56E-06 &  5.31E-06 & 2.12E-06 \\
    4  & 3.59E-06 & 1.69E-06 & 9.02E-07 & 3.10E-07 \\
    5  & 3.11E-06 & 7.11E-07 & 3.99E-07 & 1.20E-07 \\
    $\vdots$  & $\vdots$  & $\vdots$  & $\vdots$ & $\vdots$  \\
    30 & 5.73E-06 & 2.07E-06 &  4.82E-07 &  6.38E-07 \\ 
    \hline
    \midrule
    Average & 8.08E-05 & 2.93E-05 & 4.72E-06 & 7.23E-06 \\ 
    Max     & 1.38E-06 & 4.83E-07 & 1.02E-07 & 1.10E-07 \\ 
    Min     & 1.54E-03 & 5.79E-04	& 4.19E-05 & 1.63E-04 \\ 
    \bottomrule
\end{tabular}
\caption{Zero-sum games with small learning rates: average relative error of time-averaged NE solutions in dimension 7 via \ref{eqn:AltGD} for 1, 3, 5, and 10 minutes.}
\label{tab:time_avg_convergence_dim_7}
\end{table}

\begin{table}[!ht]
    \centering
    \begin{tabular}{c|p{3cm}<{\centering}p{3cm}<{\centering}p{3cm}<{\centering}|c}
        \toprule
        & \multicolumn{3}{c|}{\textbf{Avg Rel Err of $\boldsymbol{x^*}$ via Time-Average Convergence}}
        & \textbf{Time for Parallelization}\\
        \hline
        \textbf{Dimension} & \textbf{1 minute} & \textbf{3 minutes} & \textbf{5 minutes} & \textbf{($\approx 10^{-10}$ Avg Rel Err)}\\
        \hline
        \midrule
        5 &  8.85E-05 & 9.25E-05 & 6.11E-05 & 0.00128 seconds\\
        10 & 4.45E-04 & 1.00E-04 & 5.16E-05 & 0.00194 seconds\\
        15 & 1.50E-04 & 1.05E-04 & 4.90E-05 & 0.00446 seconds\\
        20 & 3.86E-04 & 1.02E-04 & 5.91E-05 & 0.00214 seconds\\
        \bottomrule
    \end{tabular}
    \caption{Zero-sum games with small learning rates: average relative error of time-averaged NE solutions via \ref{eqn:AltGD} for 1, 3, and 5 minutes, and elapsed time of parallelization method achieving $\approx 10^{-10}$ precision.}
    \label{tab:zero_avg_rel_errors_time_average_non_violation_time_minutes}
\end{table}

\def\sizer{.65}
\begin{figure}[!ht]
  \begin{subfigure}[b]{0.45\linewidth}
    \centering
    \includegraphics[width=\sizer\linewidth]{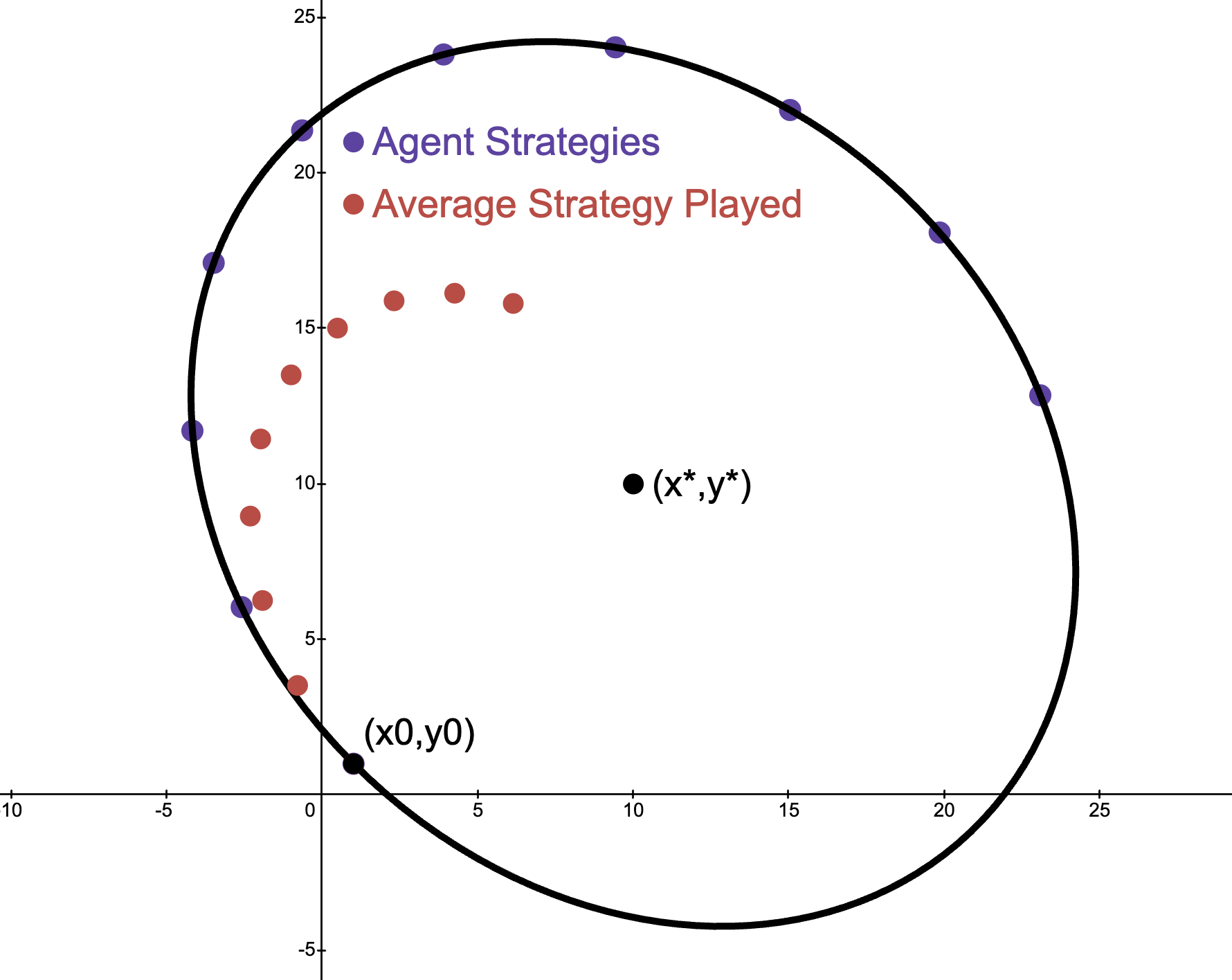} 
    \caption{Time-average trajectory in early iterations.} 
    \label{fig7:a} 
    \vspace{4ex}
  \end{subfigure}
  \begin{subfigure}[b]{0.45\linewidth}
    \centering
    \includegraphics[width=\sizer\linewidth]{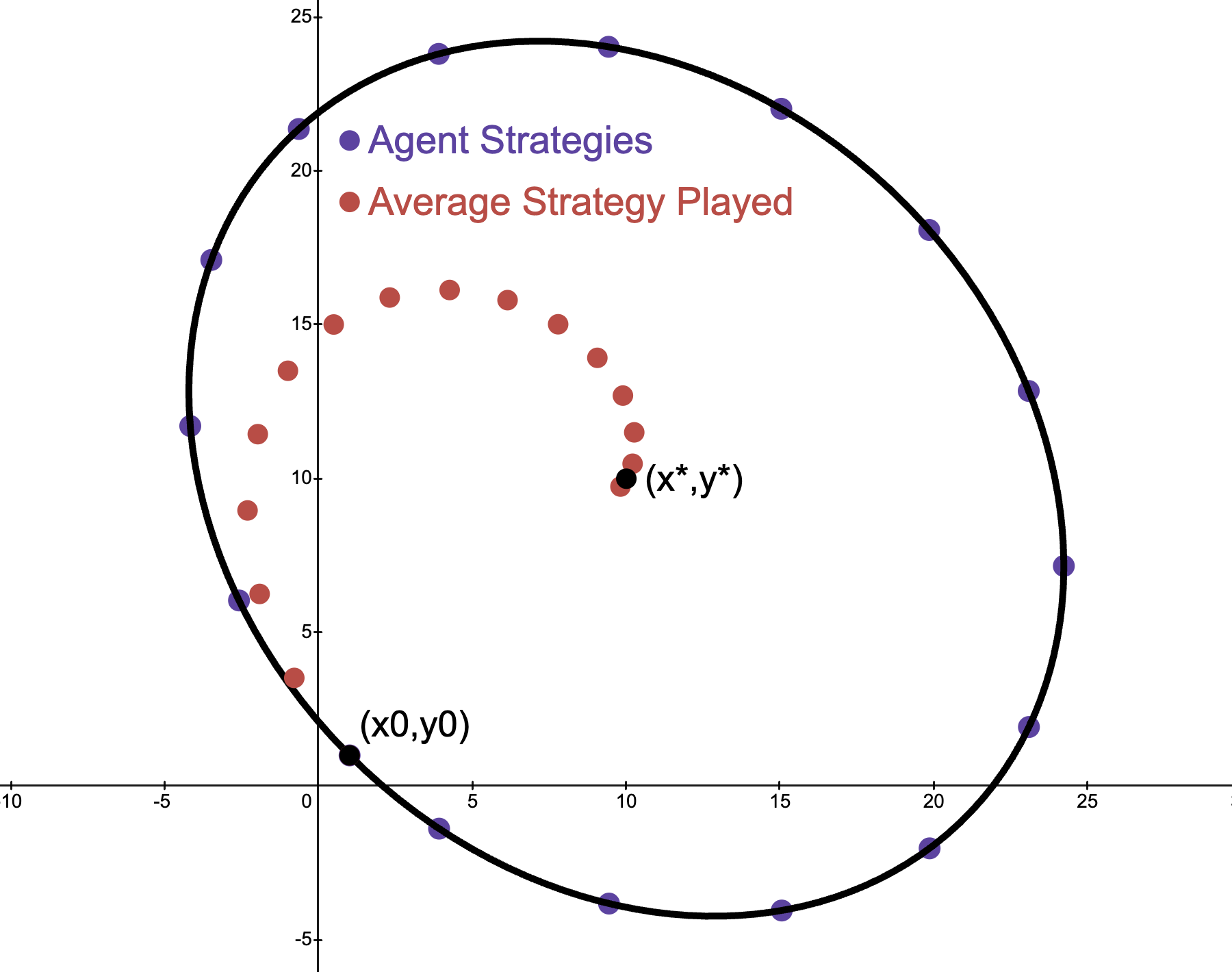} 
    \caption{First highly accurate approximation of the NE.} 
    \label{fig7:b} 
    \vspace{4ex}
  \end{subfigure} 
  \begin{subfigure}[b]{0.45\linewidth}
    \centering
    \includegraphics[width=\sizer\linewidth]{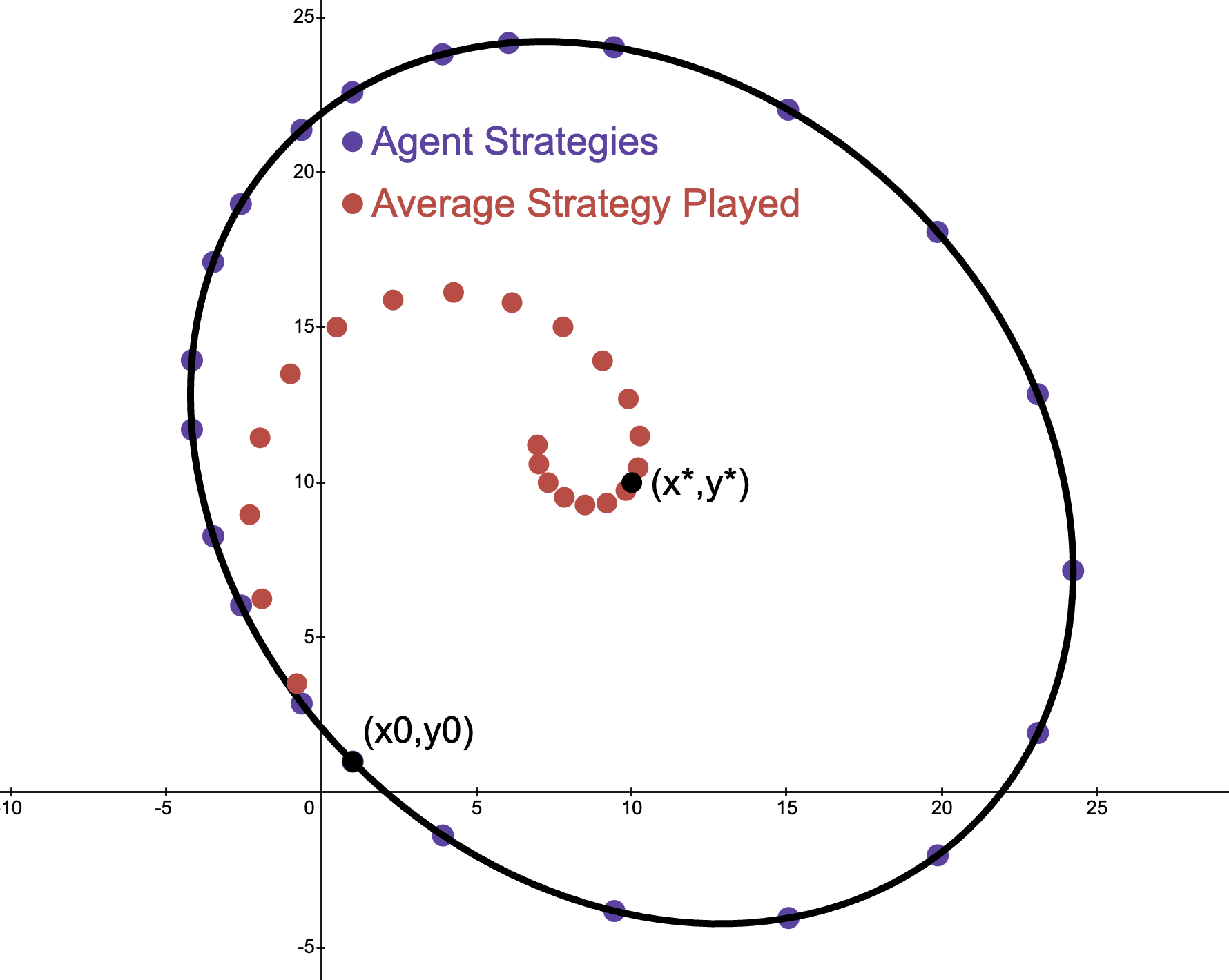} 
    \caption{Oscillatory deviation around the NE.} 
    \label{fig7:c} 
  \end{subfigure}
  \begin{subfigure}[b]{0.45\linewidth}
    \centering
    \includegraphics[width=\sizer\linewidth]{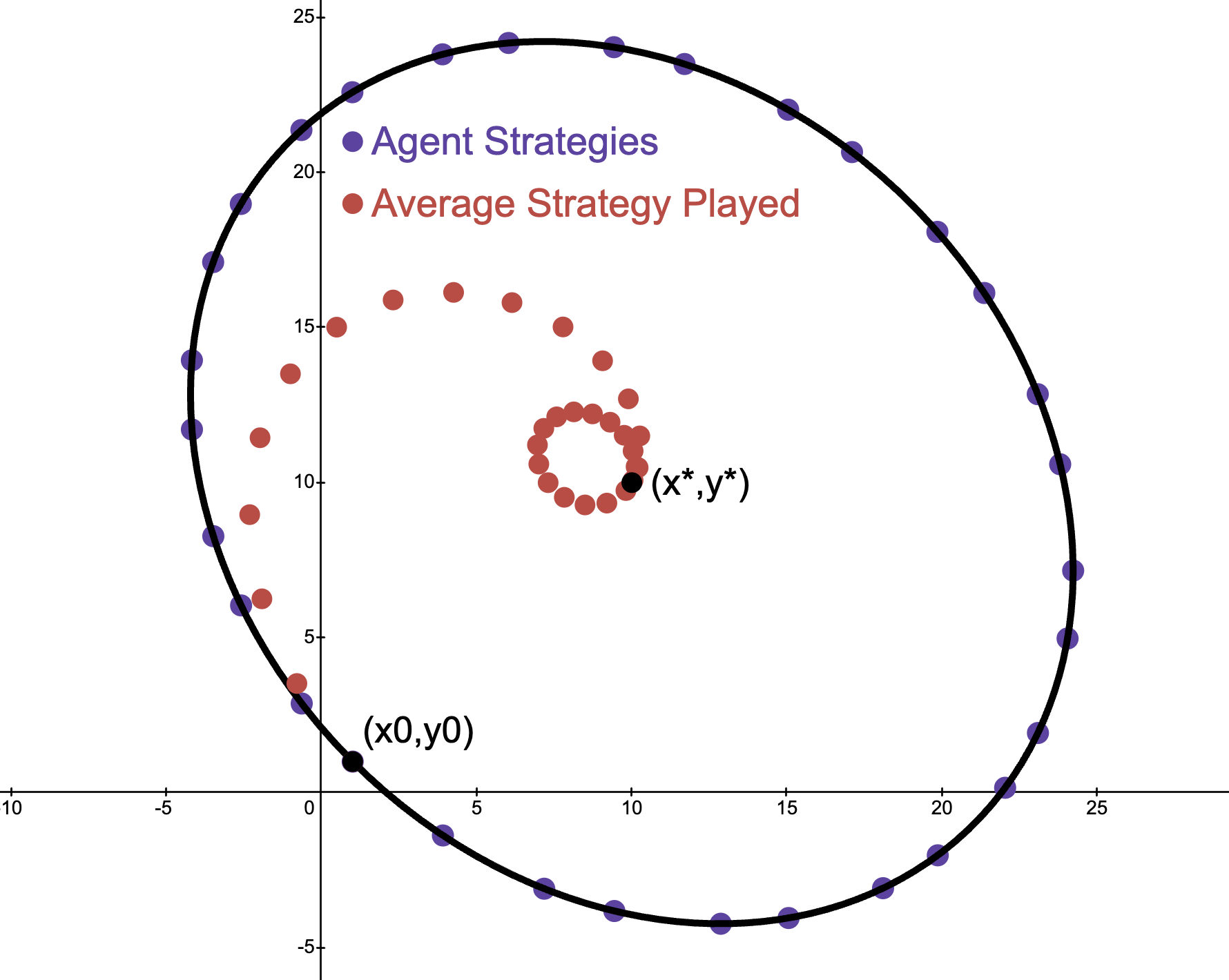} 
    \caption{Second highly accurate approximation of NE.} 
    \label{fig7:d} 
  \end{subfigure} 
  \caption{Spiral trajectory of time-average strategies towards the NE via \ref{eqn:AltGD}}
  \label{fig:spiral_trajectory_of_time_average} 
\end{figure}

\subsubsection{Zero-sum Games with Large Learning Rates and Coordination Games}

As shown in Tables \ref{tab:zero_avg_rel_errors_parallel_violation}, \ref{tab:zero_parallel_violated_rel_errors_dim18}, 
\ref{tab:coord_parallel_avg_rel_errors}, and \ref{tab:coord_parallel_rel_errors_dim20} 
in Appendix \ref{app:tables2}, our results remain consistent in both zero-sum games with large learning rates and in coordination games. 
The result in zero-sum games with large learning rates is especially notable -- every existing method requires learning to be sufficiently small to find the Nash equilibrium. 

\subsection{Discussion}

Overall, our parallelization method fundamentally changes the structure of linear systems by eliminating correlation between equations. 
The result is a massive improvement in numerical conditioning, leading to low error and consistent performance across all dimensions tested. 
It removes the need for tuning learning rates or carefully regularizing the system with a Tikhonov regularizer, and reduces the number of updates needed to approximate NE compared to standard methods. 
In contrast to the base methods in Section \ref{sec:Base_Experiments}, the parallelization method consistently approximates the set of NE with high precision while demonstrating improvements in stability and reliability.

\section{Conclusion}

We develop a new method for finding the set of Nash equilibria (NE) that only requires a finite number of inquiries from an online learning algorithm, a first in algorithmic game theory.
Further, we prove that our method can be parallelized and works with arbitrary learning rates, two additional firsts. 
We also demonstrate that the parallelized version of our approach is numerically stable; 
Our proposed method runs faster than the traditional time-average convergence algorithm, while approximating NE with higher accuracy. 
Specifically, the parallelization method achieves a precision of $\approx 10^{-10}$ in less than a second across dimensions 1 through 20, outperforming the time-average convergence approach, which reaches only $\approx 10^{-6}$ precision in 5 minutes. 
Further, as discussed in Section \ref{sec:UnboundedMotivation}, our approach can be applied to games with probability vectors once the support of the Nash equilibrium is known.

\bibliographystyle{plain}
\bibliography{References}

\appendix

\section{Proof of Theorem \ref{thm:zero_invariant}}\label{sec:zero_invariant}

\ZeroInvariant*

\begin{proof}
    By Lemma \ref{lem:zero_partial}, we have 
    \begin{align*}
        \frac{\| y^{t+1} - y^{*}\|^{2} - \| y^{t} - y^{*} \|^{2}}{{\eta_2}} 
        = 
        \left\langle x^{t+1} + x^{t} - 2x^{*}, \ A(y^{t}-y^{*}) \right\rangle
    \end{align*}
    Expanding the right-hand side yields
    \begin{align*}
        &\left\langle x^{t+1} + x^{t} - 2x^{*}, \ A(y^{t}-y^{*}) \right\rangle \\
        = &\left\langle x^{t+1}, \ Ay^{t} \right\rangle - \left\langle x^{t+1}, \ Ay^{*} \right\rangle + \left\langle x^{t} \ Ay^{t} \right\rangle - \left\langle x^{t} \ Ay^{*} \right\rangle -2 \left\langle x^{*} \ Ay^{t} \right\rangle + 2 \left\langle x^{*}, \ Ay^{*} \right\rangle \\
        = &\left\langle x^{t+1}, \ Ay^{t} \right\rangle - \left\langle x^{t+1}, \ b_1 \right\rangle + \left\langle x^{t}, \ Ay^{t} \right\rangle - \left\langle x^{t}, \ b_1 \right\rangle - 2\left\langle y^{t}, \ A^{\intercal}x^{*} \right\rangle + 2\left\langle y^*, \ A^{\intercal}x^{*} \right\rangle \\
        = &\left\langle x^{t+1}, \ Ay^{t} \right\rangle - \left\langle x^{t+1}, \ b_1 \right\rangle + \left\langle x^{t}, \ Ay^{t} \right\rangle - \left\langle x^{t}, \ b_1 \right\rangle + 2\left\langle y^{t}, \ b_2 \right\rangle + 2\left\langle y^*, \ A^{\intercal}x^{*} \right\rangle
    \end{align*}
    Therefore, 
    \begin{align*}
        &\frac{\| x^{t+1} - x^{*}\|^{2} - \| x^{t} - x^{*} \|^{2}}{{\eta_1}} \\
        =
        &\left\langle x^{t+1}, \ Ay^{t} \right\rangle - \left\langle x^{t+1}, \ b_1 \right\rangle + \left\langle x^{t}, \ Ay^{t} \right\rangle - \left\langle x^{t}, \ b_1 \right\rangle + 2\left\langle y^{t}, \ b_2 \right\rangle + 2\left\langle y^*, \ A^{\intercal}x^{*} \right\rangle
    \end{align*}
    By following the same procedure identically, we obtain
    \begin{align*}
        &\frac{\| y^{t+1} - y^{*}\|^{2} - \| y^{t} - y^{*} \|^{2}}{{\eta_2}} \\
        =
        &-\left\langle x^{t+1}, \ Ay^{t+1} \right\rangle - \left\langle y^{t+1}, \ b_2 \right\rangle - \left\langle x^{t+1}, \ Ay^{t} \right\rangle - \left\langle y^{t}, \ b_2 \right\rangle + 2\left\langle x^{t+1}, \ b_1 \right\rangle - 2\left\langle y^{*}, \ A^{\intercal}x^{*} \right\rangle
    \end{align*}
    By adding the two equations above, we obtain
    \begin{align*}
        &\frac{\| x^{t+1} - x^{*} \|^{2} - \| x^{t} - x^{*} \|^{2}}{{\eta_1}} + \frac{\| y^{t+1} - y^{*}\|^{2} - \| y^{t} - y^{*} \|^{2}}{{\eta_2}} \\ 
        = &\left\langle x^{t+1}, \ Ay^{t} \right\rangle - \left\langle x^{t+1}, \ b_1 \right\rangle + \left\langle x^{t}, \ Ay^{t} \right\rangle - \left\langle x^{t}, \ b_1 \right\rangle + 2\left\langle y^{t}, \ b_2 \right\rangle + 2\left\langle y^*, \ A^{\intercal}x^{*} \right\rangle \\ 
        &-\left\langle x^{t+1}, \ Ay^{t+1} \right\rangle - \left\langle y^{t+1}, \ b_2 \right\rangle - \left\langle x^{t+1}, \ Ay^{t} \right\rangle - \left\langle y^{t}, \ b_2 \right\rangle + 2\left\langle x^{t+1}, \ b_1 \right\rangle - 2\left\langle y^{*}, \ A^{\intercal}x^{*} \right\rangle \\
        = &\left\langle x^{t}, \ Ay^{t}-b_1 \right\rangle - \left\langle x^{t+1}, \ Ay^{t+1}-b_1 \right\rangle + \left\langle y^{t}, \ b_2 \right\rangle - \left\langle y^{t+1}, \ b_2 \right\rangle
    \end{align*}
    Reorganizing the equation with all $t+1$ terms on one side and $t$ terms on the other yields
    \begin{align*}
         &\frac{\| x^{t+1}-x^* \|^{2}}{\eta_1} + \frac{\| y^{t+1}-y^* \|^{2}}{\eta_2} + \left\langle x^{t+1}, \ Ay^{t+1}-b_1 \right\rangle + \left\langle y^{t+1}, \ b_2 \right\rangle \\ = &\frac{\| x^{t}-x^* \|^{2}}{\eta_1} + \frac{\| y^{t}-y^* \|^{2}}{\eta_2} + \left\langle x^{t}, \ Ay^{t}-b_1 \right\rangle + \left\langle y^{t}, \ b_2 \right\rangle
    \end{align*}
    By Definition \ref{def:Zero_Energy}, this implies
    \begin{align*}
         h_{-}^{t+1} = h_-^{t}
    \end{align*}
    Since it holds for all $t\in \mathbb{Z}_{\geq 0}$, it follows that $h_-^{t} = h_-^{0}$.
    Hence, the energy function $h_-^{t}$ is time-invariant.
\end{proof}

\section{Proof of Theorem \ref{thm:coord_invariant}}\label{sec:coord_invariant}

\CoordInvariant*

\begin{proof}
Expanding the right-hand side yields
\begin{align*}
    &\left\langle x^{t+1} + x^{t} - 2x^{*}, \ A(y^{t}-y^{*}) \right\rangle \\
    = 
    &\left\langle x^{t+1}, \ Ay^{t} \right\rangle - \left\langle x^{t+1}, \ Ay^{*} \right\rangle + \left\langle x^{t} \ Ay^{t} \right\rangle - \left\langle x^{t} \ Ay^{*} \right\rangle -2 \left\langle x^{*} \ Ay^{t} \right\rangle + 2 \left\langle x^{*}, \ Ay^{*} \right\rangle \\
    = 
    &\left\langle x^{t+1}, \ Ay^{t} \right\rangle - \left\langle x^{t+1}, \ b_1 \right\rangle + \left\langle x^{t}, \ Ay^{t} \right\rangle - \left\langle x^{t}, \ b_1 \right\rangle - 2\left\langle y^{t}, \ A^{\intercal}x^{*} \right\rangle + 2\left\langle y^*, \ A^{\intercal}x^{*} \right\rangle \\
    = 
    &\left\langle x^{t+1}, \ Ay^{t} \right\rangle - \left\langle x^{t+1}, \ b_1 \right\rangle + \left\langle x^{t}, \ Ay^{t} \right\rangle - \left\langle x^{t}, \ b_1 \right\rangle - 2\left\langle y^{t}, \ b_2 \right\rangle + 2\left\langle y^*, \ A^{\intercal}x^{*} \right\rangle \\ 
\end{align*}
Therefore, 
\begin{align*}
    &\frac{\| x^{t+1} - x^{*} \|^{2} - \| x^{t} - x^{*} \|^{2}}{{\eta_1}} \\
    =
    &\left\langle x^{t+1}, \ Ay^{t} \right\rangle - \left\langle x^{t+1}, \ b_1 \right\rangle + \left\langle x^{t}, \ Ay^{t} \right\rangle - \left\langle x^{t}, \ b_1 \right\rangle - 2\left\langle y^{t}, \ b_2 \right\rangle + 2\left\langle y^*, \ A^{\intercal}x^{*} \right\rangle
\end{align*}
By following the same procedure identically, we obtain
\begin{align*}
    &\frac{\| y^{t+1} - y^{*}\|^{2} - \| y^{t} - y^{*} \|^{2}}{{\eta_2}}\\
    =
    &\left\langle x^{t+1}, \ Ay^{t+1} \right\rangle - \left\langle y^{t+1}, \ b_2 \right\rangle + \left\langle x^{t+1}, \ Ay^{t} \right\rangle - \left\langle y^{t}, \ b_2 \right\rangle - 2\left\langle x^{t+1}, \ b_1 \right\rangle + 2\left\langle y^{*}, \ A^{\intercal}x^{*} \right\rangle
\end{align*}    
By subtracting the two equations above, we obtain  
\begin{align*}
    &\frac{\| x^{t+1} - x^{*} \|^{2} - \| x^{t} - x^{*} \|^{2}}{{\eta_1}} - \frac{\| y^{t+1} - y^{*}\|^{2} - \| y^{t} - y^{*} \|^{2}}{{\eta_2}} \\ 
    = 
    &\left\langle x^{t+1}, \ Ay^{t} \right\rangle - \left\langle x^{t+1}, \ b_1 \right\rangle + \left\langle x^{t}, \ Ay^{t} \right\rangle - \left\langle x^{t}, \ b_1 \right\rangle - 2\left\langle y^{t}, \ b_2 \right\rangle + 2\left\langle y^*, \ A^{\intercal}x^{*} \right\rangle \\ 
    &-\left\langle x^{t+1}, \ Ay^{t+1} \right\rangle + \left\langle y^{t+1}, \ b_2 \right\rangle - \left\langle x^{t+1}, \ Ay^{t} \right\rangle + \left\langle y^{t}, \ b_2 \right\rangle + 2\left\langle x^{t+1}, \ b_1 \right\rangle - 2\left\langle y^{*}, \ A^{\intercal}x^{*} \right\rangle \\
    = 
    &\left\langle x^{t}, \ Ay^{t}-b_1 \right\rangle - \left\langle x^{t+1}, \ Ay^{t+1}-b_1 \right\rangle - \left\langle y^{t}, \ b_2 \right\rangle + \left\langle y^{t+1}, \ b_2 \right\rangle
\end{align*}
Reorganizing with all $t+1$ terms on one side and $t$ terms on the other yields
\begin{align*}
     &\frac{\| x^{t+1}-x^* \|^{2}}{\eta_1} - \frac{\| y^{t+1}-y^* \|^{2}}{\eta_2} + \left\langle x^{t+1}, \ Ay^{t+1}-b_1 \right\rangle - \left\langle y^{t+1}, \ b_2 \right\rangle \\ 
     = 
     &\frac{\| x^{t}-x^* \|^{2}}{\eta_1} - \frac{\| y^{t}-y^* \|^{2}}{\eta_2} + \left\langle x^{t}, \ Ay^{t}-b_1 \right\rangle - \left\langle y^{t}, \ b_2 \right\rangle
\end{align*}
By Definition \ref{def:Coord_Energy}, this implies
\begin{align*}
     h_{+}^{t+1} = h_+^{t}
\end{align*}
Since it holds for all $t\in \mathbb{Z}_{\geq 0}$, it follows that $h_+^{t} = h_+^{0}$.
Hence, the energy function $h_+^{t}$ is time-invariant.
\end{proof}

\section{Proof of Lemma \ref{lem:coord_partial}}\label{sec:coord_partial}

\CoordPartial*

\begin{proof}
Expanding the right-hand side of the lemma yields
\begin{align*}
    \left\langle x^{t+1}+x^{t}-2x^{*}, \ A(y^{t}-y^{*}) \right\rangle 
    =
    &\left\langle x^{t+1}+x^{t}-2x^{*}, \ Ay^{t}-Ay^{*} \right\rangle \\
    =
    &\left\langle x^{t+1}+x^{t}-2x^{*}, \ Ay^{t}-b_1 \right\rangle
\end{align*}
since $Ay^{*} = b_1$ by \ref{eqn:NEconditions}.
Then, by \ref{eqn:AltGD} and using the fact that the payoff matrices satisfy $B=A^{\intercal}$ in a coordination game, we obtain
\begin{align*}
    \left\langle x^{t+1}+x^{t}-2x^{*}, \ A(y^{t}-y^{*}) \right\rangle
    =
    &\left\langle x^{t+1}+x^{t}-2x^{*}, \ Ay^{t}-b_1 \right\rangle \\
    = 
    &\left\langle x^{t+1}+x^{t}-2x^{*}, \ \frac{x^{t+1}-x^{t}}{{\eta_1}} \right\rangle \\
    = 
    &\frac{\left\langle x^{t+1}+x^{t}-2x^{*}, \ x^{t+1}-x^{t} \right\rangle}{{\eta_1}} \\
    = 
    &\frac{\| x^{t+1} \|^{2} -  \| x^{t} \|^{2} -2 \left\langle x^{*}, \ x^{t+1}-x^{t} \right\rangle}{{\eta_1}} \\
    = 
    &\frac{\| x^{t+1} \|^{2} - 2 \left\langle x^{*}, \ x^{t+1} \right\rangle + \| x^{*} \|^{2} - (\| x^{t} \|^{2} - 2 \left\langle x^{*}, \ x^{t} \right\rangle + \| x^{*} \|^{2})}{{\eta_1}} \\
    = 
    &\frac{\| x^{t+1} - x^{*}\|^{2} - \| x^{t} - x^{*} \|^{2}}{{\eta_1}}
\end{align*}
By following the same steps symmetrically for agent 2, we obtain
\begin{align*}
    \left\langle y^{t+1}+y^{t}-2y^{*}, \ A^{\intercal}(x^{t+1}-x^{*}) \right\rangle
    =
    \frac{\| y^{t+1} - y^{*}\|^{2} - \| y^{t} - y^{*} \|^{2}}{{\eta_2}}
\end{align*}
since $A^{\intercal}x^{*} = b_2$ and $A^{\intercal}x^{t+1}-b_2 = \frac{y^{t+1} - y^{t}}{{\eta_2}}$ in a coordination game, thereby yielding the statement of the lemma. 
\end{proof}

\section{Proof of Theorem \ref{thm:coord_linear}}\label{sec:coord_linear}

\CoordLinear*

\begin{proof}

By Definition \ref{def:Coord_Energy} and Theorem \ref{thm:coord_invariant}, the change in energy is expressed as
    \begin{align*}
        0= h_{+}^{t} - h_{+}^{t+1}
        = &\frac{\| x^{t} - \boldsymbol{x^{*}} \|^{2}}{{\eta_1}} - \frac{\| y^{t} - \boldsymbol{y^{*}} \|^{2}}{{\eta_2}} + \left\langle x^{t}, \ Ay^{t}-b_1 \right\rangle - \left\langle y^{t}, \ \boldsymbol{b_2} \right\rangle \\
        &- 
        \frac{\| x^{t+1} - \boldsymbol{x^{*}} \|^{2}}{{\eta_1}} + \frac{\| y^{t+1} - \boldsymbol{y^{*}} \|^{2}}{{\eta_2}} - \left\langle x^{t+1}, \ Ay^{t+1}-b_1 \right\rangle + \left\langle y^{t}, \ \boldsymbol{b_2} \right\rangle \\
        = & \frac{\lVert x^{t}-\boldsymbol{x^*} \rVert^{2}-\lVert x^{t+1}-\boldsymbol{x^*} \rVert^{2}}{\eta_1} - \frac{\lVert y^{t}-\boldsymbol{y^*} \rVert^{2}-\lVert y^{t+1}-\boldsymbol{y^*}\rVert^{2}}{\eta_2} \\
        &+ \left\langle x^{t}, \ Ay^{t}-b_1 \right\rangle - \left\langle y^{t}, \ \boldsymbol{b_2} \right\rangle - \left\langle x^{t+1}, \ Ay^{t+1}-b_1 \right\rangle + \left\langle y^{t}, \ \boldsymbol{b_2} \right\rangle
    \end{align*}

Note that the first two terms in the above equations are the change in the distance to NE of Lemma \ref{lem:coord_partial}.
By expanding these terms that have $x^*$ and $y^*$ as unknowns, we get
    \begin{align*}
        &\frac{\lVert x^{t}-\boldsymbol{x^*} \rVert^{2}-\lVert x^{t+1}-\boldsymbol{x^*} \rVert^{2}}{\eta_1} - \frac{\lVert y^{t}-\boldsymbol{y^*} \rVert^{2}-\lVert y^{t+1}-\boldsymbol{y^*}\rVert^{2}}{\eta_2}\\ 
        =
        & \frac{\| x^{t} \|^{2} - \| x^{t+1} \|^{2} -2\left\langle x^{t}-x^{t+1}, \boldsymbol{x^*} \right\rangle}{\eta_1} - \frac{\| y^{t} \|^{2} - \|y^{t+1}\|^{2} -2\left\langle y^{t}-y^{t+1}, \boldsymbol{y^*} \right\rangle}{\eta_2}
    \end{align*}
    Then, by plugging in the result into $0= h_{+}^{t} - h_{+}^{t+1}$, we have equations 
    \begin{align*}
        0= h_{+}^{t} - h_{+}^{t+1}
        = & 
        \frac{\| x^{t} \|^{2} - \| x^{t+1} \|^{2} -2\left\langle x^{t}-x^{t+1}, \boldsymbol{x^*} \right\rangle}{\eta_1} - \frac{\| y^{t} \|^{2} - \|y^{t+1}\|^{2} -2\left\langle y^{t}-y^{t+1}, \boldsymbol{y^*} \right\rangle}{\eta_2} \\
        & +
        \left\langle x^{t}, \ Ay^{t}-b_1 \right\rangle - \left\langle y^{t}, \ \boldsymbol{b_2} \right\rangle - \left\langle x^{t+1}, \ Ay^{t+1}-b_1 \right\rangle + \left\langle y^{t}, \ \boldsymbol{b_2} \right\rangle
    \end{align*}
    that are linear in $x^*$ and $y^*$. 
\end{proof}

\section{Proof of Theorem \ref{thm:coord_w/_b2}}\label{sec:coord_w/_b2}

\CoordWoB*

\begin{proof}
    By \ref{eqn:AltGD} and $B=A^{\intercal}$ in a coordination game, $\left\langle y^{t}, \ b_2 \right\rangle$ term in $h_+^{t}$ is transformed into
    \begin{align*}
        \left\langle y^{t}, \ \boldsymbol{b_2} \right\rangle 
        &= \left\langle y^{t}, \ A^{\intercal}\boldsymbol{x^{*}} \right\rangle \\
        &= \left\langle \boldsymbol{x^{*}}, \ Ay^{t} \right\rangle \\
        &= \left\langle \boldsymbol{x^{*}}, \ Ay^{t}-b_1 \right\rangle-\left\langle \boldsymbol{x^{*}}, \ \boldsymbol{b_1} \right\rangle
    \end{align*}
Following the above process identically, 
    \begin{align*}
        \left\langle y^{t+1}, \ \boldsymbol{b_2} \right\rangle 
        = 
        \left\langle \boldsymbol{x^{*}}, \ Ay^{t+1}-b_1 \right\rangle-\left\langle \boldsymbol{x^{*}}, \ \boldsymbol{b_1} \right\rangle
    \end{align*}
Theorem \ref{thm:coord_w/_b2} then follows by substituting these expressions for $\left\langle y^{t}, \ \boldsymbol{b_2} \right\rangle$ and $\left\langle y^{t+1}, \ \boldsymbol{b_2} \right\rangle$ in equations $h_+^{t}-h_+^{t+1} = 0$ from Theorem \ref{thm:coord_linear}.
\end{proof}

\section{Proof of Theorem \ref{thm:coord_less_yt_nonlinear}}\label{sec:coord_less_yt}

\CoordLessBnonlinear*

\begin{proof}
    Expanding the terms in $h_{+}^{t}-h_{+}^{t+1}=0$, which involve $y^{t}$ and $y^{t+1}$ in Theorem \ref{thm:coord_invariant}, results in the following.
    \begin{align*}
        -\frac{\| y^{t} - \boldsymbol{y^*} \|^{2}}{\eta_2} + \frac{\| y^{t+1} - \boldsymbol{y^*} \|^{2}}{\eta_2} 
        &= 
        -\frac{\| y^{t} \|^{2} - 2\left\langle y^{t}, \boldsymbol{y^*} \right\rangle + \|\boldsymbol{y^*} \|^{2}}{\eta_2} 
        + 
        \frac{\| y^{t+1} \|^{2} - 2\left\langle y^{t+1}, \boldsymbol{y^*} \right\rangle + \|\boldsymbol{y^*} \|^{2}}{\eta_2} \\
        &=\frac{-\| y^{t} \|^{2} + \| y^{t+1} \|^{2} + 2\left\langle y^{t}, \boldsymbol{y^*} \right\rangle - 2\left\langle y^{t+1}, \boldsymbol{y^*} \right\rangle}{\eta_2}
    \end{align*}

    Applying \ref{eqn:AltGD}, \ref{eqn:NEconditions}, and recalling $B=A^{\intercal}$ yields
    \begin{align*}
        \frac{2\left\langle y^{t}, \boldsymbol{y^*} \right\rangle}{\eta_2} - \frac{2\left\langle y^{t+1}, \boldsymbol{y^*} \right\rangle}{\eta_2} 
        &= \frac{2\left\langle y^{t}, \boldsymbol{y^*} \right\rangle}{\eta_2} - \frac{\left\langle 2y^{t} + 2\eta_2(A^{\intercal}x^{t+1}-b_2), \boldsymbol{y^*} \right\rangle}{\eta_2} \\
        &= \frac{2\left\langle y^{t}, \boldsymbol{y^*} \right\rangle}{\eta_2} - \frac{2\left\langle y^{t}, \boldsymbol{y^*} \right\rangle}{\eta_2} - 2\frac{\left\langle \eta_2(A^{\intercal}x^{t+1}-b_2), \boldsymbol{y^*} \right\rangle}{\eta_2} \\
        &= -2\left\langle A^{\intercal}x^{t+1}, \boldsymbol{y^*} \right\rangle +2\left\langle \boldsymbol{b_2}, \boldsymbol{y^*} \right\rangle \\
        &= -2\left\langle x^{t+1}, A\boldsymbol{y^*} \right\rangle +2\left\langle A^{\intercal}\boldsymbol{x^{*}}, \boldsymbol{y^*} \right\rangle \\
        &= -2\left\langle x^{t+1}, \boldsymbol{b_1} \right\rangle + 2\left\langle \boldsymbol{x^{*}}, A\boldsymbol{y^*} \right\rangle \\
        &= -2\left\langle x^{t+1}, \boldsymbol{b_1} \right\rangle + 2\left\langle \boldsymbol{x^{*}}, \boldsymbol{b_1} \right\rangle
    \end{align*}
    implying 
    \begin{align*}
        -\frac{\| y^{t}-\boldsymbol{y^*} \|^{2}}{\eta_2} + \frac{\| y^{t+1}-\boldsymbol{y^*} \|^{2}}{\eta_2} 
        &= 
        -\frac{\| y^{t} \|^{2} - \| y^{t+1} \|^{2}}{\eta_2}-2\left\langle x^{t+1}, \boldsymbol{b_1} \right\rangle +2\left\langle \boldsymbol{x^{*}}, \boldsymbol{b_1} \right\rangle
    \end{align*}

    Therefore, 
    \begin{align*}
        0 = &h_{+}^{t} - h_{+}^{t+1} \\
        = 
        &\frac{\| x^{t} \|^{2} - \| x^{t+1} \|^{2} -2\left\langle x^{t}-x^{t+1}, \boldsymbol{x^{*}} \right\rangle}{\eta_1} - \frac{\| y^{t} \|^{2} - \| y^{t+1} \|^{2}}{\eta_2}
        -
        2\left\langle x^{t+1}, \boldsymbol{b_1} \right\rangle 
        +
        2\left\langle \boldsymbol{x^{*}}, \boldsymbol{b_1} \right\rangle \\
        &+ 
        \left\langle x^{t}-\boldsymbol{x^{*}}, \ Ay^{t}-b_1 \right\rangle - \left\langle x^{t+1}-\boldsymbol{x^{*}}, \ Ay^{t+1}-b_1 \right\rangle.
    \end{align*}
\end{proof}

\section{Proof of Theorem \ref{thm:coord_no_y}}\label{sec:coord_no_y}

\CoordNoY*

\begin{proof}
    By \ref{eqn:AltGD}, \ref{eqn:NEconditions}, and $B = A^{\intercal}$,
    
    \begin{align*}
            y^{t+1} &= y^{t} +\eta_2(A^{\intercal}x^{t+1}-b_2) \\
            x^{t+1} &= x^{t} +\eta_1(Ay^{t}-b_1) \\
            A^{\intercal}\boldsymbol{x^*} &= \boldsymbol{b_2} \\
            A\boldsymbol{y^*} &= \boldsymbol{b_1}
    \end{align*}

    Thus,
    
    \begin{align*}
        -\frac{\| y^{t} \|^{2} - \| y^{t+1} \|^{2}}{\eta_2} = &\frac{\left\langle y^{t+1}-y^{t}, y^{t+1}+y^{t} \right\rangle}{\eta_2} \\
        = &\frac{\left\langle \eta_2(A^{\intercal}x^{t+1}-b_2), 2y^{t}+\eta_2(A^{\intercal}x^{t+1}-b_2) \right\rangle}{\eta_2} \\
        = &\left\langle(A^{\intercal}x^{t+1}-b_2), 2y^{t}\right\rangle + \eta_2\left\langle A^{\intercal}x^{t+1}-b_2, A^{\intercal}x^{t+1}-b_2\right\rangle \\
        = &\left\langle(A^{\intercal}x^{t+1}-b_2), 2y^{t}\right\rangle + \eta_2\|A^{\intercal}x^{t+1}-b_2\|^{2} \\
        = &2\left\langle A^{\intercal}x^{t+1}, y^{t}\right\rangle - 2\left\langle \boldsymbol{b_2}, y^{t}\right\rangle + \eta_2\|A^{\intercal}x^{t+1}-b_2\|^{2} \\
        = &2\left\langle x^{t+1}, Ay^{t}\right\rangle - 2\left\langle A^{\intercal}\boldsymbol{x^*}, y^{t}\right\rangle + \eta_2\|A^{\intercal}x^{t+1}-b_2\|^{2} \\
        = &2\left\langle x^{t+1}, Ay^{t}\right\rangle - 2\left\langle \boldsymbol{x^*}, Ay^{t}\right\rangle + \eta_2\|A^{\intercal}x^{t+1}-b_2\|^{2} \\
        = &2\left\langle x^{t+1}-\boldsymbol{x^*}, Ay^{t}\right\rangle + \eta_2\|A^{\intercal}x^{t+1}-b_2\|^{2} \\
        = &2\left\langle x^{t+1}-\boldsymbol{x^*}, \frac{x^{t+1}-x^{t}}{\eta_1}+\boldsymbol{b_1} \right\rangle + \eta_2\|A^{\intercal}x^{t+1}-b_2\|^{2}.
    \end{align*}
    
    Theorem \ref{thm:coord_no_y} then follows by substituting the above expression for $-\frac{\| y^{t} \|^{2} - \| y^{t+1} \|^{2}}{\eta_2}$ in Theorem \ref{thm:coord_less_yt_nonlinear}.
\end{proof}

\section{Additional Tables for Base Experiments}\label{app:tables}

\noindent\textbf{Zero-sum Games with Small Learning Rates with Standard Methods}

\begin{table}[H]
    \centering
    \resizebox{1\textwidth}{!}{%
    \begin{tabular}{c|cc|c|c|c}
        \toprule
        \textbf{Model}: & \multicolumn{2}{c|}{\textbf{Section \ref{sec:ModelForSimulation}}} & \textbf{Section} \ref{sec:EconomicModel} & \textbf{Section} \ref{sec:FlModel} & \\
        \hline 
        \textbf{Dimension} & \textbf{Avg Rel Err of $\boldsymbol{x^{*}}$} & \textbf{Avg Rel Err in $\boldsymbol{y^*}$} & \textbf{Avg Rel Err in $\boldsymbol{x^{*}}$} & \textbf{Avg Rel Err in $\boldsymbol{x^{*}}$} & $\boldsymbol{\textbf{det} (A_{system_1}})$ \\
        \hline
        \midrule
        1  & 8.97E-14  & 2.81E-13  & 1.29E-11  & 1.20E-13  & 7.32E-1\\
        2  & 2.30E-09  & 1.81E-09  & 4.41E-06  & 1.70E-12  & 1.27E-2\\
        3  & 1.14E-04  & 6.89E-05  & 2.00E-03  & 2.79E-08  & 2.10E-06\\
        4  & 1.85E-01  & 1.55E-01  & 5.42E+00  & 1.61E-05  & 3.06E-15\\
        5  & 1.43E+01  & 2.03E+00  & 2.69E+01  & 7.19E-05  & 3.61E-31\\
        6  & 1.48E+01  & 1.01E+02  & 4.80E+01  & 1.28E+03  & 1.50E-47\\
        7  & 2.12E+01  & 2.48E+01  & 6.36E+01  & 6.32E+00  & 7.87E-63\\
        8  & 4.71E+02  & 1.26E+03  & 1.41E+03  & 3.89E+02  & 8.90E-90\\
        9  & 6.40E+01  & 5.79E+01  & 1.45E+02  & 3.94E+01  & 2.50E-11\\
        10 & 4.40E+01  & 2.88E+01  & 2.15E+02  & 5.23E+01  & 2.20E-14\\
        11 & 7.61E+01  & 6.71E+01  & 4.91E+02  & 7.76E+01  & 3.40E-16\\
        12 & 1.15E+03  & 2.52E+02  & 4.97E+02  & 5.55E+03  & 5.50E-19\\
        13 & 9.60E+01  & 8.79E+01  & 2.06E+02  & 1.16E+04  & 6.90E-21\\
        14 & 1.61E+02  & 1.58E+02  & 4.17E+02  & 4.38E+02  & 2.50E-24\\
        15 & 1.07E+02  & 2.53E+02  & 4.65E+02  & 5.68E+02  & 3.40E-26\\
        16 & 1.85E+02  & 1.70E+02  & 2.98E+02  & 1.51E+03  & 1.60E-39\\
        17 & 1.99E+03  & 1.25E+02  & 1.58E+03  & 5.08E+02  & 1.70E-30\\
        18 & 1.14E+02  & 8.36E+01  & 3.09E+02  & 3.00E+03  & $\approx$ 0.0\\
        19 & 2.98E+03  & 7.91E+02  & 7.99E+02  & 6.94E+02  & $\approx$ 0.0\\
        20 & 1.43E+02  & 1.20E+02  & 3.05E+02  & 6.02E+02  & $\approx$ 0.0\\
        \bottomrule
    \end{tabular}
    }
    \caption{Zero-sum game with small learning rates: average relative error of NE solutions and determinant of linear system matrices in each dimension derived from solving linear systems.}
\end{table}

\noindent\textbf{Zero-sum Games with Large Learning Rates with Standard Methods}

\begin{table}[H]
    \centering
    \resizebox{1\textwidth}{!}{%
    \begin{tabular}{c|cc|c|c|c|c}
        \toprule
        \textbf{Model}: & \multicolumn{2}{c|}{\textbf{Section \ref{sec:ModelForSimulation}}} & \textbf{Section} \ref{sec:EconomicModel} & \textbf{Section} \ref{sec:FlModel} & \\
        \hline 
        \textbf{Dimension} & \textbf{Avg Rel Err of $\boldsymbol{x^{*}}$} & \textbf{Avg Rel Err in $\boldsymbol{y^*}$} & \textbf{Avg Rel Err in $\boldsymbol{x^{*}}$} & \textbf{Avg Rel Err in $\boldsymbol{x^{*}}$} & $\boldsymbol{\textbf{det} (A_{system_1}})$ & $\boldsymbol{\textbf{cond} (A_{system_1}})$ \\
        \hline
        \midrule
        1  & 3.63E$-$11  & 7.61E$-$11  & 4.34E$-$08   & 2.55E$-$15   & 7.92E+02    & 2.72E+04 \\
        2  & 7.79E$-$04  & 1.88E$-$03  & 2.14E+01     & 5.56E$-$12   & 2.62E+09    & 4.64E+09 \\
        3  & 8.99E+08    & 7.43E+09    & 1.83E+13     & 4.97E$-$08   & 1.78E+20    & 1.67E+18 \\
        4  & 7.82E+16    & 4.00E+17    & 7.01E+18     & 1.66E$-$03   & 3.37E+36    & 1.26E+23 \\
        5  & 2.69E+26    & 2.00E+27    & 4.46E+28     & 8.76E+00     & 1.11E+78    & 5.21E+27 \\
        6  & 3.71E+28    & 2.76E+29    & 1.10E+32     & 1.68E+04     & 1.32E+112   & 2.03E+31 \\
        7  & 9.36E+35    & 4.57E+36    & 2.14E+38     & 2.58E+09     & 4.73E+164   & 1.14E+35 \\
        8  & 3.47E+43    & 1.85E+44    & 1.15E+46     & 4.80E+12     & 6.39E+244   & 2.46E+39 \\
        9  & 8.47E+47    & 3.71E+48    & 8.78E+53     & 9.15E+17     & $\infty$    & 9.91E+40 \\
        10 & 5.43E+54    & 8.29E+55    & 8.46E+55     & 2.42E+19     & $\infty$    & 5.05E+44 \\
        11 & 1.58E+62    & 1.56E+62    & 6.59E+65     & 1.02E+29     & $\infty$    & 5.84E+47 \\
        12 & 4.12E+72    & 2.54E+73    & 8.52E+73     & 1.37E+31     & $\infty$    & 1.90E+51 \\
        13 & 9.25E+77    & 7.34E+78    & 6.54E+80     & 2.32E+35     & $\infty$    & 1.09E+51 \\
        14 & 1.85E+81    & 1.87E+82    & 2.99E+84     & 1.76E+37     & $\infty$    & 7.51E+53 \\
        15 & 2.24E+94    & 5.45E+95    & 2.00E+96     & 2.36E+43     & $\infty$    & 2.93E+58 \\
        16 & 5.83E+99    & 4.80E+100   & 2.07E+102    & 1.53E+46     & $\infty$    & 9.63E+58 \\
        17 & 1.05E+109   & 1.90E+110   & 2.94E+108    & 2.03E+50     & $\infty$    & 7.80E+62 \\
        18 & 3.67E+111   & 5.78E+112   & 1.71E+114    & 1.97E+54     & $\infty$    & 1.23E+64 \\
        19 & 5.82E+118   & 8.20E+119   & 2.98E+122    & 7.09E+56     & $\infty$    & 2.24E+65 \\
        20 & 1.62E+128   & 3.01E+129   & 4.63E+132    & 1.29E+62     & $\infty$    & 1.68E+67 \\
        \bottomrule
    \end{tabular}
    }
    \caption{Zero-sum game with large learning rates: average relative error of NE solutions with determinant and condition number of linear system matrices in each dimension derived from solving linear systems.}
\end{table}

\noindent\textbf{Coordination Games with  Standard Methods}
\begin{table}[H]
    \centering
    \resizebox{1\textwidth}{!}{%
    \begin{tabular}{c|cc|c|c|c}
        \toprule
        \textbf{Model:} & \multicolumn{2}{c}{\textbf{Section \ref{sec:ModelForSimulation}}} & \textbf{Section} \ref{sec:EconomicModel} & \textbf{Section} \ref{sec:FlModel} & \\
        \hline
        \textbf{Dimension} & \textbf{Avg Rel Err of $\boldsymbol{x^{*}}$} & \textbf{Avg Rel Err in $\boldsymbol{y^*}$} & \textbf{Avg Rel Err in $\boldsymbol{x^{*}}$} & \textbf{Avg Rel Err in $\boldsymbol{x^{*}}$} & $\boldsymbol{\textbf{det} (A_{system_1}})$ \\
        \hline
        \midrule
        1  & 3.30E-14 & 8.93E-14 & 3.14E-11 & 2.39E-14 & 7.61E-01 \\
        2  & 6.51E-11 & 5.24E-11 & 8.56E-09 & 7.67E-13 & 5.17E-03 \\
        3  & 1.81E-02 & 1.20E-02 & 6.42E+00 & 1.35E-08 & 1.83E-07 \\
        4  & 2.00E-04 & 2.19E-04 & 3.87E-01 & 4.40E-09 & 2.96E-18 \\
        5  & 6.10E+00 & 2.17E+00 & 1.12E+01 & 2.73E-03 & 1.19E-30 \\
        6  & 2.92E+01 & 3.90E+01 & 5.42E+01 & 2.65E-03 & 3.94E-49 \\
        7  & 5.18E+01 & 2.88E+01 & 9.62E+01 & 5.02E+00 & 5.78E-72 \\
        8  & 2.15E+02 & 1.76E+02 & 1.49E+02 & 2.88E+02 & 2.89E-89 \\
        9  & 2.48E+02 & 1.86E+03 & 1.86E+03 & 3.64E+02 & 1.30E-106 \\
        10 & 2.68E+02 & 7.15E+02 & 6.72E+02 & 2.53E+02 & 5.90E-131 \\
        11 & 1.17E+03 & 1.15E+03 & 7.67E+02 & 9.72E+02 & 2.50E-151 \\
        12 & 1.79E+02 & 1.99E+02 & 8.65E+02 & 1.91E+02 & 2.20E-170 \\
        13 & 2.76E+02 & 2.86E+02 & 4.61E+02 & 3.18E+03 & 8.50E-196 \\
        14 & 2.97E+02 & 2.35E+02 & 5.53E+02 & 1.61E+03 & 5.30E-222 \\
        15 & 1.49E+03 & 1.30E+03 & 2.90E+03 & 4.99E+02 & 1.40E-249 \\
        16 & 7.76E+02 & 7.40E+02 & 2.56E+03 & 9.18E+02 & 2.20E-262 \\
        17 & 1.32E+03 & 1.90E+03 & 1.35E+03 & 5.05E+03 & 1.40E-284 \\
        18 & 1.09E+03 & 6.72E+02 & 1.92E+03 & 8.80E+02 & 1.50E-299 \\
        19 & 1.14E+03 & 8.63E+02 & 1.06E+04 & 1.07E+03 & $\approx$ 0 \\
        20 & 9.35E+03 & 5.77E+03 & 3.59E+03 & 1.51E+03 & $\approx$ 0 \\
        \bottomrule
    \end{tabular}
    }
    \caption{Coordination game: average relative error of NE solutions and determinant of linear system matrices in each dimension derived from solving linear systems.}
    \label{tab:coord_avg_rel_errors_rt}
\end{table}

\newpage

\noindent\textbf{Zero-sum with Small Learning Rates with Least Squares}

\begin{table}[H]
    \centering
    \begin{tabular}{c|cc|c|c}
        \toprule
        \textbf{Model:} & \multicolumn{2}{c|}{\textbf{Section \ref{sec:ModelForSimulation}}} & \textbf{Section} \ref{sec:EconomicModel} & \textbf{Section} \ref{sec:FlModel} \\
        \hline
        \textbf{Dimension} & \textbf{Avg Rel Err of $\boldsymbol{x^{*}}$} & \textbf{Avg Rel Err of $\boldsymbol{y^*}$} & \textbf{Avg Rel Err of $\boldsymbol{x^{*}}$} & \textbf{Avg Rel Err of $\boldsymbol{x^{*}}$} \\
        \hline
        \midrule
        1  & 8.85E-14 & 2.79E-13 & 1.29E-11 & 1.20E-13 \\
        2  & 2.30E-09 & 1.81E-09 & 4.41E-06 & 1.71E-12 \\
        3  & 1.10E-04 & 6.91E-05 & 2.00E-03 & 2.79E-08 \\
        4  & 1.88E-02 & 2.09E-02 & 1.37E-01 & 1.61E-05 \\
        5  & 3.06E-01 & 2.59E-01 & 2.79E+00 & 7.18E-05 \\
        6  & 6.55E-01 & 6.34E-01 & 1.11E+00 & 1.27E-01 \\
        7  & 6.12E-01 & 6.33E-01 & 1.34E+00 & 2.68E-01 \\
        8  & 1.08E+00 & 9.46E-01 & 3.26E+00 & 1.28E+01 \\
        9  & 8.17E-01 & 8.40E-01 & 4.73E+00 & 7.43E+00 \\
        10 & 8.56E-01 & 8.69E-01 & 1.87E+00 & 8.83E-01 \\
        11 & 8.68E-01 & 8.50E-01 & 2.56E+00 & 1.22E+00 \\
        12 & 8.86E-01 & 8.38E-01 & 1.38E+01 & 3.12E+01 \\
        13 & 8.67E-01 & 8.91E-01 & 3.21E+00 & 1.19E+01 \\
        14 & 8.17E-01 & 8.09E-01 & 2.06E+00 & 3.30E+00 \\
        15 & 8.96E-01 & 9.05E-01 & 3.25E+00 & 4.30E+00 \\
        16 & 9.31E-01 & 8.62E-01 & 2.20E+01 & 1.96E+01 \\
        17 & 8.94E-01 & 9.26E-01 & 1.88E+01 & 1.80E+01 \\
        18 & 9.35E-01 & 9.33E-01 & 1.49E+01 & 9.08E+00 \\
        19 & 9.57E-01 & 9.97E-01 & 5.28E+01 & 1.30E+01 \\
        20 & 9.20E-01 & 9.09E-01 & 2.62E+00 & 2.79E+00 \\
        \bottomrule
    \end{tabular}
    \caption{Zero-sum game with small learning rate: average relative error of NE solutions via the least squares method in each dimension.}
\end{table}

\noindent\textbf{Zero-sum with Large Learning Rates with Least Squares}

\begin{table}[H]
    \centering
    \begin{tabular}{c|cc|c|c}
        \toprule
        \textbf{Model: }& \multicolumn{2}{c|}{\textbf{Section \ref{sec:ModelForSimulation}}} & \textbf{Section} \ref{sec:EconomicModel} & \textbf{Section} \ref{sec:FlModel} \\
        \hline
        \textbf{Dimension} & \textbf{Avg Rel Err of $\boldsymbol{x^{*}}$} & \textbf{Avg Rel Err of $\boldsymbol{y^*}$} & \textbf{Avg Rel Err of $\boldsymbol{x^{*}}$} & \textbf{Avg Rel Err of $\boldsymbol{x^{*}}$} \\
        \hline
        \midrule
        1  & 3.78E-11 & 7.63E-11 & 4.34E-08 & 2.57E-15 \\
        2  & 7.79E-04 & 1.88E-03 & 2.14E+01 & 5.59E-12 \\
        3  & 2.75E+03 & 6.15E+03 & 6.89E+07 & 4.95E-08 \\
        4  & 1.37E+07 & 6.45E+06 & 6.97E+09 & 1.66E-03 \\
        5  & 7.80E+14 & 1.67E+14 & 1.43E+18 & 4.90E+00 \\
        6  & 2.26E+20 & 5.25E+19 & 4.31E+23 & 8.66E+02 \\
        7  & 3.39E+26 & 7.11E+25 & 3.82E+29 & 4.44E+05 \\
        8  & 5.27E+33 & 3.62E+32 & 1.15E+37 & 7.13E+08 \\
        9  & 1.79E+38 & 1.25E+37 & 3.70E+45 & 2.87E+13 \\
        10 & 1.55E+42 & 3.03E+41 & 8.97E+44 & 3.13E+13 \\
        11 & 9.92E+52 & 2.79E+51 & 3.00E+56 & 2.47E+19 \\
        12 & 1.47E+63 & 9.28E+61 & 3.90E+65 & 1.58E+22 \\
        13 & 3.40E+68 & 1.91E+67 & 5.68E+71 & 3.45E+24 \\
        14 & 4.86E+71 & 4.53E+70 & 1.02E+75 & 1.53E+28 \\
        15 & 1.03E+83 & 1.46E+82 & 1.03E+86 & 1.48E+32 \\
        16 & 7.30E+87 & 6.68E+86 & 1.75E+91 & 1.44E+35 \\
        17 & 8.24E+96 & 7.39E+95 & 1.11E+97 & 1.12E+40 \\
        18 & 1.30E+101 & 1.38E+100 & 4.71E+101 & 6.85E+42 \\
        19 & 1.80E+109 & 2.14E+108 & 7.51E+112 & 3.81E+45 \\
        20 & 3.34E+117 & 2.70E+116 & 1.07E+121 & 1.37E+52 \\
        \bottomrule
    \end{tabular}
    \caption{Zero-sum game with large learning rate: average relative error of NE solutions via the least squares method in each dimension.}
    \label{tab:zero_large_avg_rel_errors_ls}
\end{table}

\newpage
\noindent\textbf{Coordination with Least Squares}

\begin{table}[H]
    \centering
    \begin{tabular}{c|cc|c|c}
        \toprule
        \textbf{Model: }& \multicolumn{2}{c|}{\textbf{Section \ref{sec:ModelForSimulation}}} & \textbf{Section} \ref{sec:EconomicModel} & \textbf{Section} \ref{sec:FlModel} \\
        \hline
        \textbf{Dimension} & \textbf{Avg Rel Err of $\boldsymbol{x^{*}}$} & \textbf{Avg Rel Err of $\boldsymbol{y^*}$} & \textbf{Avg Rel Err of $\boldsymbol{x^{*}}$} & \textbf{Avg Rel Err of $\boldsymbol{x^{*}}$} \\
        \hline
        \midrule        
        1  & 2.97E-14 & 8.82E-14 & 3.14E-11 & 2.39E-14 \\
        2  & 6.49E-11 & 5.24E-11 & 8.56E-09 & 7.74E-13 \\
        3  & 1.81E-02 & 1.20E-02 & 6.42E+00 & 1.35E-08 \\
        4  & 2.26E-04 & 2.44E-04 & 3.68E-02 & 4.41E-09 \\
        5  & 3.23E-01 & 2.90E-01 & 1.54E+00 & 2.73E-03 \\
        6  & 7.74E-01 & 1.03E+00 & 1.30E+00 & 2.65E-03 \\
        7  & 7.00E-01 & 7.83E-01 & 1.13E+00 & 5.25E-02 \\
        8  & 8.42E-01 & 9.28E-01 & 1.26E+00 & 3.88E-01 \\
        9  & 9.92E-01 & 8.68E-01 & 5.08E+00 & 2.08E+01 \\
        10 & 1.14E+00 & 8.90E-01 & 1.80E+00 & 1.26E+01 \\
        11 & 1.18E+00 & 9.62E-01 & 2.93E+00 & 5.99E+00 \\
        12 & 9.63E-01 & 9.47E-01 & 5.36E+00 & 1.29E+01 \\
        13 & 1.02E+00 & 9.41E-01 & 2.70E+00 & 8.99E+00 \\
        14 & 9.74E-01 & 9.44E-01 & 3.11E+00 & 3.95E+00 \\
        15 & 9.75E-01 & 9.91E-01 & 3.47E+00 & 4.93E+00 \\
        16 & 1.01E+00 & 9.51E-01 & 4.58E+00 & 9.18E-01 \\
        17 & 9.86E-01 & 9.70E-01 & 2.37E+00 & 1.13E+00 \\
        18 & 1.03E+00 & 1.02E+00 & 1.40E+01 & 2.89E+01 \\
        19 & 9.80E-01 & 9.74E-01 & 4.88E+00 & 3.15E+00 \\
        20 & 9.87E-01 & 9.80E-01 & 2.23E+00 & 1.11E+00 \\
        \bottomrule
    \end{tabular}
    \caption{Coordination game: average relative error of NE solutions via the least squares method in each dimension.}
    \label{tab:coord_avg_rel_errors_ls}
\end{table}

\noindent\textbf{Zero-sum Games with Small Learning Rates with Tikhonov Regularization}

\begin{table}[H]
    \centering
    \begin{tabular}{c|cc|c|c|c}
        \toprule
        \textbf{Model:} & \multicolumn{2}{c}{\textbf{Section \ref{sec:ModelForSimulation}}} & \textbf{Section} \ref{sec:EconomicModel} & \textbf{Section} \ref{sec:FlModel} & \\
        \hline
        \textbf{Dimension} & \textbf{Avg Rel Err of $\boldsymbol{x^{*}}$} & \textbf{Avg Rel Err in $\boldsymbol{y^*}$} & \textbf{Avg Rel Err in $\boldsymbol{x^{*}}$} & \textbf{Avg Rel Err in $\boldsymbol{x^{*}}$} & $\boldsymbol{\textbf{det} (A_{system_1})}$ \\
        \hline
        \midrule
        1  & 5.07E-05 & 1.46E-04 & 2.00E-03 & 4.81E-04 & 7.32E-01 \\
        2  & 2.38E-01 & 2.21E-01 & 3.92E-01 & 1.37E-01 & 1.28E-02 \\
        3  & 8.09E-01 & 7.58E-01 & 7.38E-01 & 5.52E-01 & 2.10E-07 \\
        4  & 7.38E-01 & 7.48E-01 & 8.00E-01 & 7.16E-01 & 3.06E-15 \\
        5  & 8.61E-01 & 8.40E-01 & 8.79E-01 & 8.43E-01 & 3.61E-31 \\
        6  & 8.77E-01 & 9.10E-01 & 9.00E-01 & 8.58E-01 & 1.50E-47 \\
        7  & 8.47E-01 & 8.35E-01 & 8.83E-01 & 8.65E-01 & 7.87E-63 \\
        8  & 9.41E-01 & 9.27E-01 & 9.34E-01 & 9.17E-01 & 8.90E-90 \\
        9  & 9.43E-01 & 9.36E-01 & 9.65E-01 & 9.59E-01 & 2.48E-106 \\
        10 & 9.45E-01 & 9.69E-01 & 9.50E-01 & 9.60E-01 & 2.15E-136 \\
        11 & 9.73E-01 & 9.40E-01 & 9.61E-01 & 9.66E-01 & 3.39E-156 \\
        12 & 9.48E-01 & 9.53E-01 & 9.56E-01 & 9.59E-01 & 5.51E-185 \\
        13 & 9.43E-01 & 9.67E-01 & 9.71E-01 & 9.73E-01 & 6.94E-209 \\
        14 & 9.43E-01 & 9.26E-01 & 9.56E-01 & 9.55E-01 & 2.46E-235 \\
        15 & 9.64E-01 & 9.69E-01 & 9.71E-01 & 9.78E-01 & 3.39E-259 \\
        16 & 9.32E-01 & 9.56E-01 & 9.43E-01 & 9.57E-01 & 1.62E-288 \\
        17 & 9.46E-01 & 9.66E-01 & 9.63E-01 & 9.66E-01 & 1.71E-302 \\
        18 & 9.73E-01 & 9.60E-01 & 9.86E-01 & 9.84E-01 & $\approx$ 0 \\
        19 & 9.68E-01 & 9.63E-01 & 9.80E-01 & 9.81E-01 & $\approx$ 0 \\
        20 & 9.75E-01 & 9.56E-01 & 9.69E-01 & 9.76E-01 & $\approx$ 0 \\
        \bottomrule
    \end{tabular}
    \caption{Zero-sum game with small learning rates: average relative error of NE solutions and determinant of linear system matrices in each dimension derived via Tikhonov regularization method.}
    \label{tab:zero_avg_rel_errors_non_violation_ridge}
\end{table}

\newpage

\noindent\textbf{Zero-sum Games with Large Learning Rates with Tikhonov Regularization}

\begin{table}[!ht]
    \centering
    \resizebox{1\textwidth}{!}{%
    \begin{tabular}{c|cc|c|c|c|c}
        \toprule
        \textbf{Model}: & \multicolumn{2}{c|}{\textbf{Section \ref{sec:ModelForSimulation}}} & \textbf{Section} \ref{sec:EconomicModel} & \textbf{Section} \ref{sec:FlModel} & \\
        \hline 
        \textbf{Dimension} & \textbf{Avg Rel Err of $\boldsymbol{x^{*}}$} & \textbf{Avg Rel Err in $\boldsymbol{y^*}$} & \textbf{Avg Rel Err in $\boldsymbol{x^{*}}$} & \textbf{Avg Rel Err in $\boldsymbol{x^{*}}$} & $\boldsymbol{\textbf{det} (A_{system_1}})$ & $\boldsymbol{\textbf{cond} (A_{system_1}})$ \\
        \hline
        \midrule
        1  & 1.38E$-$02 & 1.52E$-$02 & 8.85E$-$02 & 6.60E$-$06 & 7.92E+02  & 2.72E+04 \\
        2  & 1.77E$-$01 & 4.48E$-$01 & 5.20E+00   & 1.18E$-$03 & 2.62E+09  & 4.64E+09 \\
        3  & 6.75E+00   & 2.09E+01   & 1.12E+03   & 2.91E$-$02 & 1.78E+20  & 1.67E+18 \\
        4  & 9.27E+04   & 9.61E+05   & 4.78E+06   & 8.86E$-$01 & 3.37E+36  & 1.26E+23 \\
        5  & 8.04E+22   & 2.46E+25   & 8.01E+16   & 5.90E+01   & 1.11E+78  & 5.21E+27 \\
        6  & 1.46E+16   & 4.11E+17   & 2.01E+19   & 2.14E+02   & 1.32E+112 & 2.03E+31 \\
        7  & 9.34E+33   & 2.13E+34   & 3.11E+25   & 6.30E+03   & 4.73E+164 & 1.14E+35 \\
        8  & 2.76E+30   & 7.70E+31   & 4.58E+34   & 5.61E+07   & 6.39E+244 & 2.46E+39 \\
        9  & 1.18E+35   & 8.28E+34   & 6.02E+43   & 6.34E+11   & $\infty$  & 9.91E+40 \\
        10 & 1.05E+41   & 3.26E+41   & 5.62E+44   & 2.81E+12   & $\infty$  & 5.05E+44 \\
        11 & 5.05E+49   & 1.45E+50   & 8.94E+54   & 5.14E+17   & $\infty$  & 5.84E+47 \\
        12 & 1.95E+59   & 3.11E+59   & 6.69E+62   & 1.25E+17   & $\infty$  & 1.90E+51 \\
        13 & 1.72E+65   & 8.01E+66   & 4.63E+68   & 1.94E+23   & $\infty$  & 1.09E+51 \\
        14 & 1.14E+70   & 1.28E+71   & 3.80E+73   & 5.42E+26   & $\infty$  & 7.51E+53 \\
        15 & 2.60E+83   & 8.97E+84   & 1.07E+86   & 2.18E+31   & $\infty$  & 2.93E+58 \\
        16 & 1.69E+87   & 2.77E+88   & 2.01E+90   & 2.58E+48   & $\infty$  & 9.63E+58 \\
        17 & 1.25E+91   & 9.14E+91   & $\infty$   & 3.03E+37   & $\infty$  & 7.80E+62 \\
        18 & $\infty$   & 1.19E+97   & $\infty$   & 3.38E+43   & $\infty$  & 1.23E+64 \\
        19 & 2.36E+92   & 1.25E+93   & 1.57E+92   & 1.97E+45   & $\infty$  & 2.24E+65 \\
        20 & 4.64E+90   & 2.50E+91   & 8.51E+90   & 1.35E+50   & $\infty$  & 1.68E+67 \\
        \bottomrule
    \end{tabular}}
    \caption{Zero-sum game with large learning rates: average relative error of NE solutions with determinant and condition number of linear system matrices in each dimension via Tikhonov regularization method.}
    \label{tab:zero_avg_rel_errors_tikhonov_violation}
\end{table}

\noindent\textbf{Coordination Games with Small Learning Rates with Tikhonov Regularization}

\begin{table}[H]
    \centering
    \resizebox{1\textwidth}{!}{%
    \begin{tabular}{c|cc|c|c|c}
        \toprule
        \textbf{Model:} & \multicolumn{2}{c}{\textbf{Section \ref{sec:ModelForSimulation}}} & \textbf{Section} \ref{sec:EconomicModel} & \textbf{Section} \ref{sec:FlModel} & \\
        \hline
        \textbf{Dimension} & \textbf{Avg Rel Err of $\boldsymbol{x^{*}}$} & \textbf{Avg Rel Err of $\boldsymbol{y^*}$} & \textbf{Avg Rel Err of $\boldsymbol{x^{*}}$} & \textbf{Avg Rel Err of $\boldsymbol{x^{*}}$} & \textbf{det($\boldsymbol{A_{system_1}}$)} \\
        \hline
        \midrule
        1  & 1.22E-02 & 7.35E-03 & 1.30E-01 & 3.81E-03 & 7.61E-01 \\
        2  & 4.64E-01 & 4.52E-01 & 5.88E-01 & 1.40E-01 & 5.17E-03 \\
        3  & 8.12E-01 & 8.16E-01 & 8.53E-01 & 5.66E-01 & 1.83E-07 \\
        4  & 8.14E-01 & 8.08E-01 & 7.60E-01 & 6.25E-01 & 2.96E-18 \\
        5  & 8.98E-01 & 9.36E-01 & 8.31E-01 & 7.03E-01 & 1.19E-30 \\
        6  & 9.03E-01 & 1.02E+00 & 9.38E-01 & 8.33E-01 & 3.94E-49 \\
        7  & 9.34E-01 & 9.72E-01 & 9.55E-01 & 8.81E-01 & 5.78E-72 \\
        8  & 9.67E-01 & 1.00E+00 & 9.65E-01 & 9.09E-01 & 2.89E-89 \\
        9  & 9.86E-01 & 9.48E-01 & 1.01E+00 & 8.91E-01 & 1.32E-106 \\
        10 & 9.88E-01 & 9.83E-01 & 9.74E-01 & 9.37E-01 & 5.93E-131 \\
        11 & 9.88E-01 & 9.83E-01 & 9.91E-01 & 9.55E-01 & 2.51E-151 \\
        12 & 9.79E-01 & 9.80E-01 & 9.95E-01 & 9.76E-01 & 2.18E-170 \\
        13 & 9.85E-01 & 9.73E-01 & 9.92E-01 & 9.61E-01 & 8.53E-196 \\
        14 & 9.92E-01 & 9.75E-01 & 9.92E-01 & 9.72E-01 & 5.32E-222 \\
        15 & 9.91E-01 & 1.00E+00 & 9.96E-01 & 9.65E-01 & 1.42E-249 \\
        16 & 9.99E-01 & 9.82E-01 & 1.01E+00 & 9.86E-01 & 2.16E-262 \\
        17 & 9.91E-01 & 9.90E-01 & 9.93E-01 & 9.73E-01 & 1.41E-284 \\
        18 & 1.03E+00 & 9.69E-01 & 1.01E+00 & 9.87E-01 & 1.46E-299 \\
        19 & 9.87E-01 & 9.83E-01 & 9.88E-01 & 9.80E-01 & $\approx$ 0 \\
        20 & 9.91E-01 & 9.89E-01 & 9.95E-01 & 9.82E-01 & $\approx$ 0 \\
        \bottomrule
    \end{tabular}
    }
    \caption{Coordination game: average relative error of NE solutions and determinant of linear system matrices in each dimension via Tikhonov regularization method.}
    \label{tab:coord_avg_rel_errors_ridge}
\end{table}

\newpage

\section{Additional Tables for Parallelized Experiments}\label{app:tables2}

\noindent\textbf{Zero-sum Games with Small Learning Rates with Parallelization}

\begin{table}[H]
    \centering
    \begin{tabular}{c|cc|c|c|c}
        \toprule
        \textbf{Model:} & \multicolumn{2}{c}{\textbf{Section \ref{sec:ModelForSimulation}}} & \textbf{Section} \ref{sec:EconomicModel} & \textbf{Section} \ref{sec:FlModel} & \\
        \hline
        \textbf{Dimension} & \textbf{Avg Rel Err of $\boldsymbol{x^{*}}$} & \textbf{Avg Rel Err in $\boldsymbol{y^*}$} & \textbf{Avg Rel Err in $\boldsymbol{x^{*}}$} & \textbf{Avg Rel Err in $\boldsymbol{x^{*}}$} & $\boldsymbol{\textbf{det} (A_{par}})$ \\
        \hline
        \midrule
        1  & 2.76E-14 & 5.21E-15 & 2.76E-14 & 2.76E-14 & 8.20E-01 \\
        2  & 6.00E-14 & 5.77E-14 & 6.00E-14 & 6.00E-14 & 4.03E-01 \\
        3  & 1.81E-13 & 3.04E-13 & 1.81E-13 & 1.81E-13 & 5.23E-01 \\
        4  & 3.10E-13 & 2.34E-13 & 3.10E-13 & 3.10E-13 & 3.30E+00 \\
        5  & 3.56E-13 & 3.60E-13 & 3.56E-13 & 3.56E-13 & 4.71E+01 \\
        6  & 6.22E-13 & 7.69E-13 & 6.22E-13 & 6.22E-13 & 1.59E+02 \\
        7  & 1.20E-12 & 1.27E-12 & 1.20E-12 & 1.20E-12 & 1.39E+03 \\
        8  & 2.68E-12 & 2.84E-12 & 2.68E-12 & 2.68E-12 & 2.73E+04 \\
        9  & 3.26E-12 & 2.56E-12 & 3.26E-12 & 3.26E-12 & 8.58E+06 \\
        10 & 1.68E-12 & 1.78E-12 & 1.68E-12 & 1.68E-12 & 3.33E+06 \\
        11 & 7.89E-12 & 7.39E-12 & 7.89E-12 & 7.89E-12 & 1.03E+08 \\
        12 & 4.30E-12 & 3.28E-12 & 4.30E-12 & 4.30E-12 & 2.66E+09 \\
        13 & 2.08E-11 & 2.03E-11 & 2.08E-11 & 2.08E-11 & 3.56E+11 \\
        14 & 1.84E-12 & 5.05E-12 & 1.84E-12 & 1.84E-12 & 1.14E+14 \\
        15 & 7.44E-11 & 1.36E-10 & 7.44E-11 & 7.44E-11 & 7.62E+14 \\
        16 & 7.86E-11 & 1.27E-10 & 7.86E-11 & 7.86E-11 & 6.80E+16 \\
        17 & 4.27E-12 & 3.16E-12 & 4.27E-12 & 4.27E-12 & 6.45E+18 \\
        18 & 5.28E-12 & 5.51E-12 & 5.28E-12 & 5.28E-12 & 1.12E+20 \\
        19 & 8.66E-11 & 4.91E-11 & 8.66E-11 & 8.66E-11 & 1.03E+22 \\
        20 & 3.02E-11 & 2.58E-11 & 3.02E-11 & 3.02E-11 & 7.02E+23 \\
        \bottomrule
    \end{tabular}
    \caption{Zero-sum games with small learning rates: relative errors of NE solutions and average determinants of linear system matrices in each dimension, captured by solving $A_{system}x = b$ in zero-sum games via a parallelization method.}
\end{table}

\begin{table}[H]
\centering
\resizebox{1\textwidth}{!}{%
\begin{tabular}{c|*{5}{>{\raggedleft\arraybackslash}p{2.5cm}}|c}
    \toprule
    & \multicolumn{5}{c|}{\textbf{Tolerance Level for Approximation in Time-Average Convergence (\%)}} & \\ 
    \textbf{Instance} & \textbf{10} & \textbf{5} & \textbf{1} & \textbf{0.5} & \textbf{0.1} & \textbf{Parallelization} \\
    \hline
    \midrule
    1  & 3073 & 3241 & 6842 & 10208 & 34148 & 14 \\
    2  & 454  & 497  & 1786 & 3067  & 10535 & 14 \\
    3  & 1633 & 1678 & 6967 & 8713  & 24466 & 14 \\
    4  & 425  & 569  & 2365 & 3683  & 10042 & 14 \\
    5  & 133  & 548  & 1390 & 2233  & 5164  & 14 \\
    6  & 330  & 357  & 2451 & 3162  & 9516  & 14 \\
    7  & 3085 & 3184 & 3286 & 3301  & 3315  & 14 \\
    8  & 1125 & 1177 & 1223 & 2458  & 3703  & 14 \\
    9  & 607  & 1174 & 1214 & 3669  & 9793  & 14 \\
    10 & 530  & 1883 & 4794 & 4801  & 25972 & 14 \\
    11 & 678  & 695  & 2190 & 3156  & 11430 & 14 \\
    12 & 153  & 163  & 824  & 828   & 3819  & 14 \\
    13 & 730  & 776  & 2425 & 4911  & 20442 & 14 \\
    14 & 447  & 465  & 1681 & 2172  & 6544  & 14 \\
    15 & 11056 & 11145 & 11254 & 56580 & 243743 & 14 \\
    16 & 414  & 865  & 1302 & 3489  & 8309  & 14 \\
    17 & 1405 & 1442 & 2960 & 4459  & 8949  & 14 \\
    18 & 7506 & 7653 & 7960 & 16000 & 24113 & 14 \\
    19 & 664  & 693  & 1448 & 5084  & 7282  & 14 \\
    20 & 2227 & 2304 & 2392 & 2407  & 4827  & 14 \\
    21 & 4586 & 4719 & 4862 & 4884  & 4904  & 14 \\
    22 & 470  & 486  & 1005 & 1508  & 4036  & 14 \\
    23 & 479  & 996  & 1577 & 2616  & 11033 & 14 \\
    24 & 270  & 768  & 1337 & 2386  & 8776  & 14 \\
    25 & 189  & 260  & 1372 & 1638  & 9325  & 14 \\
    26 & 1167 & 1250 & 3842 & 8920  & 46128 & 14 \\
    27 & 717  & 748  & 3100 & 3899  & 12513 & 14 \\
    28 & 379  & 838  & 1252 & 1257  & 2524  & 14 \\
    29 & 10229 & 10537 & 10861 & 10911 & 10954 & 14 \\
    30 & 509  & 532  & 802  & 2418  & 8055  & 14 \\
    \hline
    \midrule
    Average & 1856 & 2055 & 3225 & 6161 & 19812 & 14 \\
    Max     & 11056 & 11145 & 11254 & 56580 & 243743 & 14 \\
    Min     & 133 & 163 & 802 & 828 & 2524 & 14 \\
    \bottomrule
\end{tabular}}
\caption{Time-average convergence vs. model in Section \ref{sec:ModelForSimulation} via parallelization method: the number of iterations needed to approximate NE in dimension 7.}
\label{tab:zero_nonviolation_num_iters_updated_comparison_time_avg}
\end{table}

\begin{table}[H]
\centering
\begin{tabular}{c|*{4}{>{\raggedleft\arraybackslash}p{2cm}}}
    \toprule
    & \multicolumn{4}{c}{\textbf{Rel Err of $x^*$ via Time-Average Convergence}} \\ 
    \textbf{Instance} & \textbf{1 minute} & \textbf{3 minutes} & \textbf{5 minutes} & \textbf{10 minutes}\\
    \hline
    \midrule
    1  & 2.32E-04 & 8.30E-05 & 2.58E-06 & 3.88E-06 \\
    2  & 1.38E-06 & 4.83E-07 & 3.28E-07 & 1.10E-07 \\
    3  & 1.81E-05 &  1.56E-06 &  5.31E-06 & 2.12E-06 \\
    4  & 3.59E-06 & 1.69E-06 & 9.02E-07 & 3.10E-07 \\
    5  & 3.11E-06 & 7.11E-07 & 3.99E-07 & 1.20E-07 \\
    6  & 5.10E-06 &  4.95E-07 &  1.41E-06 & 6.47E-07 \\
    7  & 7.89E-05 & 1.94E-05 & 6.72E-06 & 6.54E-06 \\
    8  & 1.46E-05 & 3.34E-06 &  4.15E-07 &  1.90E-06 \\
    9  & 4.68E-06 &  5.57E-07 &  8.03E-07 & 3.94E-07 \\
    10 & 1.79E-05 & 8.16E-06 & 4.44E-06 & 9.33E-07 \\
    11 & 2.91E-06 & 8.18E-07 & 4.99E-07 & 3.03E-07 \\
    12 & 3.53E-06 & 1.56E-06 & 4.68E-07 & 4.51E-07 \\
    13 & 2.17E-06 & 1.28E-05 & 7.82E-07 & 8.85E-07 \\
    14 & 1.68E-06 & 5.86E-07 & 1.02E-07 & 1.49E-07 \\
    15 & 2.30E-05 & 2.20E-05 & 9.40E-06 & 8.15E-07 \\
    16 & 9.83E-06 & 6.13E-07 &  3.28E-07 &  1.07E-06 \\
    17 & 7.77E-06 & 3.36E-06 & 2.55E-05 & 1.82E-06 \\
    18 & 1.54E-03 & 5.79E-04 &  4.19E-05 &  1.63E-04 \\
    19 & 2.21E-05 & 2.41E-06 & 1.29E-06 & 5.74E-07 \\
    20 & 1.15E-04 & 1.55E-05 & 4.05E-06 & 1.26E-06 \\
    21 & 1.09E-05 & 1.14E-05 &  4.22E-06 &  9.98E-06 \\
    22 & 1.40E-06 & 7.12E-06 & 2.10E-07 & 1.92E-07 \\
    23 & 6.31E-06 & 4.75E-06 & 1.94E-06 & 1.19E-06 \\
    24 & 1.11E-05 & 7.30E-07 & 1.84E-06 & 8.96E-07 \\
    25 & 2.39E-06 & 5.60E-07 & 6.75E-07 & 3.30E-07 \\
    26 & 6.01E-05 & 3.42E-06 &  2.22E-06 &  5.81E-06 \\
    27 & 6.21E-05 & 1.07E-05 & 3.96E-07 & 6.25E-06 \\
    28 & 7.35E-06 & 1.00E-06 & 1.07E-06 & 7.74E-07 \\
    29 & 1.46E-04 & 8.03E-05 & 2.10E-05 & 3.49E-06 \\
    30 & 5.73E-06 & 2.07E-06 &  4.82E-07 &  6.38E-07 \\ 
    \hline
    \midrule
    Average & 8.08E-05 & 2.93E-05 & 4.72E-06 & 7.23E-06 \\ 
    Max     & 1.38E-06 & 4.83E-07 & 1.02E-07 & 1.10E-07 \\ 
    Min     & 1.54E-03 & 5.79E-04	& 4.19E-05 & 1.63E-04 \\ 
    \bottomrule
\end{tabular}
\caption{Zero-sum games with small learning rates: average relative error of time-averaged NE solutions in dimension 7 via \ref{eqn:AltGD} for 1, 3, 5, and 10 minutes.}
\label{tab:time_avg_convergence_dim_7_total}
\end{table}

\newpage

\noindent\textbf{Zero-sum Games with Large Learning Rates with Parallelization}

\begin{table}[!ht]
    \centering
    \resizebox{1\textwidth}{!}{%
    \begin{tabular}{c|cc|c|c|c|c}
        \toprule
        \textbf{Model}: & \multicolumn{2}{c|}{\textbf{Section \ref{sec:ModelForSimulation}}} & \textbf{Section} \ref{sec:EconomicModel} & \textbf{Section} \ref{sec:FlModel} & \\
        \hline 
        \textbf{Dimension} & \textbf{Avg Rel Err of $\boldsymbol{x^{*}}$} & \textbf{Avg Rel Err in $\boldsymbol{y^*}$} & \textbf{Avg Rel Err in $\boldsymbol{x^{*}}$} & \textbf{Avg Rel Err in $\boldsymbol{x^{*}}$} & $\boldsymbol{\textbf{det} (A_{system_1}})$ & $\boldsymbol{\textbf{cond} (A_{system_1}})$ \\
        \hline
        \midrule
        1  & 1.85E$-$14 & 1.89E$-$14 & 1.85E$-$14 & 1.85E$-$14 & 1.87E+03 & 2.85E+01 \\
        2  & 2.84E$-$13 & 7.88E$-$13 & 2.84E$-$13 & 2.84E$-$13 & 2.45E+06 & 1.08E+03 \\
        3  & 3.94E$-$12 & 2.14E$-$12 & 3.94E$-$12 & 3.94E$-$12 & 4.31E+10 & 6.12E+03 \\
        4  & 8.91E$-$12 & 7.14E$-$12 & 8.91E$-$12 & 8.91E$-$12 & 3.01E+15 & 1.69E+04 \\
        5  & 1.14E$-$11 & 9.96E$-$12 & 1.14E$-$11 & 1.14E$-$11 & 9.34E+18 & 1.52E+05 \\
        6  & 1.41E$-$11 & 2.39E$-$11 & 1.41E$-$11 & 1.41E$-$11 & 6.68E+23 & 5.39E+04 \\
        7  & 7.81E$-$11 & 3.18E$-$11 & 7.81E$-$11 & 7.81E$-$11 & 1.04E+29 & 1.36E+05 \\
        8  & 1.05E$-$10 & 1.10E$-$10 & 1.05E$-$10 & 1.05E$-$10 & 1.11E+34 & 4.46E+05 \\
        9  & 1.06E$-$10 & 1.15E$-$10 & 1.06E$-$10 & 1.06E$-$10 & 8.76E+39 & 5.59E+05 \\
        10 & 8.10E$-$11 & 6.18E$-$11 & 8.10E$-$11 & 8.10E$-$11 & 1.70E+46 & 1.19E+06 \\
        11 & 5.19E$-$10 & 5.43E$-$10 & 5.19E$-$10 & 5.19E$-$10 & 6.40E+51 & 3.05E+06 \\
        12 & 3.32E$-$09 & 4.67E$-$09 & 3.32E$-$09 & 3.32E$-$09 & 1.14E+58 & 9.11E+06 \\
        13 & 2.90E$-$10 & 1.61E$-$10 & 2.90E$-$10 & 2.90E$-$10 & 2.23E+63 & 6.72E+05 \\
        14 & 2.92E$-$10 & 3.79E$-$10 & 2.92E$-$10 & 2.92E$-$10 & 9.19E+70 & 5.38E+05 \\
        15 & 9.14E$-$10 & 1.05E$-$09 & 9.14E$-$10 & 9.14E$-$10 & 1.12E+76 & 3.44E+06 \\
        16 & 1.21E$-$09 & 4.16E$-$09 & 1.21E$-$09 & 1.21E$-$09 & 2.44E+82 & 6.55E+06 \\
        17 & 1.72E$-$09 & 5.20E$-$10 & 1.72E$-$09 & 1.72E$-$09 & 1.37E+89 & 1.40E+06 \\
        18 & 5.57E$-$10 & 6.13E$-$10 & 5.57E$-$10 & 5.57E$-$10 & 4.38E+96 & 2.17E+06 \\
        19 & 2.19E$-$08 & 7.16E$-$09 & 2.19E$-$08 & 2.19E$-$08 & 9.72E+101 & 1.03E+07 \\
        20 & 8.33E$-$10 & 1.38E$-$09 & 8.33E$-$10 & 8.33E$-$10 & 6.33E+108 & 4.34E+06 \\
        \bottomrule
    \end{tabular}}
    \caption{Zero-sum game with large learning rates: average relative error of NE solutions with determinant and condition number of linear system matrices in each dimension via parallelization method.}
    \label{tab:zero_avg_rel_errors_parallel_violation}
\end{table}

\begin{table}[H] 
    \centering
    \begin{tabular}{c|cc|c|c}
        \toprule
        \textbf{Model:} & \multicolumn{2}{c}{\textbf{Section \ref{sec:ModelForSimulation}}} & \textbf{Section} \ref{sec:EconomicModel} & \textbf{Section} \ref{sec:FlModel} \\
        \hline
        \textbf{Instance} & \textbf{Avg Rel Err of $\boldsymbol{x^{*}}$} & \textbf{Avg Rel Err in $\boldsymbol{y^*}$} & \textbf{Avg Rel Err in $\boldsymbol{x^{*}}$} & \textbf{Avg Rel Err in $\boldsymbol{x^{*}}$} \\
        \hline
        \midrule
        1  & 1.65E$-$09 & 4.31E$-$09 & 1.65E$-$09 & 1.65E$-$09 \\
        2  & 1.18E$-$09 & 4.45E$-$10 & 1.18E$-$09 & 1.18E$-$09 \\
        3  & 2.06E$-$10 & 8.35E$-$10 & 2.06E$-$10 & 2.06E$-$10 \\
        4  & 7.74E$-$10 & 4.43E$-$10 & 7.74E$-$10 & 7.74E$-$10 \\
        5  & 3.99E$-$10 & 3.23E$-$10 & 3.99E$-$10 & 3.99E$-$10 \\
        6  & 3.53E$-$10 & 3.73E$-$10 & 3.53E$-$10 & 3.53E$-$10 \\
        7  & 3.32E$-$10 & 4.04E$-$10 & 3.32E$-$10 & 3.32E$-$10 \\
        8  & 2.14E$-$10 & 6.73E$-$10 & 2.14E$-$10 & 2.14E$-$10 \\
        9  & 1.47E$-$09 & 1.16E$-$10 & 1.47E$-$09 & 1.47E$-$09 \\
        10 & 4.04E$-$10 & 2.57E$-$10 & 4.04E$-$10 & 4.04E$-$10 \\
        11 & 4.84E$-$10 & 2.25E$-$11 & 4.84E$-$10 & 4.84E$-$10 \\
        12 & 1.12E$-$09 & 1.13E$-$09 & 1.12E$-$09 & 1.12E$-$09 \\
        13 & 2.55E$-$10 & 7.58E$-$10 & 2.55E$-$10 & 2.55E$-$10 \\
        14 & 7.09E$-$10 & 1.11E$-$09 & 7.09E$-$10 & 7.09E$-$10 \\
        15 & 1.16E$-$10 & 1.27E$-$11 & 1.16E$-$10 & 1.16E$-$10 \\
        16 & 2.77E$-$09 & 9.15E$-$10 & 2.77E$-$09 & 2.77E$-$09 \\
        17 & 2.31E$-$10 & 2.87E$-$10 & 2.31E$-$10 & 2.31E$-$10 \\
        18 & 9.02E$-$10 & 2.35E$-$10 & 9.02E$-$10 & 9.02E$-$10 \\
        19 & 3.50E$-$11 & 6.55E$-$11 & 3.50E$-$11 & 3.50E$-$11 \\
        20 & 5.28E$-$10 & 4.80E$-$10 & 5.28E$-$10 & 5.28E$-$10 \\
        21 & 4.33E$-$11 & 4.86E$-$11 & 4.33E$-$11 & 4.33E$-$11 \\
        22 & 8.17E$-$10 & 2.41E$-$09 & 8.17E$-$10 & 8.17E$-$10 \\
        23 & 1.11E$-$10 & 7.28E$-$11 & 1.11E$-$10 & 1.11E$-$10 \\
        24 & 1.04E$-$10 & 8.21E$-$11 & 1.04E$-$10 & 1.04E$-$10 \\
        25 & 4.12E$-$10 & 1.75E$-$09 & 4.12E$-$10 & 4.12E$-$10 \\
        26 & 2.79E$-$10 & 1.23E$-$10 & 2.79E$-$10 & 2.79E$-$10 \\
        27 & 4.60E$-$11 & 7.48E$-$11 & 4.60E$-$11 & 4.60E$-$11 \\
        28 & 1.77E$-$10 & 1.14E$-$10 & 1.77E$-$10 & 1.77E$-$10 \\
        29 & 1.72E$-$10 & 3.65E$-$10 & 1.72E$-$10 & 1.72E$-$10 \\
        30 & 4.30E$-$10 & 1.57E$-$10 & 4.30E$-$10 & 4.30E$-$10 \\
        \bottomrule
    \end{tabular}
    \caption{Zero-sum games with large learning rates: relative errors of NE solutions in dimension 18 across 30 instances, approximated via parallelization method.}
    \label{tab:zero_parallel_violated_rel_errors_dim18}
\end{table}

\begin{table}[H]
    \centering
    \resizebox{1\textwidth}{!}{%
    \begin{tabular}{c|cc|c|c|c}
        \toprule
        \textbf{Model:} & \multicolumn{2}{c}{\textbf{Section \ref{sec:ModelForSimulation}}} & \textbf{Section} \ref{sec:EconomicModel} & \textbf{Section} \ref{sec:FlModel} & \\
        \hline
        \textbf{Dimension} & \textbf{Avg Rel Err of $\boldsymbol{x^{*}}$} & \textbf{Avg Rel Err in $\boldsymbol{y^*}$} & \textbf{Avg Rel Err in $\boldsymbol{x^{*}}$} & \textbf{Avg Rel Err in $\boldsymbol{x^{*}}$} & $\boldsymbol{\textbf{det} (A_{par}})$ \\
        \hline
        \midrule
        1  & 1.49E-14 & 1.72E-14 & 1.49E-14 & 1.27E-14 & 6.83E-01 \\
        2  & 3.80E-11 & 5.25E-11 & 3.80E-11 & 3.96E-13 & 7.67E-01 \\
        3  & 2.50E-12 & 2.60E-12 & 2.50E-12 & 3.28E-12 & 1.79E+00 \\
        4  & 9.22E-14 & 8.66E-14 & 9.22E-14 & 1.65E-11 & 2.86E+00 \\
        5  & 3.58E-13 & 2.69E-13 & 3.58E-13 & 1.92E-09 & 2.29E+01 \\
        6  & 5.13E-12 & 4.22E-12 & 5.13E-12 & 4.08E-11 & 1.87E+02 \\
        7  & 3.14E-13 & 4.02E-13 & 3.14E-13 & 1.60E-10 & 6.22E+02 \\
        8  & 5.44E-11 & 2.73E-11 & 5.44E-11 & 6.04E-10 & 2.07E+04 \\
        9  & 2.80E-12 & 2.09E-12 & 2.80E-12 & 1.08E-09 & 5.05E+05 \\
        10 & 1.62E-12 & 1.22E-12 & 1.62E-12 & 1.27E-08 & 1.00E+07 \\
        11 & 3.74E-12 & 4.29E-12 & 3.74E-12 & 2.29E-09 & 2.30E+08 \\
        12 & 3.45E-11 & 4.17E-11 & 3.45E-11 & 7.20E-09 & 5.59E+09 \\
        13 & 1.08E-11 & 8.62E-12 & 1.08E-11 & 1.70E-09 & 2.65E+11 \\
        14 & 7.19E-12 & 5.06E-12 & 7.19E-12 & 1.02E-08 & 3.01E+12 \\
        15 & 6.78E-12 & 7.27E-12 & 6.78E-12 & 1.02E-09 & 1.02E+15 \\
        16 & 5.61E-12 & 6.56E-12 & 5.61E-12 & 7.20E-07 & 1.37E+17 \\
        17 & 2.00E-11 & 2.11E-11 & 2.00E-11 & 6.45E-09 & 9.65E+17 \\
        18 & 1.03E-11 & 1.88E-11 & 1.03E-11 & 4.73E-08 & 1.05E+20 \\
        19 & 1.32E-11 & 1.93E-11 & 1.32E-11 & 5.22E-08 & 4.68E+22 \\
        20 & 2.18E-11 & 2.12E-11 & 2.18E-11 & 1.19E-08 & 3.82E+23 \\
        \bottomrule
    \end{tabular}}
    \caption{Coordination game: average relative error of NE solution and determinant of linear system matrices in each dimension derived via parallelization method.}
    \label{tab:coord_parallel_avg_rel_errors}
\end{table}

\begin{table}[H] 
    \centering
    \begin{tabular}{c|cc|c|c}
        \toprule
        \textbf{Model:} & \multicolumn{2}{c}{\textbf{Section \ref{sec:ModelForSimulation}}} & \textbf{Section} \ref{sec:EconomicModel} & \textbf{Section} \ref{sec:FlModel} \\
        \hline
        \textbf{Instance} & \textbf{Avg Rel Err of $\boldsymbol{x^{*}}$} & \textbf{Avg Rel Err in $\boldsymbol{y^*}$} & \textbf{Avg Rel Err in $\boldsymbol{x^{*}}$} & \textbf{Avg Rel Err in $\boldsymbol{x^{*}}$} \\
        \hline
        \midrule
        1  & 8.70E-12 & 8.42E-12 & 8.70E-12 & 8.70E-12 \\
        2  & 5.58E-13 & 4.51E-13 & 5.58E-13 & 5.58E-13 \\
        3  & 3.16E-12 & 3.06E-12 & 3.16E-12 & 3.16E-12 \\
        4  & 4.08E-11 & 2.00E-11 & 4.08E-11 & 4.08E-11 \\
        5  & 4.07E-13 & 5.72E-13 & 4.07E-13 & 4.07E-13 \\
        6  & 4.32E-13 & 1.29E-13 & 4.32E-13 & 4.32E-13 \\
        7  & 2.13E-12 & 2.07E-12 & 2.13E-12 & 2.13E-12 \\
        8  & 6.67E-13 & 6.66E-13 & 6.67E-13 & 6.67E-13 \\
        9  & 4.94E-12 & 1.45E-12 & 4.94E-12 & 4.94E-12 \\
        10 & 2.03E-10 & 3.82E-11 & 2.03E-10 & 2.03E-10 \\
        11 & 2.35E-12 & 1.60E-11 & 2.35E-12 & 2.35E-12 \\
        12 & 1.31E-12 & 6.45E-13 & 1.31E-12 & 1.31E-12 \\
        13 & 7.20E-13 & 9.59E-13 & 7.20E-13 & 7.20E-13 \\
        14 & 5.18E-12 & 2.46E-12 & 5.18E-12 & 5.18E-12 \\
        15 & 4.04E-13 & 2.41E-12 & 4.04E-13 & 4.04E-13 \\
        16 & 8.67E-12 & 2.70E-12 & 8.67E-12 & 8.67E-12 \\
        17 & 3.18E-13 & 2.78E-13 & 3.18E-13 & 3.18E-13 \\
        18 & 1.12E-10 & 3.43E-11 & 1.12E-10 & 1.12E-10 \\
        19 & 2.02E-12 & 2.23E-12 & 2.02E-12 & 2.02E-12 \\
        20 & 2.90E-12 & 7.51E-12 & 2.90E-12 & 2.90E-12 \\
        21 & 4.70E-13 & 4.70E-13 & 4.70E-13 & 4.70E-13 \\
        22 & 2.78E-12 & 1.16E-12 & 2.78E-12 & 2.78E-12 \\
        23 & 2.35E-11 & 1.15E-11 & 2.35E-11 & 2.35E-11 \\
        24 & 3.75E-12 & 1.95E-12 & 3.75E-12 & 3.75E-12 \\
        25 & 1.99E-13 & 2.96E-13 & 1.99E-13 & 1.99E-13 \\
        26 & 4.60E-10 & 5.94E-10 & 4.60E-10 & 4.60E-10 \\
        27 & 7.38E-12 & 9.23E-12 & 7.38E-12 & 7.38E-12 \\
        28 & 1.84E-12 & 3.22E-12 & 1.84E-12 & 1.84E-12 \\
        29 & 2.12E-12 & 5.32E-12 & 2.12E-12 & 2.12E-12 \\
        30 & 2.11E-12 & 2.26E-12 & 2.11E-12 & 2.11E-12 \\
        \bottomrule
    \end{tabular}
    \caption{Coordination games: relative errors of NE solutions in dimension 20 across 30 instances, approximated via parallelization method.}
    \label{tab:coord_parallel_rel_errors_dim20}
\end{table}

\end{document}